\documentclass{class/jfm}
\usepackage{graphicx}
\usepackage{epstopdf, epsfig}

\usepackage[utf8]{inputenc}
\usepackage{amsmath,amssymb,mathtools}
\usepackage{colortbl}
\usepackage{appendix}
\usepackage{bm}
\usepackage{subcaption}
\usepackage[rgb]{xcolor}

\newtheorem{prop}{Proposition}

\shorttitle{Scattering by cascades with complex boundary conditions}
\shortauthor{P. J. Baddoo and L. J. Ayton}

\title{Acoustic scattering by cascades with complex boundary conditions: compliance, porosity and impedance}

\author{Peter J. Baddoo
 \and Lorna J. Ayton}

\affiliation{Department of Applied Mathematics and Theoretical Physics, Centre for Mathematical Sciences, University of Cambridge}
\allowdisplaybreaks
\def\fade{0}

\usepackage{tikz,pgfplots,pstool}

\usetikzlibrary{calc, 
				decorations.markings, 
				decorations.text,
				decorations.pathreplacing,
				shapes.multipart,
				positioning,
				arrows.meta,
				arrows,
				matrix,
				patterns,
				shadows,
				hobby}

\tikzset{
    side by side/.style 2 args={
        line width=2pt,
        #1,
        postaction={
            clip,postaction={draw,#2}
        }
    },
    circle node/.style={
        circle,
        draw,
        fill=white,
        minimum size=1.3cm
    }
}

\if\fade1
\usetikzlibrary{fadings}
\tikzfading[name=fade out,
inner color=transparent!0,
outer color=transparent!100]
\fi

\tikzset{test/.style={
    postaction={
        decorate,
        decoration={
            markings,
            mark=at position \pgfdecoratedpathlength-0.5pt with {\arrow[blue,line width=#1] {}; },
            mark=between positions 0 and \pgfdecoratedpathlength-0pt step 0.5pt with {
                \pgfmathsetmacro\myval{multiply(divide(
                    \pgfkeysvalueof{/pgf/decoration/mark info/distance from start}, \pgfdecoratedpathlength),100)};
                \pgfsetfillcolor{white!\myval!black};
                \pgfpathcircle{\pgfpointorigin}{#1};
                \pgfusepath{fill};}
}}}}

\def\myaxis#1#2#3#4#5{
	\begin{scope}[shift={#1}, rotate=#2] 
		\draw [Latex-Latex,thick] (0,#3)  node (yaxis) [above] {#5}
		|- (#3,0) node (xaxis) [right] {#4};
	\end{scope}
}

\def\rTriangle#1#2#3{ %
	\begin{scope}[shift={#1}]
		\node [scale=#2, draw, rotate=45, triangle,  anchor=center, fill = myblue] at (0,0) {};
		\node[below = .1cm] at (0,0) {#3};	    
	\end{scope}
}

\tikzset{
    mark position/.style args={#1(#2)}{
        postaction={
            decorate,
            decoration={
                markings,
                mark=at position #1 with \coordinate (#2);
            }
        }
    }
}

\def\myShiftUp#1{\raisebox{1ex}}
\def\myShiftDown#1{\raisebox{-2.5ex}}

\usetikzlibrary{decorations.text,fit,chains,calc,shapes.geometric,intersections}

\newlength{\myww}
\newlength{\myhh}

\DeclareRobustCommand\newBlueLine{\tikz[baseline=(current bounding box.base)] \draw[ultra thick,blue] (0,2.5pt)--(.75,2.5pt);}
\DeclareRobustCommand\redLine{\tikz[baseline=(current bounding box.base)] \draw[ultra thick,red] (0,2.5pt)--(.75,2.5pt);}

\definecolor{themecolour}{rgb}{0.64, 0.76, 0.68}
\definecolor{mygray}{rgb}{0.7, 0.7, 0.7}
\definecolor{tablegray}{rgb}{0.75, 0.75, 0.75}
\definecolor{confcol}{RGB}{255, 255, 255}
\definecolor{camblue}{RGB}{108, 172, 228}
\definecolor{camred}{RGB}{213, 0, 50}
\definecolor{camnavy}{RGB}{0, 60, 113}
\definecolor{camgreen}{RGB}{114, 180, 49}

\definecolor{myblue}{rgb}{0 ,  0.4470 , 0.7410}
\definecolor{myorange}{rgb}{0.8500,    0.3250,    0.0980}
\definecolor{myyellow}{rgb}{0.9290,    0.6940,    0.1250}
\definecolor{mypurple}{rgb}{ 0.4940,    0.1840,    0.5560}
\definecolor{myred}{rgb}{     0.6350 ,   0.0780 ,   0.1840}
\definecolor{mygreen}{rgb}{         0.4660  ,  0.6740   , 0.1880}

\definecolor{branchCol}{rgb}{0, .8, .8}

\definecolor{prsared}{RGB}{ 219,25,73}
\definecolor{prsablue}{RGB}{ 4,146,210}

\definecolor{confblue}{rgb}{0.75, 0.72, 0.95}
\definecolor{confgreen}{rgb}{0.76, 1, 0.74}
\definecolor{confgrey}{rgb}{0.9, .9, 0.9}

\definecolor{cvblue}{rgb}{0.22,0.45,0.70}%

\definecolor{matlab1}{rgb}{0,0.4470,0.7410}
\definecolor{matlab2}{rgb}{0.8500,0.3250,0.0980}
\definecolor{matlab3}{rgb}{0.9290,0.6940,0.0980}
\definecolor{matlab4}{rgb}{0.4940,0.1840,0.5560}
\definecolor{matlab5}{rgb}{0.4660,0.6740,0.1880}
\definecolor{matlab6}{rgb}{0.3010,0.7450,0.9330}
\definecolor{matlab7}{rgb}{0.6350,0.0780,0.1840}

\usetikzlibrary{external}
\graphicspath{{images/}}

\newlength{\fheight}
\newlength{\fwidth}
\pgfplotsset{compat=newest}

\usepackage{MnSymbol}%
\usepackage{wasysym}%
\begin{document}

\maketitle

\begin{abstract}
We present a solution for the scattered field caused by an incident wave interacting with an infinite cascade of blades with complex boundary conditions. This extends previous studies by allowing the blades to be compliant, porous or satisfying a generalised impedance condition. Beginning with the convected wave equation, we employ Fourier transforms to obtain an integral equation amenable to the Wiener--Hopf method. This Wiener--Hopf system is solved using a method that avoids the factorisation of matrix functions. The Fourier transform is inverted to obtain an expression for the acoustic potential function that is valid throughout the entire domain. We observe that the principal effect of complex boundary conditions is to perturb the zeros of the Wiener--Hopf kernel, which correspond to the duct modes in the inter-blade region. We focus efforts on understanding the role of porosity, and present a range of results on sound transmission and generation. The behaviour of the duct modes is discussed in detail in order to explain the physical mechanisms behind the associated noise reductions. In particular, we show that cut-on duct modes do not exist for arbitrary porosity coefficients. Conversely, the acoustic modes are unchanged by modifications to the boundary conditions. Consequently, we observe that even modest values of porosity can result in reductions in the sound power level of $5$ dB for the first mode and $20$ dB for the second mode. The solution is essentially analytic (the only numerical requirements are matrix inversion and root finding) and is therefore extremely rapid to compute.
\end{abstract}


\def\i{\textrm{i}}
\def\e{\textrm{e}}
\def\d{\textrm{d}}
\def\Im{\mathop{\mathbb{I}\textrm{m}}}
\def\Re{\mathop{\mathbb{R}\textrm{e}}}
\def\sgn{\textrm{sgn}}
\def\Xint#1{\mathchoice
	{\XXint\displaystyle\textstyle{#1}}%
	{\XXint\textstyle\scriptstyle{#1}}%
	{\XXint\scriptstyle\scriptscriptstyle{#1}}%
	{\XXint\scriptscriptstyle\scriptscriptstyle{#1}}%
	\!\int}
\def\XXint#1#2#3{{\setbox0=\hbox{$#1{#2#3}{\int}$}
		\vcenter{\hbox{$#2#3$}}\kern-.5\wd0}}
\def\ddashint{\Xint=}
\def\dashint{\Xint-}
\def\sourceloc{\mathbf{x}_0}
\def\sourcethet{\theta_0}
\def\incidentG{\hat{G}_i}
\def\reflectedG{\hat{G}_r}
\def\scatteredG{\hat{G}_s}
\def\FTG{\hat{\mathcal{G}}_s}
\def\FTGm{\hat{\mathcal{G}}_{s-}}
\def\FTGp{\hat{\mathcal{G}}_{s+}}
\def\FTvpm{\mathcal{V}^\pm_-}
\def\FTv{\mathcal{V}_-}
\def\Jmap{\mathcal{J}}

\newcommand{\amp}{A}
\newcommand{\namp}{w_0}

\newcommand{\ssig}{h_{1,\Sigma}}
\newcommand{\sinc}{h_{i}}

\newcommand{\dsou}{D_{1,S}}
\newcommand{\dsoup}{D_{1,S,+}}
\newcommand{\dsoupm}{D_{1,S,\pm}}
\newcommand{\dsoum}{D_{1,S,-}}

\newcommand{\ddel}{D_{1,\Delta}}
\newcommand{\ddelp}{D_{1,\Delta,+}}
\newcommand{\ddelpm}{D_{1,\Delta,\pm}}
\newcommand{\ddelm}{D_{1,\Delta,-}}

\newcommand{\dsig}{D_{1,\Sigma}}
\newcommand{\dsigp}{D_{1,\Sigma,+}}
\newcommand{\dsigpm}{D_{1,\Sigma,\pm}}
\newcommand{\dsigm}{D_{1,\Sigma,-}}

\newcommand{\dgen}{D}

\newcommand{\Agam}{\mathcal{A}_{1,\Gamma,n}}
\newcommand{\Asig}{\mathcal{A}_{1,\Sigma,n}}
\newcommand{\Asou}{\mathcal{A}_{1,S,n}}
\newcommand{\Adel}{\mathcal{A}_{1,\Delta,n}}
\newcommand{\Agen}{\mathcal{A}_{1,n}}
\newcommand{\Azero}{\mathcal{A}_{0,n}}

\newcommand{\Bgam}{\mathcal{B}_{1,\Gamma,n}}
\newcommand{\Bsig}{\mathcal{B}_{1,\Sigma,n}}
\newcommand{\BsigR}{\mathcal{B}_{1,\Sigma,k}}
\newcommand{\Bsou}{\mathcal{B}_{1,S,n}}
\newcommand{\Bdel}{\mathcal{B}_{1,\Delta,n}}
\newcommand{\Bgen}{\mathcal{B}_{1,n}}
\newcommand{\Bzero}{\mathcal{B}_{0,n}}

\newcommand{\Cgam}{\mathcal{C}_{1,\Gamma,n}}
\newcommand{\Csig}{\mathcal{C}_{1,\Sigma,n}}
\newcommand{\Csou}{\mathcal{C}_{1,S,n}}
\newcommand{\Cdel}{\mathcal{C}_{1,\Delta,n}}
\newcommand{\Cgen}{\mathcal{C}_{1,n}}
\newcommand{\Czero}{\mathcal{C}_{0,n}}

\newcommand{\Tgam}{\mathcal{T}_{1,\Gamma,l}}
\newcommand{\Tsig}{\mathcal{T}_{1,\Sigma,l}}
\newcommand{\Tsou}{\mathcal{T}_{1,S,l}}
\newcommand{\Tdel}{\mathcal{T}_{1,\Delta,l}}
\newcommand{\Tgen}{\mathcal{T}_{1,l}}
\newcommand{\TgenV}[1]{\mathcal{T}_{1,#1}}

\newcommand{\Sgam}{\mathcal{S}_{1,\Gamma,l}}
\newcommand{\Ssig}{\mathcal{S}_{1,\Sigma,l}}
\newcommand{\Ssou}{\mathcal{S}_{1,S,l}}
\newcommand{\Sdel}{\mathcal{S}_{1,\Delta,l}}
\newcommand{\Sgen}{\mathcal{S}_{1,l}}

\newcommand{\Rgam}{\mathcal{R}_{1,\Gamma,m}}
\newcommand{\Rsou}{\mathcal{R}_{1,S,m}}
\newcommand{\Rdel}{\mathcal{R}_{1,\Delta,m}}
\newcommand{\Rgen}{\mathcal{R}_{1,m}}

\newcommand{\dgamres}{D_{r,\Gamma,k}}
\newcommand{\dsoures}{D_{r,S,k}}
\newcommand{\dsigres}{D_{r,\Sigma,k}}
\newcommand{\ddelres}{D_{r,\Delta,k}}

\newcommand{\Ugam}{\mathcal{U}_{1,\Gamma,m}}
\newcommand{\Usou}{\mathcal{U}_{1,S,m}}
\newcommand{\Udel}{\mathcal{U}_{1,\Delta,m}}
\newcommand{\Ugen}{\mathcal{U}_{1,m}}

\newcommand{\transformedSourceTerms}{S}
\newcommand{\phiTransSou}{S_\phi}
\newcommand{\souInvKernel}[1]{\mathfrak{S}_{#1}}
\newcommand{\souInvKernelResPM}[2]{\mathfrak{S}^{\pm r}_{#1,#2}}
\newcommand{\souInvKernelResP}[2]{\mathfrak{S}^{+r}_{#1,#2}}
\newcommand{\souInvKernelResM}[2]{\mathfrak{S}^{-r}_{#1,#2}}
\newcommand{\souIBR}{S^i}
\newcommand{\souIBRA}{S^{i,(1)}}
\newcommand{\souIBRB}{S^{i,(2)}}
\newcommand{\souU}{S^{u}}
\newcommand{\souD}{S^{d}}
\newcommand{\souUR}[1]{S^{u,r}_{#1}}
\newcommand{\souURA}[1]{S^{u,r,(1)}_{#1}}
\newcommand{\souDR}[1]{S^{d,r}_{#1}}
\newcommand{\souDRB}[1]{S^{d,r,(2)}_{#1}}
\newcommand{\souDRBs}[1]{S^{*d,r,(2)}_{#1}}

\newcommand{\sourceTerms}{\mathbb{S}}
\newcommand{\souCoef}[1]{\mathbb{S}_{#1}}
\newcommand{\souIntCoefs}{\mathbb{S}^i}

\newcommand{\souModeP}[1]{M^+_{#1}}
\newcommand{\souModeM}[1]{M^-_{#1}}
\newcommand{\souModeVarP}[1]{\tilde{M}^+_{#1}}
\newcommand{\souModeVarM}[1]{\tilde{M}^-_{#1}}
\newcommand{\souModeVarPM}[1]{\tilde{M}^\pm_{#1}}

\newcommand{\souTransFac}[1]{\mathcal{F}_{#1}}
\newcommand{\souTransFacResPM}[1]{\mathcal{F}_{#1}^{r,\pm}}
\newcommand{\souTransFacResP}[1]{\mathcal{F}_{#1}^{r,+}}
\newcommand{\souTransFacResM}[1]{\mathcal{F}_{#1}^{r,-}}

\newcommand{\souA}{S^{(1)}}
\newcommand{\souB}{S^{(2)}}
\newcommand{\souBs}{S^{*(2)}}

\newcommand{\Pgen}{{P}}
\newcommand{\Pgam}{{P}_{1,\Gamma}}
\newcommand{\Psou}{{P}_{1,S}}
\newcommand{\Pdel}{{P}_{1,\Delta}}
\newcommand{\Psig}{{P}_{1,\Sigma}}
\newcommand{\PgenE}{P_1}
\newcommand{\PO}{P_0}

\newcommand{\fgam}{{f}_{1,\Gamma}}
\newcommand{\fsou}{{f}_{1,S}}
\newcommand{\fdel}{{f}_{1,\Delta}}
\newcommand{\fsig}{{f}_{1,\Sigma}}

\newcommand{\FgamP}{{F}_{1,\Gamma,+}}
\newcommand{\FsouP}{{F}_{1,S+}}
\newcommand{\FdelP}{{F}_{1,\Delta,+}}
\newcommand{\FsigP}{{F}_{1,\Sigma,+}}

\newcommand{\FgamM}{{F}_{1,\Gamma,-}}
\newcommand{\FsouM}{{F}_{1,S}}
\newcommand{\FdelM}{{F}_{1,\Delta,-}}
\newcommand{\FsigM}{{F}_{1,\Sigma,-}}

\newcommand{\Vgam}{\mathcal{V}_{1,\Gamma}}
\newcommand{\gamEntire}{E_\Gamma^+}
\newcommand{\gamRes}{Res_{\Gamma,n}}
\newcommand{\gamResDMAM}{Res_{\Gamma,}}
\newcommand{\gamResG}{Res_{\gModeO,\Gamma,n}}
\newcommand{\gamFun}{L_{\Gamma}}

\def\mySec{section }
\def\sec{section }

\newcommand{\sumCoefs}[1]{\mathcal{c}_{\Sigma,#1}}
\newcommand{\jumpCoefs}[1]{\mathcal{c}_{\Delta,#1}}

\newcommand{\aModeP}{\lambda_m^+}
\newcommand{\aModeM}{\lambda_m^-}
\newcommand{\aModePM}{\lambda_m^\pm}
\newcommand{\aModeMP}{\lambda_m^\pm}

\newcommand{\zModeP}{\zeta_m^+}
\newcommand{\zModeM}{\zeta_m^-}
\newcommand{\zModePM}{\zeta_m^\pm}
\newcommand{\zModeMP}{\zeta_m^\pm}

\newcommand{\aModePv}[1]{\lambda_{#1}^+}
\newcommand{\aModeMv}[1]{\lambda_{#1}^-}
\newcommand{\aModePMv}[1]{\lambda_{#1}^\pm}
\newcommand{\aModeMPv}[1]{\lambda_{#1}^\mp}
\newcommand{\zModePv}[1]{\zeta_{#1}^+}
\newcommand{\zModeMv}[1]{\zeta_{#1}^-}
\newcommand{\zModePMv}[1]{\zeta_{#1}^\pm}
\newcommand{\zModeMPv}[1]{\zeta_{#1}^\mp}

\newcommand{\kfun}{k}
\newcommand{\kResM}{k_{r,m}^{-}}
\newcommand{\kResP}{k_{r,m}^{+}}

\newcommand{\gModel}{\kappa_l^-}
\newcommand{\gModelA}{\kappa_{l_1}^-}
\newcommand{\gModeV}[1]{\kappa_{#1}^-}
\newcommand{\gModeO}{\kappa_0^-}
\newcommand{\zgMode}[1]{\zeta_{\kappa^-_{#1}}}

\newcommand{\DzeroCoefs}[1]{\mathcal{D}_{0,{#1}}}
\newcommand{\DzeroG}{\mathcal{D}_{0,\gModeO}}
\newcommand{\DzeroCoefsPr}[1]{\mathcal{D}_{0,{#1}}^\prime}
\newcommand{\DzeroGPr}{\mathcal{D}_{0,\gModeO}^\prime}
\newcommand{\DzeroGPrPMG}{\mathcal{D}_{0,\gModeO}^{\prime\pm \Gamma}}

\newcommand{\dModeM}{\theta_n^-}
\newcommand{\dModeP}{\theta_n^+}
\newcommand{\dModePr}{\theta_k^+}
\newcommand{\dModePv}[1]{\theta_{#1}^+}
\newcommand{\dModeMv}[1]{\theta_{#1}^-}
\newcommand{\dModeMr}{\theta_m^+}
\newcommand{\dModePM}{\theta_n^\pm}
\newcommand{\dModeMP}{\theta_n^\mp}

\newcommand{\LUHP}{\mathcal{L}^\mp}
\newcommand{\ULHP}{\mathcal{L}^{\pm}} %
\newcommand{\UHP}{\mathcal{L}^{+}} %
\newcommand{\LHP}{\mathcal{L}^{-}} %
\newcommand{\porousRootsI}[1]{\theta_{#1,R}^+}
\newcommand{\porousRootsII}[1]{\theta_{#1,L}^+}
\newcommand{\porousRootsIM}[1]{\theta_{#1,R}^-}
\newcommand{\porousRootsIIM}[1]{\theta_{#1,L}^-}

\newcommand{\porousRRootsI}[1]{\theta_{#1,R}^{r+}}
\newcommand{\porousIRootsI}[1]{\theta_{#1,R}^{i+}}
\newcommand{\asymPorousRootsCoefsI}[1]{a_{R}^{(#1)}} 
\newcommand{\asymPorousRootsCoefsII}[1]{a_{L}^{(#1)}} 
\newcommand{\asymPorousRootsFunI}[2]{\theta_R^{(#1)}(#2)}
\newcommand{\asymPorousRootsFunII}[2]{\theta_L^{(#1)}(#2)}
\newcommand{\asymPorousAcouP}[2]{\lambda^{(#1)+}_{#2}}
\def\i{\textrm{i}}
\def\e{\textrm{e}}

\def\Gdel{G_{\Delta}}
\def\GdelA{G_{\Delta}^{(1)}}
\def\GdelB{G_{\Delta}^{(2)}}
\def\GdelBs{G_{\Delta}^{*(2)}}

\newcommand{\invKernA}{A}
\newcommand{\invKernB}{B}

\newcommand{\vd}{s^\ast}
\newcommand{\hd}{d^\ast}

\newcommand{\vPS}{s_{\phi}}
\newcommand{\hPS}{d_{\phi}}
\newcommand{\vND}{s}
\newcommand{\hND}{d}
\newcommand{\densityInfD}{\rho^\ast}

\newcommand{\AaResP}[1]{A^{+,r}_{a,#1}}
\newcommand{\AbResP}[1]{A^{+,r}_{b,#1}}
\newcommand{\AaResM}[1]{A^{-,r}_{a,#1}}
\newcommand{\AbResM}[1]{A^{-,r}_{b,#1}}
\newcommand{\AdModeaP}[1]{A_{a,\dModePv{#1}}}
\newcommand{\AdModeaM}[1]{A_{a,\dModeMv{#1}}}
\newcommand{\AdModebP}[1]{A_{b,\dModePv{#1}}}
\newcommand{\AdModebM}[1]{A_{b,\dModeMv{#1}}}
\newcommand{\AgModea}[1]{A_{b,\gModeV{#1}}}
\newcommand{\AgModeb}[1]{A_{b,\gModeV{#1}}}
\newcommand{\BgMode}[1]{B_{\gModeV{#1}}}

\def\bPG{\beta_\infty}

\def\bd{b^\ast}

\newcommand{\xd}{x^\ast}
\newcommand{\yd}{y^\ast}

\newcommand{\zPG}{z_{\beta}}
\newcommand{\zPS}{z_{\phi}}

\newcommand{\Ud}{U_{\infty}^\ast}
\newcommand{\D}{\textnormal{D}}
\newcommand{\At}{A_t}
\newcommand{\Ats}{\tilde{A_t}}
\newcommand{\As}{A_s} 
\newcommand{\ks}{k_s}
\newcommand{\An}{A_n}
\newcommand{\kn}{k_n}

\def\Parallelogram{\mathcal{P}}

\newcommand{\bladeDist}{\Delta_\phi}
\newcommand{\bladeDistO}{\Delta}
\newcommand{\phic}{c_\phi}  %

\newcommand{\caseHallA}{A}
\newcommand{\caseHallB}{B}
\newcommand{\caseElhadidi}{B}
\newcommand{\caseVerdon}{A}

\def\ra{(I)}
\def\rb{(II)} 
\def\rc{(III)}
\def\rd{(IV)}
\def\re{(V)}

\def\iBP{\sigma}
\def\miBP{\sigma^\prime}

\definecolor{mygray}{rgb}{.7, .7, .7}
\definecolor{myblue}{rgb}{0 ,  0.4470 , 0.7410}
\definecolor{myorange}{rgb}{0.8500,    0.3250,    0.0980}
\definecolor{myyellow}{rgb}{0.9290,    0.6940,    0.1250}
\definecolor{mypurple}{rgb}{ 0.4940,    0.1840,    0.5560}
\definecolor{myred}{rgb}{     0.6350 ,   0.0780 ,   0.1840}
\definecolor{mygreen}{rgb}{         0.4660  ,  0.6740   , 0.1880}

\newcommand*{\TikzDist}{3}%
\newcommand*{\TikzStag}{25}%
\newcommand*{\TikzDistb}{2}%
\newcommand*{\TikzStagb}{45}%
\newcommand*{\StagArcRad}{1}
\newcommand*{\StagArcRadb}{.5}
\newcommand*{\TikzAoA}{30}%
\newcommand*{\AoAArcRad}{2}
\newcommand*{\TikzIP}{-35}%
\newcommand*{\AScale}{5}%
\newcommand*{\bh}{.1} %

\newcommand*{\InvaModeCoefsU}[1]{\mathcal{H}^+_{#1}}
\newcommand*{\InvSouCoefsU}[1]{\mathcal{S}^+_{#1}}
\newcommand*{\InvModeCoefsUA}[1]{\mathcal{H}^+_{a,#1}}
\newcommand*{\InvModeCoefsDB}[1]{\mathcal{H}^-_{b,#1}}
\newcommand*{\InvModeCoefsRL}[1]{\mathcal{H}^-_{\mathcal{R},a,#1}}
\newcommand*{\InvModeCoefsUD}[1]{\mathcal{H}^+_{\mathcal{U},b,#1}}
\newcommand*{\InvaModeCoefsD}[1]{\mathcal{H}^-_{#1}}
\newcommand*{\InvSouCoefsD}[1]{\mathcal{S}^-_{#1}}
\newcommand*{\InvGamCoefsAM}[1]{\mathcal{H}^-_{\Gamma,\aModeMv{},#1}}
\newcommand*{\InvGamCoefsGM}[1]{\mathcal{H}^-_{\Gamma,\gModeO}}
\newcommand*{\InvModeCoefsGUA}[1]{\mathcal{H}^+_{G,a,#1}}
\newcommand*{\InvModeCoefsGLA}[1]{\mathcal{H}^-_{G,a,#1}}
\newcommand*{\InvModeCoefsGUB}[1]{\mathcal{H}^+_{G,b,#1}}
\newcommand*{\InvModeCoefsGLB}[1]{\mathcal{H}^-_{G,b,#1}}
\newcommand*{\InvPresLT}{\mathcal{H}^-_{P,b}}

\newcommand*{\InvSouCoefsI}[1]{\mathcal{S}^i_{#1}}
\newcommand*{\InvDelA}[1]{\mathcal{H}_{\Delta,a,#1}}
\newcommand*{\InvDelB}[1]{\mathcal{H}_{\Delta,b,#1}}
\newcommand*{\InvThetaPlusA}[1]{\mathcal{H}_{\theta^+,a,#1}}
\newcommand*{\InvThetaPlusB}[1]{\mathcal{H}_{\theta^+,b,#1}}
\newcommand*{\InvThetaMinusA}[1]{\mathcal{H}_{\theta^-,a,#1}}
\newcommand*{\InvThetaMinusB}[1]{\mathcal{H}_{\theta^-,b,#1}}
\newcommand*{\InvBCA}[1]{\mathcal{H}_{\Sigma,a,#1}}
\newcommand*{\InvBCB}[1]{\mathcal{H}_{\Sigma,b,#1}}

\newcommand*{\InvSouIA}[1]{\mathcal{H}_{S,a,#1}}
\newcommand*{\InvSouIB}[1]{\mathcal{H}_{S,b,#1}}
\newcommand*{\InvSouIC}[1]{\mathcal{H}_{S,c,#1}}
\newcommand*{\InvSouID}[1]{\mathcal{H}_{S,d,#1}}

\newcommand{\invKerSou}{I_{S}}
\newcommand{\invKerD}{I_{D}}
\newcommand{\invKerG}{I_{G}}

\newcommand*{\evenodd}[1]{\chi_{#1}}

\newcommand{\signGam}{{{\pm}_\Gamma}}

\renewcommand{\ssig}{\phi}
\renewcommand{\dsig}{D}
\renewcommand{\Psig}{P}
\renewcommand{\fsig}{f}
\renewcommand{\dsigres}{D_{r,k}}
\renewcommand{\Asig}{\mathcal{A}_n}
\renewcommand{\FsigM}{F_-}
\renewcommand{\FsigP}{F_+}
\renewcommand{\dsigp}{D_+}
\renewcommand{\dsigm}{D_-}
\renewcommand{\Bsig}{\mathcal{B}_n}
\renewcommand{\Csig}{\mathcal{C}_n}
\renewcommand{\gModeO}{\kappa}
\section{Introduction}
\noindent Turbomachinery noise remains a significant contributor to overall aero-engine noise \citep{Peake2012}. A considerable source of broadband noise is the so-called ``rotor-stator interaction'', where unsteady wakes shed by compressor rotors interact with downstream stators. Much progress has been made in understanding rotor-stator interaction noise when the blades are modelled as flat plates \citep{Peake1992,Glegg1999,Posson2010,Bouley2017} or with realistic geometries \citep{Evers2002,Baddoo2019b}. However, research into blades with complex boundaries -- where the blades are not necessarily stationary and rigid -- is far less developed, especially from an analytic standpoint. This paper presents the first analytic treatment of scattering by cascades with complex boundary conditions, including compliance, porosity and impedance.

Aspirations for lighter and more efficient engines have driven the design of thinner and lighter blades in turbomachinery \citep{Saiz2008}. As a result, aeroelastic effects such as flutter and resonance must be considered in modern turbomachinery design and testing. The rapid and accurate prediction of the aeroacoustic performance of turbomachinery with consideration of aeroelastic effects is therefore essential in evaluating the appropriateness of potential blade designs. Analytic solutions are excellent candidates for this task \citep{Glegg1999,Posson2010,Peake1992}, but are presently limited to rigid blades with no consideration of aeroelastic effect. The present study permits the consideration of compliant blades \citep{Crighton1970}, where blade deforms with a (purely) local response to the pressure gradient across the blade and elastic restoring forces are ignored. This is particularly relevant to marine applications where inertial effects dominate elastic restoring forces and is an important first step towards a more general linearised elastic blade \citep{Cavalieri2016}.

An influential trend in aeroacoustic research is to modify aerofoils with noise reducing technologies. A popular choice is a poroelastic extension, which was originally inspired by the silent flight of owls \citep{Graham1934}, and the corresponding experimental support \cite{Geyer2010}. A detailed review of research into the silent flight of owls is available in \cite{Jaworski2020}. Poroelastic extensions have been applied to semi-infinite \citep{Jaworski2013} and finite \citep{Ayton2016a} blades and have demonstrated considerable noise reductions. Structural requirements limit the application of highly porous blades in turbomachinery, but there remains the possibility that significant noise reductions are available for modest porosity values. The approach of the present research permits porous blades through the assumption of a Darcy-type condition where the seepage velocity through the blade is proportional to the pressure jump across the blade. We will also consider the more general situation of an impedance relation along the blades \citep{Myers1980}.

In this paper we extend the analyses of \cite{Glegg1999} and \cite{Posson2010} to analyse the scattering by a cascade of blades with a range of boundary conditions. The problem is solved with tools from complex variable theory, including the Wiener--Hopf method. Taking a Fourier transform maps the problem into the spectral plane where the Wiener--Hopf analysis is carried out in a similar way to \cite{Glegg1999}. An inverse Fourier transform is applied to return the problem to physical space, and contour integration is applied to recover the acoustic field \citep{Posson2010}. A significant advantage of the presented technique is that the method is identical regardless of the boundary condition -- the only effects of modifying the boundary condition are to modify the kernel in the Wiener--Hopf method.

A striking feature of the analysis is that modifications to the boundary conditions do not affect the modal structure of the solution in the far field. In the spectral plane, the only effect of modifying the boundary condition is to vary the locations of the zeros of the  Wiener--Hopf kernel. This has a significant effect on the acoustic field in the inter-blade region since the cut-on frequencies of the duct modes are modified to account for energy being absorbed or produced by the blades. In fact, we will show that, in the case where the blades are porous, the duct modes are never cut-on. The poles of the Wiener--Hopf kernel correspond to the acoustic modes scattered into the far field and are invariant under modifications to the boundary conditions. Consequently, the cut-on frequencies of the acoustic modes are unchanged and the model structure of the upstream and downstream acoustic fields are the same regardless of boundary condition, although the coefficients of these modes do change, thus permitting a far-field noise reduction.

We consider four possible boundary conditions labelled cases 0---III. Physically, case 0 corresponds to rigid blades; case I, porous or compliant blades with no background flow; case II, porous blades with background flow; and case III, a general impedance relation. Mathematically, case 0 corresponds to a Neumann boundary condition; case I, a Robin boundary condition; case II, an oblique derivative boundary condition; and case III, a generalised Cauchy boundary condition.

We begin by presenting a mathematical model for the blade row in section \ref{c5:Sec:MathForm}, including the modelling of the various boundary conditions. We then present some details of the mathematical solution in section \ref{c5:Sec:MathSol}. In section \ref{c5:Sec:Results} we conduct a detailed investigation of the role of blade porosity. In particular, we present a range of results on sound generation and sound transmission. Finally, in section \ref{c5:Sec:Conclusion} we summarise our work and suggest future directions of research.

\section{Mathematical Formulation} 
\label{c5:Sec:MathForm}
We consider a rectilinear cascade of blades in a uniform, subsonic flow as illustrated in figure \ref{Fig:poroCascade}. As is typical in these analyses \citep{Glegg1999,Peake1992,Peake1993}, it is useful to rotate the coordinate system so that
\begin{align*}
({x^\ast},{y^\ast},{z^\ast}) & = \left( \tilde{x} \cos(\chi^\ast) -  \tilde{y}\sin(\chi^\ast),   \tilde{x} \sin(\chi^\ast) +  \tilde{y} \cos(\chi^\ast), \tilde{z}\right),
\end{align*}
\begin{figure}
\centering
\begin{tikzpicture}[remember picture]
\coordinate (coord) at ($({-7*cos(\TikzStag)},{-7*sin(\TikzStag)})$);

\draw[gray,ultra thick,rotate=\TikzStag] (0,0)--(\AScale,0);
\draw[gray,ultra thick,shift = {(0,\TikzDist)},rotate=\TikzStag] (0,0)--(\AScale,0);
\draw[gray,ultra thick,shift = {(0,-\TikzDist)},rotate=\TikzStag] (0,0)--(\AScale,0);

\draw[dashed,ultra thick,shift = {(0,\TikzDist)},rotate=\TikzStag] (0,0)--(\AScale,0);
\draw[dashed,ultra thick,rotate=\TikzStag] (0,0)--(\AScale,0);
\draw[dashed,ultra thick,shift = {(0,-\TikzDist)},rotate=\TikzStag] (0,0)--(\AScale,0);

\myaxis{(coord)}{0}{2}{$\tilde{x}$}{$\tilde{y}$}
\draw ($(coord)+(\StagArcRad,0)$) arc[start angle=0, end angle=\TikzStag, radius=\StagArcRad];

\draw[-Latex,shift={($(coord)+(0,1)$)},rotate={\TikzStag}] (0,0)--(\AoAArcRad+.5,0);
\draw[-Latex,shift={(coord)},rotate={\TikzStag}] (0,0)--(\AoAArcRad+0.5,0);
\draw[-Latex,shift={($(coord)+(0,0.5)$)},rotate={\TikzStag}] (0,0)--(\AoAArcRad+.5,0);
\node[shift={($(coord)+({(\AoAArcRad+3)/2*cos(\TikzStag+\TikzAoA)},{(\AoAArcRad+.5)*sin(\TikzStag+\TikzAoA)})$)}] {$U^\ast_\infty$};

\coordinate (wave) at (-6,4);

 \draw[thick, scale=1,domain=0:1.5,smooth,variable=\x,rotate=\TikzStag-40,shift={(wave)},samples=200] plot ({\x},{0.5*sin(2*360*\x)});
\draw[thick,-Latex,shift={(wave)},rotate=\TikzStag-40] (0,0)--(3,0);
\node[above right=10pt] at (wave) {unsteady perturbation};

\coordinate (meet) at ($({-\TikzDist*cos(\TikzStag)*sin(\TikzStag)},{\TikzDist*((cos(\TikzStag)*cos(\TikzStag))})$);
\draw[Latex-Latex, shift={($-0.1*(meet)$)}] (0,0) -- ($({\AScale*cos(\TikzStag)},{\AScale*sin(\TikzStag)})$) node [below, midway,rotate=\TikzStag] {$2b^\ast$};

\draw[Latex-Latex] ($({\AScale*cos(\TikzStag)/2},{\AScale*sin(\TikzStag)/2})$) -- ($({\AScale*cos(\TikzStag)/2},{\AScale*sin(\TikzStag)/2+\TikzDist})$) node [right, midway] {$\Delta^\ast$};

\draw[gray,dashdotted] (coord)--(0,0);%
\node at ($(coord)+({(.5+\StagArcRad)*cos(\TikzStag/2)},{(.5+\StagArcRad)*sin(\TikzStag/2)})$) {$\chi^\ast$};
\node[below=4pt] at ($(coord)+(1,0)$) {background flow};

\path[fill=mygray] ({\AScale*cos(\TikzStag)/2},4.4) circle (0.07);
\path[fill=mygray] ({\AScale*cos(\TikzStag)/2},4.7) circle (0.07);
\path[fill=mygray] ({\AScale*cos(\TikzStag)/2},5) circle (0.07);

\path[fill=mygray] ({\AScale*cos(\TikzStag)/2},-3.4) circle (0.07);
\path[fill=mygray] ({\AScale*cos(\TikzStag)/2},-2.8) circle (0.07);
\path[fill=mygray] ({\AScale*cos(\TikzStag)/2},-3.1) circle (0.07);
\end{tikzpicture}
\caption{A rectilinear cascade of flat plates with complex boundaries.}
\label{Fig:poroCascade}
\end{figure}
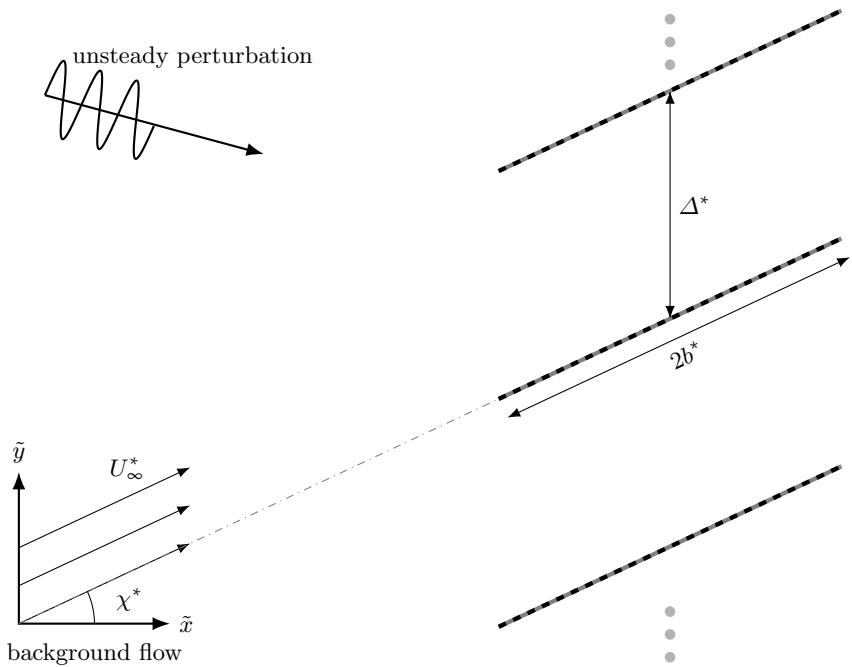
\begin{figure}
\begin{tikzpicture}[scale=1.1,y={(-1cm,0.5cm)},x={(1cm,0.5cm)}, z={(0cm,1cm)}]

\begin{scope}[shift = {(-5,0)} ]
\draw[-latex] (-.3,-.3,0) -- +(2, 0,  0)  node [below, pos = .75] {$U^\ast_\infty$};
\draw[-latex] (-.3,.3,0) -- +(2, 0,  0);
\draw[-latex] (-.3,0,0) -- +(2, 0,  0);
\draw[-latex] (-.3,-.3,0) -- +(0,  -2, 0) node [below, pos =.75] {$W^\ast_\infty$};
\draw[-latex] (0,-.3,0) -- +(0,  -2, 0);
\draw[-latex] (.3,-.3,0) -- +(0,  -2, 0);
\fill[black] (-.3,-.3) -- (-.3,.3) -- (.3,.3) --(.3,-.3) -- cycle;
\draw (-.3,.3)--(.3,.3);
\end{scope}

\foreach \z in {-1,0,1}
{
\draw[fill=gray!10!white] ($(-1.5+.5*\z,-3.5,\z)$) -- ($(-1.5+.5*\z,3.5,\z)$) -- ($(1.5+.5*\z,3.5,\z)$) -- ($(1.5+.5*\z,-3.5,\z)$) -- cycle;
\foreach \x in {-2,-1,0,1,2}
\foreach \y in {-3,...,3}
\draw[fill = gray!50!white] ($(.5*\x+.5*\z,\y,\z)$) circle (.1);
}

\draw[latex-latex] (-1.5,.5,0)--(-1,.5,0) node[below,pos = 1.2] {$d^\ast$};
\draw[latex-latex] (-1,.5,0)--(-1,.5,1) node[midway, right=1pt] {$s^\ast$};

\begin{scope}[shift = {(-3,-4.5)} ]
\coordinate (O) at (0, 2, 0);
\draw[ultra thick,-latex] (O) -- +(1, 0,  0) node [below,pos = 1] {$x^\ast$};
\draw[ultra thick,-latex] (O) -- +(0,  -1, 0) node [below,pos=1] {$z^\ast$};
\draw[ultra thick,-latex] (O) -- +(0,  0, 1) node [left, pos = 1] {$y^\ast$};
\end{scope}

\end{tikzpicture}
\caption{A three-dimensional view of the cascade the rotated, dimensional $(x^\ast,y^\ast,z^\ast)$ coordinate system. The chordwise and spanwise background velocities are denoted by $U^\ast_\infty$ and $W^\ast_\infty$ respectively. The complex boundaries are illustrated by the ``holes'' on each blade, which may represent compliance, porosity, or impedance.}
\label{Fig:3d}
\end{figure}
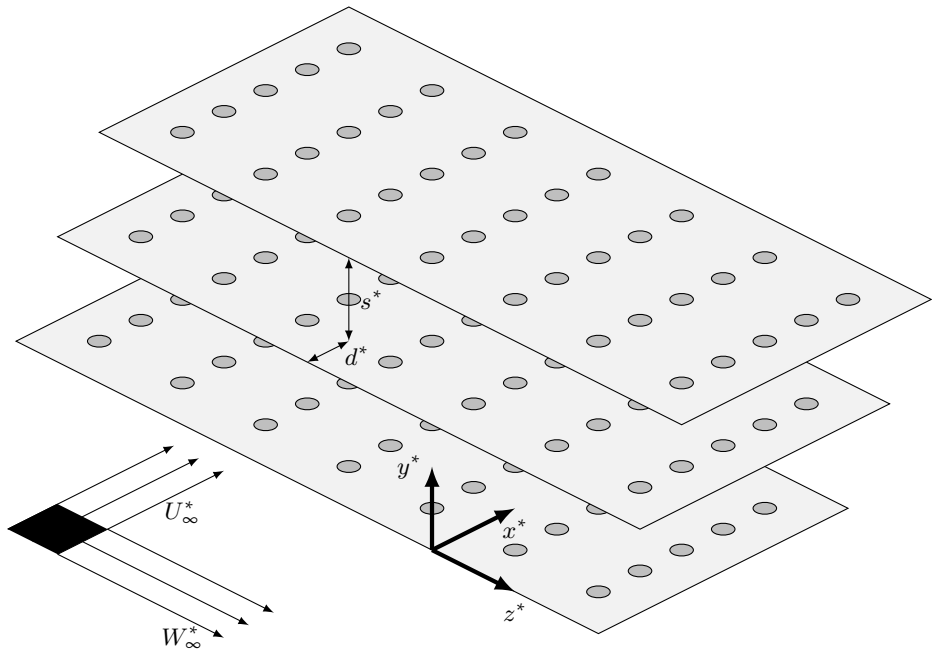
and the $x^\ast$ and $y^\ast$ coordinates are tangential and normal to the blades respectively, which have dimensional length $2b^\ast$. The background flow is tangential to the blades and may have a spanwise component such that \mbox{$\boldsymbol{U}_0^\ast = (U^\ast_\infty,0,W^\ast_\infty)$}, as illustrated in figure \ref{Fig:3d}. The blades in the cascade are inclined at stagger angle $\chi^\ast$, and the distance between adjacent blades is $\Delta^\ast$. Consequently, the spacing between blades is simply given by
\begin{align*}
(d^\ast, s^\ast ) & = \Delta^\ast \left(\sin(\chi^\ast),\cos(\chi^\ast)\right).
\end{align*}
We further assume that a vortical or acoustic wave is incident on the cascade, resulting in a velocity perturbation $\boldsymbol{u}^\ast$ to the mean flow. The Kutta condition is satisfied by ensuring that there is no pressure jump at the blades' trailing-edges.

We introduce an acoustic potential function for the scattered field defined by 
\begin{align*}
\nabla {\phi^\ast} & = \boldsymbol{u}^\ast.
\end{align*}
Consequently, conservation of momentum yields the scattered pressure as
\begin{align}
p^\ast = - \rho_0 \frac{\textrm{D}_0{\phi^\ast}}{\textrm{D}_0 t^\ast}, \label{c5:Eq:scatPres}
\end{align}
where the (linearised) convective derivative is defined as
\begin{align*}
\frac{\textrm{D}_0}{\textrm{D}_0 t^\ast} & = \frac{\partial }{\partial t^\ast} + \boldsymbol{U}_0^\ast \cdot \nabla =  \frac{\partial }{\partial t^\ast} + U^\ast_\infty \frac{\partial}{\partial x^\ast} + W^\ast \frac{\partial}{\partial z^\ast}.
\end{align*}
Accordingly, conservation of mass yields the convected wave equation
\begin{align}
\frac{1}{c_0^2} \frac{\textrm{D}_0^2 \phi^\ast}{\textrm{D}_0 t^{\ast2}} - \nabla^2 {\phi^\ast} & = 0, \label{c5:Eq:convecHelm}
\end{align}
where $c_0$ is the isentropic speed of sound.

We suppose that the unsteady perturbation incident on the cascade takes the form
\begin{align}
\phi_i^\ast &=\exp \left[ \i \left(k_x^\ast x^\ast + k_y^\ast y^\ast + k_z^\ast z^\ast -\omega^\ast t^\ast\right) \right]. \label{c5:Eq:incField}
\end{align}
Dimensional variables are denoted with $\ast$ whereas non-dimensional variables have no such annotation. Since the system is infinite in the spanwise direction, the scattered solution $\boldsymbol{u}^\ast$ must also have harmonic dependence in the $z$-direction. Accordingly, making the following convective transformation and non-dimensionalisations
\begin{align*}
\phi^\ast(x^\ast, y^\ast, z^\ast, t^\ast) = U^\ast_\infty b^\ast \phi(x,y) \exp \left[\i \omega \left(-M^2 \delta  x + k_z z -  t \right) \right],
\end{align*}
\begin{align*}
d^\ast = b^\ast d, \qquad s^\ast = b^\ast \frac{s}{\beta}, \qquad \Delta = \sqrt{d^2 + s^2},
\end{align*}
\begin{align*}
x^\ast = b^\ast x, \qquad y^\ast = b^\ast \frac{y}{\beta}, \qquad z^\ast = b^\ast z, \qquad   t^\ast = \frac{\omega}{\omega^\ast} t,
\end{align*}
\begin{align*}
k_x^\ast = \frac{\delta (k_x- M^2 \omega)}{b^\ast}, \qquad \qquad k_y^\ast =  \frac{\omega \beta k_y}{b^\ast  }, \qquad \qquad k_z^\ast = \frac{\omega k_z}{b^\ast},
\end{align*}
 \begin{align*}
 M = U^\ast_\infty/c_0, \qquad \beta = \sqrt{1 - M^2}, \qquad \delta = 1/{\beta^2},
 \end{align*}
 \begin{align*}
\omega  = \frac{b^\ast}{U^\ast_\infty} (\omega^\ast - W^\ast_\infty k_z^\ast), \qquad W^\ast_\infty = U^\ast_\infty W_\infty, %
\end{align*}
\begin{align*}
w^2 = (M \delta)^2 - (k_z/\beta)^2 -(2+2 \i) k_z  \delta M^2 W(1-W_\infty k_z),
\end{align*}
reduces  \eqref{c5:Eq:convecHelm} to
\begin{align}
\left(\frac{\partial^2}{\partial x^2} + \frac{\partial^2}{\partial y^2} + \omega^2 w^2\right) \phi & = 0. \label{c5:Eq:helmholtz}
\end{align}
In terms of these new variables, the scattered pressure \eqref{c5:Eq:scatPres} becomes
\begin{align}
p^\ast = -\rho^\ast U^{\ast 2} p \e^{\i( z - t) }, \qquad 
 p =\frac{\partial}{\partial x} \left( \phi \e^{-\i \omega \delta x} \right) \e^{\i \omega  x },  \label{c5:Eq:presTrans}
\end{align}
where $p$ is the non-dimensional pressure, and the incident perturbation becomes
\begin{align}
\phi_i^\ast = U^\ast b^\ast \phi^i(x,y) \e^{\i \omega \left(-M^2 \delta  x + k_z z -  t \right)}, \qquad
 \phi_i = \exp \left[ \i (\delta k_x x + \omega k_y y)  \right] .\label{c5:Eq:incFieldNon}
\end{align}
\subsection{Boundary Conditions}
We now introduce the boundary conditions for the problem. It is sufficient to specify the behaviour along $y=ns^\pm$ for $n \in \mathbb{Z}$. We use $\Delta_n$ and $\Sigma_n$ to denote the difference and sum of a given quantity either side of the $n^\textnormal{th}$ blade or wake.
\subsubsection{Upstream Boundary Condition}
There may be no discontinuities upstream of the blade row. Consequently, we write
\begin{align}
\Delta_n \phi(x) &= 0, \qquad \qquad x<n d. \label{c5:Eq:upBC}
\end{align}
\subsubsection{Blade Surface Boundary Conditions}
We now introduce several possible blade surface boundary conditions that can be modelled with the present approach. The boundary conditions we consider are the classical impermeable, rigid blade, a porous blade without background flow, a porous blade with background flow, and a general impedance condition. The advantage of this approach is that a spectrum of boundary conditions of practical interest can be modelled without needing to modify the method of solution. As we shall see later, the effect of modifying the boundary condition is to modify the kernel in the ensuing Wiener--Hopf analysis.\\

\noindent {\textbf{Case 0}}

\noindent When the blade is rigid and impermeable, the no-flux condition is simply 
\begin{align}
\boldsymbol{u}_T \cdot \boldsymbol{n}  &= 0,  & n d <x<nd + 2, \; \; y = ns^{\pm} ,\label{c5:Eq:noflux1}
\end{align}
where $\boldsymbol{n}$ is the normal vector directed into the blade and $\boldsymbol{u}_T$ denotes the total (incident and scattered) velocity field. We sum the contributions of \eqref{c5:Eq:noflux1} either size of each blade to obtain
\begin{align}
 \Sigma_n \left[\frac{\partial \phi }{\partial y} \right](x)= &- 2 \namp \exp \left[\i( k_x\delta(nd + x) + \omega k_y n s)\right], \qquad n d <x<nd + 2. \label{c5:Eq:noFluxRigid}
\end{align}
where $w_0= \i \omega k_y $ is the non-dimensional amplitude of the normal velocity of the incident perturbation on the 0$^{\textnormal{th}}$ blade. This case has been considered in detail in previous research \citep{Glegg1999,Posson2010} and is therefore not considered further in the present work.\\

\noindent \textbf{{Case I}}

\noindent We now generalise the no-flux condition \eqref{c5:Eq:noflux1} to permit a proportional relationship between the normal velocity and pressure on the surface so that
\begin{align}
\boldsymbol{u}_T \cdot \boldsymbol{n} & = C_{I} p,  & n d <x<nd + 2, \; \; y = ns^{ \pm} , \label{c5:Eq:noflux2}
\end{align}
for some constant $C_{I}$. Summing the contributions either side of the blade in \eqref{c5:Eq:noflux2} yields
\begin{align}
\Sigma_n \left[ \frac{\partial \phi }{\partial y} \right] (x)=&- 2\namp \exp \left[\i( k_x\delta (nd + x) + \omega k_y n s)\right] \notag \\
&+ C_{I} \Delta_n\left[p \right](x), \qquad  nd <x<nd +2. \label{c5:Eq:darcyFlowNMF1}
\end{align}
In the absence of background flow ($M=0$), this condition becomes
\begin{align}
\Sigma_n \left[ \frac{\partial \phi }{\partial y} \right] (x)=&- 2\namp \exp \left[\i (k_x\delta (nd + x) + \omega  k_y n s)\right] \notag \\
&+ C_{I} \Delta_n\left[\phi \right](x), \qquad  nd <x<nd +2. \label{c5:Eq:darcyFlowNMF}
\end{align}
This boundary condition is capable of modelling a range of scenarios. Early research in aeroelasticity \citep{Crighton1970} used a boundary condition of the form of \eqref{c5:Eq:darcyFlowNMF} to analyse the scattering of aerodynamic sound by a compliant plate. In that study, the plate was modelled as possessing inertia, but negligible elastic resistance to deformation. Consequently, the pressure difference across the compliant plate was proportional to the specific mass of the plate multiplied by the acceleration so that, in the notation of the present work, $C_{I} = \left(-\i \omega m \right)^{-1}$ where $m$ is the (non-dimensional) mass of the plate per unit area.

\cite{Leppington1977} later showed that the compliant flat plate model is equivalent to that of a rigid screen with periodically arranged circular apertures when the apertures width is small and the wavelength is large compared with the separation. This model has gained popularity as a tool for analysing the aerodynamic scattering of porous edges \citep{Jaworski2013,Ayton2016a,Kisil2018}. In this case, the non-dimensional porosity parameter is
\begin{align*}
C_{I} &= \frac{\alpha_H K_R}{\pi R^2}.
\end{align*}
where $R$ is the radius of the circular apertures of radius, $K_R$ is the  Rayleigh conductivity, and the fractional open area is $\alpha_H$. \\

\noindent \textbf{{Case II}}

\noindent In the presence of a background flow, the boundary condition \eqref{c5:Eq:darcyFlowNMF1}  becomes
\iffalse
\begin{align*}
\frac{\partial \phi^\ast }{\partial y^\ast}(x^\ast,ns^\ast)+ \frac{\partial \phi_i^\ast}{\partial y^\ast}(x^\ast,n s^\ast) &= B^\ast \Delta_n p^\ast(x^\ast, n s^\ast), \qquad \qquad 0<x<2,
\end{align*}
where we have dropped the $z^\ast$ and $t^\ast$ dependence. Consequently, we obtain in non-dimensional variables
\fi
\begin{align}
\Sigma_n \left[\frac{\partial \phi }{\partial y} \right] (x)=&- 2\namp \exp \left[\i( k_x\delta (nd +x) + \omega  k_y n s)\right]\notag \\
& +  C_{II} \left(\i \omega \delta \Delta_n \left[\phi \right](x) - \Delta_n \left[ \phi_x \right] (x)\right), \qquad  nd <x<nd + 2, \label{c5:Eq:darcyFlow}
\end{align}
where $C_{II}$ is a constant and we have applied \eqref{c5:Eq:presTrans}. The definition of $C_{II}$ is dependent on the exact geometry of the apertures in the blade and flow speed above and below the blades. For example, for circular apertures in low speed flow, \cite{Howe1996a} calculates the revised Rayleigh conductivity parameter $K_R$ as
\begin{align*}
K_R = 2 R (\Gamma_R - \i \Delta_R),
\end{align*} 
where $\Gamma_R$ and $\Delta_R$ are real valued constants. Therefore, in contrast to case I boundary conditions, $C_{II}$ may have real and imaginary components. Moreover, different expressions are obtained for different shaped apertures.
\\

\noindent {\textbf{Case III}}

\noindent We may also consider the effects of an impedance boundary condition. In the presence of background flow, the impedance boundary condition is given by \cite{Myers1980}
\begin{align*}
\boldsymbol{u}_T \cdot \boldsymbol{n} & = \left(\i \omega + \boldsymbol{U}_0 \cdot \nabla - \boldsymbol{n} \cdot \left(\boldsymbol{n} \cdot \nabla \boldsymbol{U}_0\right)\right) \frac{p}{\i \omega Z}.
\end{align*}
The real part of the impedance $Z$ is termed the acoustic resistance and represents the energy transfer of of the blade: if $\Re[Z]>0$ then the blades absorb energy whereas if $\Re[Z]<0$ then the blades produce energy. Since the flow is uniform, this condition applied on the upper and lower surfaces of the blades becomes
\begin{align*}
v_T^{ \pm} & = \mp \left(-\i \omega + {U} \frac{\partial}{\partial x} + W \frac{\partial }{\partial z} \right) \frac{p}{\i \omega Z}.
\end{align*}
We now sum the upper and lower components of this impedance condition to obtain a condition on the sum of the velocity either side of the blade. In terms of non-dimensional variables, the condition becomes
\begin{align}
\Sigma_n \left[\frac{\partial \phi }{\partial y}\right](x) &=-2 \namp\exp \left[\i (k_x \delta (nd + x) + \omega k_y n s)\right] \notag \\
& + C_{III} \left(-2 \omega^2(1+ W k_z) \Delta_n\left[\phi\right](x) - 2 \i M^2 \omega\Delta_n\left[\phi_x\right](x) + \Delta_n \left[\phi_{x,x} \right](x)\right),\label{c5:Eq:bcImpedance}
\end{align}
where $C_{III} = U^{ 3}/(\i \omega Z) $.

The presence of higher order derivatives requires further regularity at the blades' edges. Since the blades are fixed (and are only locally reacting), we enforce
\begin{align*}
\Delta_n[\phi](0) = \Delta_n[\phi](2)=0.
\end{align*}

\noindent \textbf{Summary of Blade Surface Boundary Conditions}

We may characterise all the modified boundary conditions \eqref{c5:Eq:darcyFlowNMF}, \eqref{c5:Eq:darcyFlow} and \eqref{c5:Eq:bcImpedance} in the general form
\begin{align}
\begin{split}
\Sigma_n \left[\frac{\partial \phi }{\partial y} \right] (x) &= -2 \namp \exp \left[\i (k_x\delta (nd + x) + \omega  k_y n s)\right] \\
&+ \mu_0 \Delta_n \left[\phi\right](x) +\mu_1 \Delta_n \left[\phi_x\right](x) + \mu_2 \Delta_n \left[\phi_{x,x}\right](x), \qquad nd <x<nd + 2,
\end{split}
\label{c5:Eq:generalBC}
\end{align}
where the $\mu_n$ are summarised in table \ref{Tab:muSummary} for the different boundary conditions. Furthermore, in the present analysis we do not allow any added mass and enforce that there is no jump in the normal velocity either side of the plate. Accordingly, we may write
\begin{align}
\Delta_n \left[ \frac{\partial \phi}{\partial y} \right](x) &= 0, \qquad \qquad nd  < x< nd+ 2. \label{c5:Eq:bcNFB}
\end{align}
{\def\arraystretch{2}
	\setlength\arrayrulewidth{2pt}
	\setlength{\tabcolsep}{.4em}
\def\myCiteGlegg{1}
\def\myCitePosson{2}
\def\myCiteLepp{3}
\def\myCiteHowe{4}
\def\myCiteJaw{5}
\def\myCiteKisil{6}
\def\myCiteHoweA{7}
\def\myCiteMyers{8}
\def\myCiteBramb{9}
	\centering
	\begin{table}
		\centering
		\begin{tabular}{cccccc}
		\hline
		Case & Model &$\mu_0$ &  $\mu_1$ & $\mu_2$\\
		\hline
		\rowcolor{tablegray!30!white}
	\parbox[c]{2cm}{\centering Case 0 \\[0pt] [\myCiteGlegg,\myCitePosson]}& rigid, impermeable & 0 & 0 & 0 \\
\parbox[c]{2cm}{\centering Case I \\[0pt] [\myCiteLepp,\myCiteHowe,\myCiteJaw,\myCiteKisil] }  & \parbox[c]{5cm}{\centering porous, compliant\\(no background flow)}  &$C_I$ & 0 & 0 \\[6pt]
		\rowcolor{tablegray!30!white}
\parbox[c]{2cm}{\centering Case II \\[0pt] [\myCiteHoweA] } & \parbox[c]{5cm}{\centering porous \\(with background flow)}  & $\i \omega \delta C_{II}  $ & $-C_{II}$ & 0 \\[6pt]
\parbox[c]{2cm}{\centering Case III \\[0pt] [\myCiteMyers,\myCiteBramb]} & impedance & 
		$-2 \omega^2(1+ W k_z) C_{III}$ & $ - 2 \i M^2 \omega C_{III}$ & $C_{III}$ \\
		\end{tabular}
		\caption[Summary of possible boundary conditions.]{Summary of possible boundary conditions and corresponding $\mu_0$, $\mu_1$ and $\mu_2$ values for equation \eqref{c5:Eq:generalBC}. The references highlight relevant papers, although only [\myCiteGlegg,\myCitePosson] consider cascade geometries and are restricted to rigid boundaries. The reference numbers correspond to [\myCiteGlegg] \citep{Glegg1999}, [\myCitePosson] \citep{Posson2010}, [\myCiteLepp] \citep{Leppington1977}, [\myCiteHowe] \citep{Howe1998}, [\myCiteJaw] \citep{Jaworski2013}, [\myCiteKisil] \citep{Kisil2018}, [\myCiteHoweA] \citep{Howe1996a}, [\myCiteMyers] \citep{Myers1980}, and [\myCiteBramb] \citep{Brambley2009}.} \label{Tab:muSummary}
	\end{table}
}
\subsubsection{Downstream Boundary Conditions}
Downstream, we require the pressure jump across the wake to vanish:
\begin{align}
\Delta_n \left[ p\right](x) = 0, \qquad \qquad x>2 + n d.
\end{align}
By employing the pressure definition \eqref{c5:Eq:presTrans} and integrating with respect to $x$, we may write the above condition as 
\iffalse
\begin{align*}
\left(\i \omega \delta \left[\Delta_n\right] \phi - \Delta_n \left[\phi_x \right]\right) = 0 \qquad \qquad x>2.
\end{align*}
This expression may be integrated to obtain
\fi
\begin{align}
\Delta_n \left[\phi\right](x) & = 2 \pi \i P \exp \left[\i {\omega}\delta x\right], \qquad \qquad x>nd + 2, \label{c5:Eq:npjWake}
\end{align}
where $P$ is a constant of integration that will be specified by enforcing the Kutta condition.

Additionally, the normal velocity across the wake must vanish, i.e.
\begin{align}
\Delta_n \left[\frac{\partial \phi}{\partial y} \right] (x) = 0, \qquad \qquad x>nd + 2. \label{c5:Eq:bcNFW}
\end{align}
\subsubsection{Summary of Full Boundary Conditions}
All in all, we have five boundary conditions. In the upstream region we do not permit any discontinuities \eqref{c5:Eq:upBC}. Along each blade we have a relation for the sum of normal velocities either side of the blade \eqref{c5:Eq:generalBC}, and do not permit a jump in normal velocity across the blade \eqref{c5:Eq:bcNFB}. Finally, across the wake we do  not  permit a jump in pressure \eqref{c5:Eq:npjWake} or normal velocity \eqref{c5:Eq:bcNFW}. The boundary conditions are illustrated in figure \ref{Fig:boundaryConditions}.
\tikzexternalexportnextfalse%
\begin{figure}
	\centering\scalebox{.9}{
	\begin{tikzpicture}[scale=1.1]
\node at (0,1) {};
	\coordinate (coord) at (-2,1);
	\path[draw, dotted] (-3,0)--(0,0) node[midway,above] {\scriptsize no discontinuities \eqref{c5:Eq:upBC}};
	\node[anchor=east] at (-3.3,0) {$y=ns^\pm$};
		\node[anchor=north] at (0,-\bh) {$x=nd$};	
		\node[anchor=north] at (\AScale,0) {$x=nd+2$};
	\draw[dotted] (\AScale,0)--($(\AScale+4,0)$) node[midway,above] {
	\begin{tabular}{c}
\scriptsize	no pressure jump \eqref{c5:Eq:npjWake} \&\\[-.4em]
\scriptsize no normal velocity jump \eqref{c5:Eq:bcNFW}
	\end{tabular}};
	\draw[line width = 1mm] (0,0)--(\AScale,0);
	\draw[gray,line width = 1mm, dashed]  (0,0)--(\AScale,0) node[ black, midway,above] {
		\begin{tabular}{c}
\scriptsize		complex boundary \eqref{c5:Eq:generalBC} \&\\[-.4em]
\scriptsize		no normal velocity jump \eqref{c5:Eq:bcNFB}
		\end{tabular}};
	\end{tikzpicture}}
	\caption{Schematic illustrating where each boundary condition is applied.}
	\label{Fig:boundaryConditions}
\end{figure}
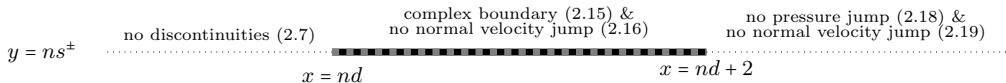
This completes the description of the mathematical model.

\section{Solution} 
\label{c5:Sec:MathSol}
We now present the mathematical solution to the Helmholtz equation \eqref{c5:Eq:helmholtz} subject to the boundary conditions \eqref{c5:Eq:upBC}, \eqref{c5:Eq:generalBC}, \eqref{c5:Eq:bcNFB}, \eqref{c5:Eq:npjWake} and \eqref{c5:Eq:bcNFW}. For clarity, we present a ``road map'' of the solution in figure \ref{Fig:roadmap}. %

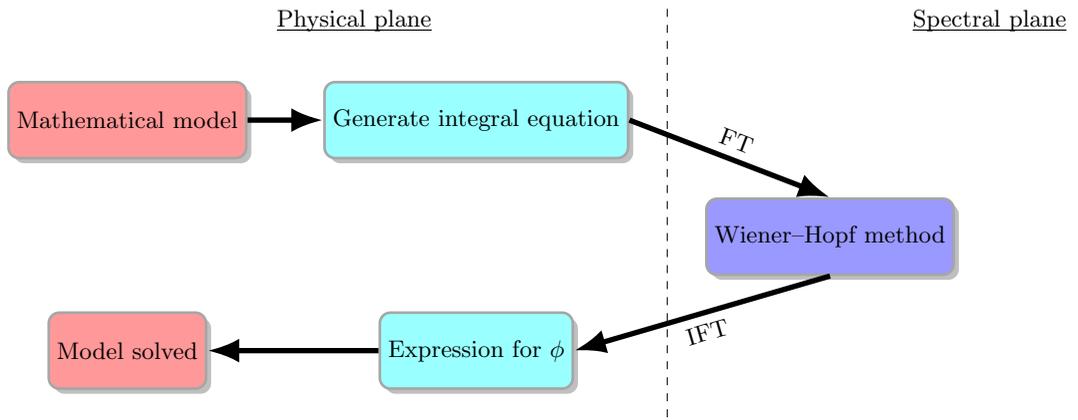
\begin{figure}
\colorlet{colD}{red!40}
\colorlet{colIP}{cyan!40}
\colorlet{colV}{blue!40}
\colorlet{colBorder}{gray!70}
\tikzset
  {mybox/.style=
    {rectangle,rounded corners,drop shadow,minimum height=1cm,
     minimum width=2cm,align=center,fill=#1,draw=colBorder,line width=1pt
    },
   myarrow/.style=
    {draw=#1,line width=2pt,-LaTeX,rounded corners
    },
   mylabel/.style={text=#1}
  }
\centering
\begin{tikzpicture}
  \node[mybox=colD] (D) {Mathematical model};
  \node[mybox=colIP,right=of D] (IP) {Generate integral equation};
  \node[mybox=colV,below right= .5cm and 1cm of IP] (V) {Wiener--Hopf method};
    \node[mybox=colIP,below = 2cm of IP] (V2) {Expression for $\phi$};
        \node[mybox=colD,below = 2cm of D] (V3) {Model solved};
    \draw[dashed] ($(V.west)+(-.5,3)$)--($(V.west)+(-.5,-2.5)$);
    \node at (3,1.2) {\underline{\smash{Physical plane}}};
    \node at (11.4,1.2) {\underline{\smash{Spectral plane}}};
      \draw[myarrow = black] (D)--(IP);
  \draw[myarrow = black] (IP.east)--(V.north) node[midway,above,sloped] {FT};
    \draw[myarrow = black] (V.south)--(V2.east) node[midway,below,sloped] {IFT};
      \draw[myarrow=black] (V2) -- (V3);
\end{tikzpicture}
\caption{Schematic diagram illustrating the solution method. The abbreviations ``FT'' and ``IFT'' stand for ``Fourier transform'' and ``Inverse Fourier transform'' respectively. \label{Fig:roadmap}}
\end{figure}

As is typical in cascade acoustics problems we employ integral transforms to obtain a solution that is uniformly valid throughout the entire domain \citep{Peake1992,Glegg1999,Posson2010}. However, $\phi$ is discontinuous across each blade and wake in the $y$-direction. Therefore, $\partial \phi/\partial y$ possesses non-integrable singularities thus preventing the application of a Fourier transform. Consequently, we must regularise the derivatives of $\phi$ and to remove these non-integrable singularities. To this end, we introduce introduce generalised derivatives \citep{Lighthill1958} and write 
\begin{align}
\frac{\partial^2 \phi}{\partial y^2} &=\frac{\tilde{\partial}^2\phi}{\tilde{\partial}y^2} - \sum_{n=-\infty}^{\infty} \Delta_n \left[\phi \right](x) \delta^{\prime}(y-n s) - \sum_{n=-\infty}^{\infty} \Delta_n \left[\frac{\partial \phi}{\partial y}\right](x) \delta(y- n s), \label{c5:Eq:GeneralisedLaplace}
\end{align}
where $\tilde{\partial}$ represents the partial derivative with discontinuities removed. The second term in \eqref{c5:Eq:GeneralisedLaplace} vanishes because there is zero jump in normal velocity across the blade \eqref{c5:Eq:bcNFB} and wake \eqref{c5:Eq:bcNFW}.

\def\miBP{\sigma^\prime}
\def\hPS{d}
\def\vPS{s}

The scattered solution must obey the same quasi-periodicity relation as the incident field  \eqref{c5:Eq:incField}. Consequently, the scattered acoustic potential function in the entire plane may be reduced to that of a single channel in the domain by writing
\begin{align}
\phi(x + n d, y + n s) = \phi(x, y) \e^{\i n \sigma^\prime}, \label{c5:Eq:InterBladePhase}
\end{align}
where the inter-blade phase angle for $\phi$ is $\miBP = k_x\delta d + \omega  k_y s$. Substituting \eqref{c5:Eq:GeneralisedLaplace} into the Helmholtz equation \eqref{c5:Eq:helmholtz} and applying the inter-blade phase angle relation \eqref{c5:Eq:InterBladePhase} yields 
\begin{align}
\frac{\partial^2 \phi}{\partial x^2} + \frac{\partial^2 \phi}{\partial y^2} + \omega^2 w^2 \phi=  & \sum_{n=-\infty}^{\infty} \Delta_0 \left[\phi \right] (x- n \hPS) \delta^{\prime} (y-n \vPS)\e^{\i n \miBP}. \label{c5:Eq:OeGE}
\end{align}
We define the Fourier integral transform and its inverse as
{\setlength{\jot}{3pt}
	\begin{align*}
	F(\gamma,\eta)&=\frac{1}{(2\pi)^2} \int_{-\infty}^{\infty} \int_{-\infty}^{\infty} f(x,y)\e^{\i \gamma x+\i \eta y} \; \d x\; \d y,\\
	f(x,y)&= \phantom{\frac{1}{(2\pi)^2}}  \int_{-\infty}^{\infty} \int_{-\infty}^{\infty}F(\gamma, \eta) \e^{-\i \gamma x-\i \eta y} \; \d \gamma \; \d \eta.
	\end{align*}
}
Applying the transform to the left hand side of \eqref{c5:Eq:OeGE} yields
\begin{align}
(-\gamma^2-\eta^2+\omega ^2 w^2) \Phi(\gamma,\eta) &=\frac{1}{2 \pi \i} \sum_{n=-\infty}^{\infty} \eta D(\gamma) \e^{\i n(\miBP +\gamma \hPS + \eta \vPS)}, \label{c5:Eq:FTge}
\end{align}
The problem is now to find $D(\gamma)$ which represents the Fourier transform of the jump in acoustic potential either side of the blade and wake. We invert the Fourier transform to obtain an expression for the acoustic potential in terms of $D$:
\begin{align}
\phi(x,y)=\frac{1}{2 \pi \i }\int_{-\infty}^{\infty} \int_{-\infty}^{\infty}  \frac{\e^{-\i \gamma x- \i \eta y}}{\omega^2 w^2-\eta^2-\gamma^2} \cdot 
\sum_{n=-\infty}^{\infty} \eta D(\gamma) \e^{\i n(\miBP+\gamma \hPS + \eta \vPS)} \d \gamma \d \eta . \label{c5:Eq:FourierInversionGeneral}
\end{align}
The $\eta$ integral may be performed by closing the contour of integration in an appropriate upper or lower half-plane to obtain
\begin{align}
\phi(x,y)=- \frac{1}{2}  \int_{-\infty}^{\infty}  \sum_{n=-\infty}^{\infty}  D(\gamma) \sgn\left( n\vPS - y \right) \e^{\i n(\miBP+\gamma \hPS ) + \i \zeta|n \vPS- y| }\e^{-\i \gamma x} \d \gamma, \label{c5:Eq:AcField}
\end{align}
where $\zeta=\sqrt{\omega^2w^2-\gamma^2}$. The branch cut is defined so that $\Im \left[ \zeta \right]>0$ when $\gamma$ is in a strip for the Wiener--Hopf method.

To obtain an equation for the unknown $D$ we must apply the relevant boundary conditions. We first differentiate \eqref{c5:Eq:AcField} with respect to $y$ and consider the limits $y \rightarrow 0^\pm$. Summing the contributions from each of these limits yields the integral equation
\begin{align}
\Sigma_0\left[\frac{\partial \phi}{\partial y} \right] (x)&= -4 \pi \int_{-\infty}^{\infty} D(\gamma) j(\gamma)\e^{-\i \gamma x} \d \gamma , \label{c5:Eq:intEq}
\end{align}
where
\begin{align}
j(\gamma)&=  \frac{\i \zeta}{4 \pi} \sum_{n \in \mathbb{Z}} \e^{ \i n (\miBP+ \gamma \hPS)+ \i \zeta |n \vPS|}  = \frac{\zeta}{4 \pi} \cdot  \frac{\sin \left( \zeta \vPS \right)}{\cos \left( \zeta \vPS \right) - \cos \left( \gamma \hPS + \miBP \right)} . \label{c5:Eq:jSumTrig}
\end{align}
We now solve equation \eqref{c5:Eq:intEq} subject to the remaining boundary conditions applied on $y = 0$:
\begin{align}
\addtocounter{equation}{1}
\Delta_0 \left[\phi \right](x) &= 0, &x<0, \label{c5:Eq:upBC2}\tag{\theequation .a}\\[1em]
\Sigma_0 \left[\frac{\partial \phi }{\partial y} \right] (x) &= \mu_0 \Delta_0 \left[\phi \right](x) +\mu_1 \Delta_0 \left[\phi_x\right](x) + \mu_2 \Delta_0\left[\phi_{x,x}\right](x) \notag \\
&-2 \namp \exp \left[\i k_x \delta  x\right] & 0 <x<2,\label{c5:Eq:generalBC2} \tag{\theequation .b} \\[1em]
\Delta_0 \left[\phi \right](x) & = 2 \pi \i P \exp \left[\i \delta {\omega} x\right],  & x>2. \label{c5:Eq:npjWake2} \tag{\theequation .c}
\end{align}
The system (\ref{c5:Eq:intEq}, \ref{c5:Eq:upBC2}, \ref{c5:Eq:generalBC2}, \ref{c5:Eq:npjWake2}) represents an integral equation subject to mixed value boundary conditions. We solve this system via the Wiener--Hopf method as detailed in appendix \ref{c5:Sec:WHsol}. The solution for $D$ is given by
\begin{align}
D(\gamma) =& \frac{\namp}{(2 \pi)^2 \i ( \gamma + \delta k_x)K_-(-\delta k_x) K_+(\gamma)} + \frac{\namp \delta ( \omega - k_x) \e^{2 \i (\gamma + \delta k_x)}}{(2 \pi)^2 \i(\gamma+ \delta k_x)(\gamma + \delta \omega ) K_+(- \delta k_x)K_-(\gamma)} \notag \\
-& \sum_{n=0}^{\infty} \frac{\left(\Asig + \Csig \right) \e^{2 \i (\gamma - \dModeM) }}{\i(\gamma + \delta \omega )(\gamma - \dModeM)}\cdot\frac{K_-(\dModeM) }{K_-(\gamma)} - \sum_{n=0}^{\infty} \frac{\Bsig}{\gamma - \dModeP} \cdot \frac{K_+(\dModeP)}{K_+(\gamma)}, \label{c5:Eq:Dsol}
\end{align}
where all new variables are defined in appendix \ref{c5:Sec:WHsol}. Note that the solution is identical to that for the rigid cascade \citep{Glegg1999}, except the original Wiener--Hopf kernel $j$ is now replaced with the modified kernel $K$. This original kernel is recovered when \mbox{$\mu_1 = \mu_2 = \mu_3 =0$} and the solution reduces to that derived by \cite{Glegg1999}.
\subsection{Inversion of Fourier Transform} \label{c5:Sec:FourInversion}

We now invert the Fourier transform of the acoustic field in the previous section. Since $D$ is now known, the Fourier inversion integral in \eqref{c5:Eq:AcField} can now be computed. Similarly to the analysis in \cite{Posson2010}, the inversion is performed by splitting the physical domain into five regions as illustrated in figure \ref{c5:Fig:FourInvRegions}.
\begin{figure}
	\centering
	\begin{tikzpicture}[scale=1.25,remember picture]
	\coordinate (coord) at (-2,1);
	\coordinate (le1) at ($({\TikzDistb*sin(\TikzStagb)},{\TikzDistb*cos(\TikzStagb)})$);
	\coordinate (le1x) at ($({\TikzDistb*sin(\TikzStagb)},0)$);	
	\coordinate (le2x) at ($(0,{-\TikzDistb*cos(\TikzStagb)})$);
	\coordinate (le2) at ($({-\TikzDistb*sin(\TikzStagb)},{-\TikzDistb*cos(\TikzStagb)})$);
	\draw[ultra thick] (0,0)--(\AScale,0);
	\draw[ultra thick] (le1)--($(le1)+(\AScale,0)$);
    \draw[ultra thick] (le2)--($(le2)+(\AScale,0)$);

	\draw[dashed] (le2)--(le1);
	\draw[dashed] ($(le2)+(\AScale,0)$)--($(le1)+(\AScale,0)$);
	\draw[dashed] (le1)--(le1x);
		\draw[Latex-Latex] ($(le1)+(.1,0)$)--($(le1x)+(.1,0)$) node[midway,right] {$\vPS$};
			\draw[Latex-Latex,] (0,-.1)--($(le1x)-(0,.1)$) node[midway,below] {$\hPS$};
 \draw[dashed] ($(\AScale,0)-(le1x)$)--($(\AScale,0)+(le2)$);
\draw[dashed] ($(\AScale,0)-(le2x)$)--(\AScale,0);
	\draw[dashed] (0,0)--(le2x);
	\myaxis{(coord)}{0}{.5}{$x$}{$y$}
	
	\node at ($(-\AScale/4,0)+.5*(le1)$) {\ra};
	\node at ($(-\AScale/4,0)-.5*(le1)$) {\ra};
	
	\node at ($.25*(le1x)+.5*(le1)$) {\rb};
	\node at ($.25*(le1x)-.5*(le1)$) {\rb};
		
	\node at ($(\AScale/2,0)+.5*(le1)$) {\rc};
	\node at ($(\AScale/2,0)-.5*(le1)$) {\rc};
			
	\node at ($(\AScale,0)-.25*(le1x)+.5*(le1)$) {\rd};
	\node at ($(\AScale,0)-.25*(le1x)-.5*(le1)$) {\rd};						
	
	\node at ($(5*\AScale/4,0)-.5*(le1)$) {\re};
	\node at ($(5*\AScale/4,0)+.5*(le1)$) {\re};
		
	\draw[Latex-Latex] ($(le1)+(0,.1)$)--($(le1)+(\AScale,.1)$) node[midway, above] {$2$};
	\end{tikzpicture}
	\caption{Diagram indicating the different regions in the $(x,y)-$plane which require different areas of contour integration in the Fourier inversion.} \label{c5:Fig:FourInvRegions}
\end{figure}
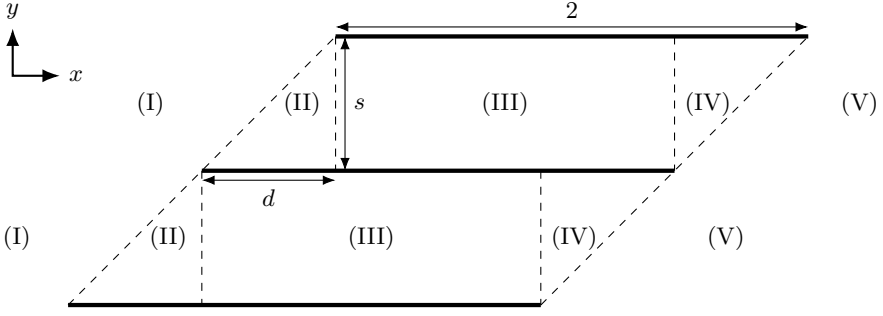

The details can be found in Appendix \ref{Ap:FourInv} and the final results are stated below. All undefined functions are defined in Appendices \ref{c5:Sec:WHsol} and \ref{Ap:FourInv}.

\subsubsection[Upstream Region \ra]{Upstream Region \ra}
In the upstream region, 
\begin{align}
\phi(x,y) &=  \pi \i \sum_{m=-\infty}^{\infty} D^{(1,3)}(\lambda_m^+) A^r(\aModeP,x,y).  \label{c5:Eq:solU}
\end{align}
\subsubsection[Inter-Blade Upstream Region \rb]{Inter-Blade Upstream Region \rb}
In the inter-blade upstream region, 
\begin{align*}
\phi(x,y)=&  -\pi \sum_{n=0}^{\infty} \frac{\Asig+\Csig}{\dModeM+\delta \omega} A_d(\dModeM,x,y) - \pi \i \sum_{n=0}^{\infty}  \Bsig A_d(\dModeP,x,y) \\
&-\pi \frac{A_d(-\delta k_x,x,y)}{K(-\delta k_x)} \cdot \frac{w_0}{(2 \pi)^2}+ \pi \i \sum_{m=-\infty}^{\infty} D^{(1,3)}(\lambda_m^+) A_{u}^r(\aModeP,x,y). 
\end{align*}
\subsubsection[Inter-Blade Inner Region \rc]{Inter-Blade Inner Region \rc}
In the inter-blade inner region, 
\begin{align*}
\phi(x,y)=& -\pi \sum_{n=0}^{\infty} \frac{\Asig+\Csig}{\dModeM+\delta \omega} A(\dModeM,x,y) - \pi \i \sum_{n=0}^{\infty}  \Bsig A(\dModeP,x,y) \\
&-\pi \frac{A(-\delta k_x,x,y)}{K(-\delta k_x)} \cdot \frac{w_0}{(2 \pi)^2}.
\end{align*}
\subsubsection[Inter-Blade Downstream Region \rd]{Inter-Blade Downstream Region \rd}
In the inter-blade downstream region, 
\begin{align*}
\phi(x,y)=& -\pi \sum_{n=0}^{\infty} \frac{\Asig+\Csig}{\dModeM+\delta \omega} A_u(\dModeM,x,y) - \pi \i \sum_{n=0}^{\infty}  \Bsig A_{u}(\dModeP,x,y) \\
&-\pi \frac{A_u(-\delta k_x,x,y)}{K(-\delta k_x)} \cdot \frac{w_0}{(2 \pi)^2} \\
&- \pi \i \sum_{m=-\infty}^{\infty} D^{(2,4)}(\aModeM) A_{d}^r(\aModeM,x,y)  + \pi \i P A_d(-\delta \omega,x,y) .
\end{align*}
\subsubsection[Downstream Region \re]{Downstream Region \re}
In the downstream region, 
\begin{align}
\phi(x,y) =& - \pi \i \sum_{m=-\infty}^{\infty} D^{(2,4)}(\aModeM) A^r(\aModeM,x,y) + \pi \i P A(-\delta \omega,x,y). \label{c5:Eq:solD}
\end{align}

\section{Results} \label{c5:Sec:Results}
We now use the solution of the previous section to calculate aeroacoustic quantities of practical interest. In particular, we are interested in the role that porosity may play in noise reduction. Accordingly, the results in this section are focused on porous blades, which is a type II boundary condition in the present nomenclature. We now use the analytic solution derived in section \ref{c5:Sec:FourInversion} to explore the aeroacoustic performance of a blade row with modified boundary conditions. In particular, we focus on the role of porosity due to its potential to attenuate sound, as seen previously in \cite{Jaworski2013} for trailing-edge scattering. The results in the present research also show significant sound reductions for modest changes in porosity. We argue that this is attributed to the strong effect of porosity on the duct modes and unsteady loading: in cascade configurations, the blade loading changes the upstream and downstream flows and therefore influences the intensity of the scattered sound. 

{\def\arraystretch{2.5}
	\setlength\arrayrulewidth{2pt}
	\setlength{\tabcolsep}{.3em}
	\centering
	\begin{table}
		\centering
		\begin{tabular}{cccccc}
		\hline
		Reference & \parbox[c]{2cm}{\centering gap-to-chord ratio\\ $\Delta/2$} & \parbox[c]{2cm}{\centering stagger angle\\ $\chi$} & \parbox[c]{2cm}{\centering Mach number \\$M$} & \parbox[c]{2cm}{\centering reduced frequency\\ $\omega$} & \parbox[c]{2cm}{\centering inter-blade phase angle\\ $\sigma$}  \\
		\hline
		\cite{Glegg1999} & 0.6 & $40^\circ$ & 0.3 & 0--40 & $3 \pi/4$  \\
		\rowcolor{tablegray!30!white}
		Case A & .5 &$0^\circ$ & 0.5&$5\pi/4$&$5\pi/2$ \\
		Case B & .5 &$0^\circ$ & 0.5&$13\pi/4$&$13\pi/2$  \\
				\rowcolor{tablegray!30!white}
		Case C & 1 &$30^\circ$ & 0.3 & 5 & $2 \pi/3$\\
		Case D & .8& $10^\circ$ & 0.4 &$0-20$ & $4 \pi /3$\\
						\rowcolor{tablegray!30!white}
		Case E &$ 0.6$ & $ 40^{\circ}$& 0.3 & 12.5 & 4.3  \\ %
		\end{tabular}
		\caption{Summary of parameters used in results section.} \label{Tab:resSummary}
	\end{table}
}

\subsection{Validation}
We first present a comparison to three previous solutions for cascades of rigid blades in figure \ref{Fig:val}. Firstly, we compare our results to a solution exploiting the Wiener--Hopf method \citep{Posson2010}. Secondly, we compare to a quasi-numerical a mode-matching technique \citep{Bouley2017}, and thirdly, we compare to a fully numerical method \citep{Hall1997}. The solutions show excellent agreement -- in fact our solution is mathematically equivalent to the of \cite{Posson2010} in the special case where the blades are rigid. It is worth noting at this point that our solution satisfies the Kutta condition, as indicated by the pressure jump vanishing at the trailing-edge.
\begin{figure}
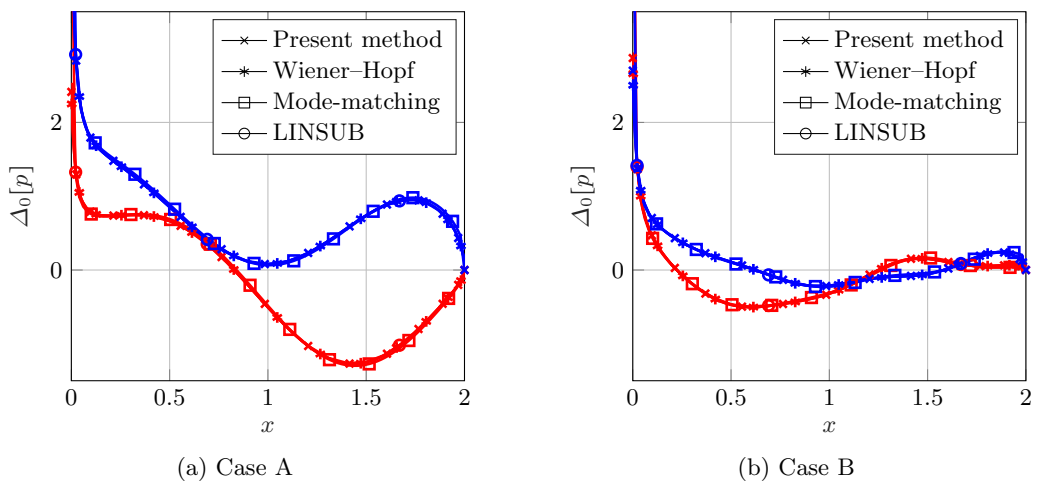

	\begin{subfigure}[b]{.45\linewidth}
		\centering
		\setlength\fheight{.8\textwidth}
		\setlength\fwidth{.9\textwidth}
		\input{images/HallCase1Col.tex}
		\caption{Case \caseHallA}
		 \label{Fig:validationA}
	\end{subfigure}
	\hfill
	\begin{subfigure}[b]{.45\linewidth}
		\centering
		\setlength\fheight{.8\textwidth} 
		\setlength\fwidth{.9\textwidth}
		\input{images/HallCase2Col.tex}
		\caption{Case \caseHallB}
		 \label{Fig:validationB}
	\end{subfigure} 
	\caption{Comparisons of surface pressure jump for flat plates for case 0 boundary conditions for cases defined in \ref{Tab:resSummary}. The real and imaginary parts (\newBlueLine~and~\redLine~respectively) of the pressure jump is compared to three alternative solutions: a Wiener--Hopf method \citep{Posson2010}, a mode-matching technique \citep{Bouley2017}, and a numerical method \citep{Hall1997}.}
	\label{Fig:val}
	\end{figure}

\subsection{Duct Mode Analysis}

The cascade may be divided into five regions as illustrated in figure \ref{c5:Fig:FourInvRegions}. Since we only consider the case where the blades are overlapping, the inter-blade inner region (called region III) is bounded by adjacent blades and therefore may be viewed as a duct. The solution in the duct is matched to the upstream and downstream regions by virtue of the inter-blade upstream region (II) and the inter-blade downstream region (IV). The duct region is therefore essential in establishing the relationship between the upstream and downstream regions, and controls the unsteady lift and sound power output. We now explore the behaviour of the solution in the duct in order to understand the effects of blade porosity.

The acoustic potential in the duct may be expanded into a sum of exponential functions whose modes are the ``duct modes" \citep{Glegg}. When the blades are rigid, the duct has rigid walls and the modes are given by the standard formula
\begin{align}
\hat{\theta}_n^\pm = \pm \sqrt{\omega^2 w^2 - \left(\frac{n \pi }{s}\right)^2}, \label{Eq:rigidDM}
\end{align}
where we have used $\hat{\cdot}$ to indicate that this solution is valid for the rigid case $C_{II}=0$.
Conversely, there is no simple expression for the duct modes when the blades are porous; for porosity constant $C_{II}$ the duct modes satisfy the transcendental equation
\begin{align}
\zeta(\theta_n^\pm) \sin(s \zeta(\theta_n^\pm)) = -\i C_{II}\left(\omega \delta +\theta_n^\pm\right) \left(\cos(s \zeta(\theta_n^\pm)) - \cos ( d \theta_n^\pm + \miBP)\right). \label{Eq:dmSol}
\end{align}
The solutions for large $\theta_n^\pm$ are available via the asymptotic analysis presented in appendix \ref{c5:Sec:zerosFac}, but otherwise the solutions must generally be found with a numerical root finding algorithm. 

It is straightforward to see from \eqref{Eq:rigidDM} that the $n$-th rigid duct mode is pure imaginary when $\omega w <n \pi/s$ and pure real when $\omega w>n \pi /s $. These conditions correspond to the duct mode being cut-on or cut-off. However, inspection of \eqref{Eq:dmSol} reveals that for an arbitrary finite porosity constant (but not pure imaginary), it is impossible for the duct modes to be cut-on. This is  readily seen by noting that there are no real solutions to equation \eqref{Eq:dmSol}. If a real root did exist, then the left hand side would be pure real. However, in that case the right hand side would be pure imaginary and we have a contradiction. Consequently, for any non-imaginary porosity coefficients, the duct modes are always complex and never pure real. Accordingly, blade porosity effects a reduction in the magnitude of the pressure field in the inter-blade inner region which, when matched to the upstream and downstream regions, results in a reduction in the far-field scattered sound.

The dependence of the duct modes on blade porosity is illustrated in figure \ref{Fig:dmTraj} for two frequencies. There are no acoustic modes cut-on in figure \ref{Fig:dmTraja}, whereas there are two acoustic modes cut on in figure \ref{Fig:dmTrajb}. Modes located in $\ULHP$ are propagating in the upstream and downstream directions respectively. We consider a range of arguments for the porosity coefficient to represent a phase difference between the pressure jump and normal velocity, which is permitted due to the presence of the background flow \citep{Howe1996a}. Evidently, the relationship between the duct modes is highly complex, although some general trends may be observed. In compliance with the argument in the preceding paragraph, all values of porosity (except pure imaginary) perturb the cut-on duct modes away from the real line. For zero porosity, the duct modes are located at the rigid duct modes, $\hat{\theta}_n^\pm$. As the porosity is increased, the duct modes follow a path that generally ends at either the convected mode ($-\omega \delta$) or an acoustic mode ($\lambda_n^\pm$).

\setlength{\fheight}{8cm}
\begin{figure}
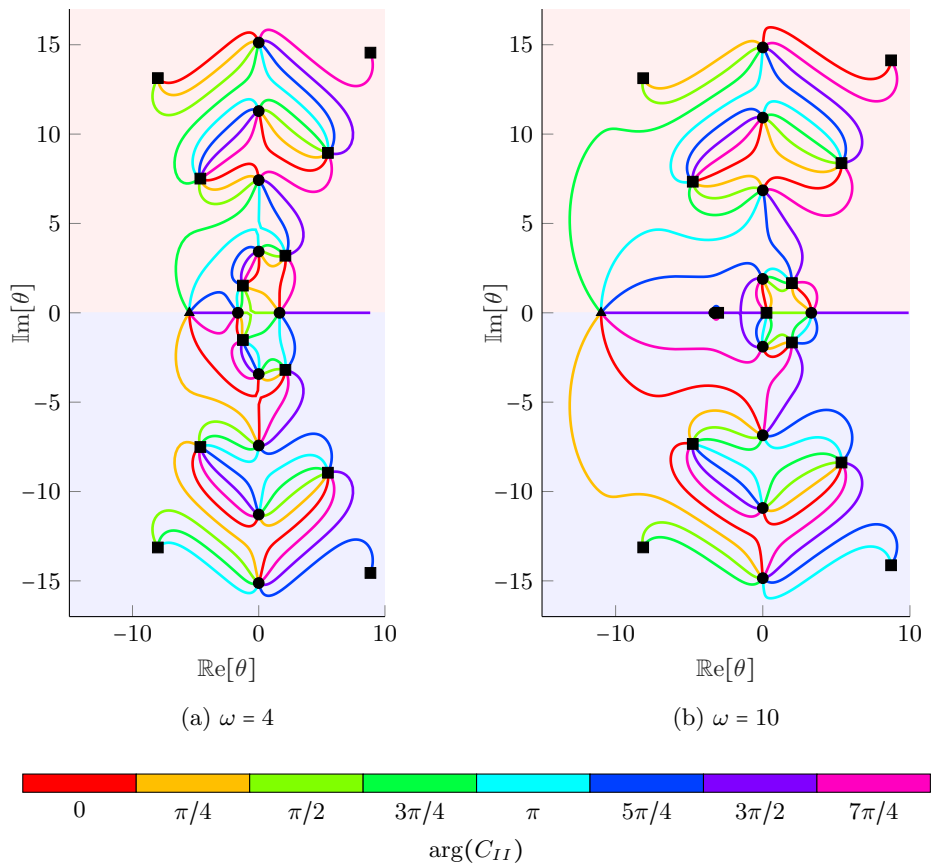

\centering
\begin{subfigure}[!h]{.45\linewidth}
\centering
\setlength{\fwidth}{1.5\linewidth}
\input{images/mode-traj-1.tex}
\caption{$\omega = 4$}
\label{Fig:dmTraja}
\end{subfigure}
\hfil
\begin{subfigure}[!h]{.45\linewidth}
\centering
\setlength{\fwidth}{1.5\linewidth}
\input{images/mode-traj-2.tex}
\caption{$\omega = 10$}
\label{Fig:dmTrajb}
\end{subfigure}

\vspace{.5cm}

\begin{tikzpicture}
\begin{scope}[scale = 6]
\foreach \x in {0,0.125,...,.875} {
	\definecolor{currentcolor}{hsb}{\x, 1, 1}
\draw[fill = currentcolor] (-1+2*\x,0) rectangle (-1+2*\x+.25,.04);
}
\node[below] at (.875,0) {$7\pi/4$};
\node[below] at (.625,0) {$3\pi/2$};
\node[below] at (.375,0) {$5\pi/4$};
\node[below] at (.125,0) {$\phantom{/}\pi\phantom{/}$};
\node[below] at (-.125,0) {$3\pi/4$};
\node[below] at (-.375,0) {$\pi/2$};
\node[below] at (-.625,0) {$\pi/4$};
\node[below] at (-.875,0) {$0$};
\node[below = .5cm] at (0,0) {$\arg(C_{II})$};
\end{scope}
\end{tikzpicture}

\caption{The trajectories of the duct modes for a range of (complex) porosity coefficients for case C described in table \ref{Tab:resSummary}, with $k_x=4$. For example, real values of $C_{II}$ are illustrated in \textcolor{red}{red}. The duct modes for rigid blades (i.e. $C_{II}=0$) are denoted by $\blacksquare$, the acoustic modes, $\lambda_m$, are labelled denoted by $\CIRCLE$ and the convected mode ($-\omega \delta$) is denoted by $\blacktriangle$. The upper half plane $\UHP$ is shaded in red and the lower half plane $\LHP$ is shaded in blue.}
\label{Fig:dmTraj}
\end{figure}

It is instructive to inspect the asymptotic forms of the duct modes for small and large values of porosity. For small porosity coefficients ($C_{II}\ll1$,  $\arg(C_{II}) \neq \pm \pi/2$), the roots are a small perturbation away from the rigid duct modes:
\begin{align}
\theta_n^\pm &=  \hat{\theta}_n^\pm + C_{II} \frac{\i (\omega \delta + \hat{\theta}_n^\pm)}{s(1+\delta_{0,n})\hat{\theta}_n^\pm} \left(1-(-1)^n \cos (d \hat{\theta}_n^\pm + \miBP) \right) + o(C_{II}^2). \label{Eq:smallApprox}
\end{align}
Conversely, for large porosity coefficients ($C_{II}\gg 1$, $\arg(C_{II}) \neq \pm \pi/2$), the duct modes are a small perturbation away from either the hydrodynamic mode or the acoustic modes:
\begin{align}
\addtocounter{equation}{1}
\theta_0 &= -\omega \delta + \frac{1}{C_{II}} \cdot  \frac{\i \zeta(-\omega \delta) \sin(s \zeta(-\omega \delta)}{\cos(s \zeta(-\omega \delta)) - \cos(-d\omega \delta \miBP)}  + o(C_{II}^{-1}),\tag{\theequation.a} \label{Eq:largeApproxWake}&\\
\theta_n^\pm &= \lambda_n^\pm \pm \frac{1}{C_{II}} \cdot \frac{\i (\zeta_n^\pm)^2}{(\lambda_n^\pm + \omega \delta)\delta \sqrt{\omega^2 w^2 - f_n^2}} + o(C_{II}^{-1}),& n \neq 0. \tag{\theequation.b} \label{Eq:largeApprox}
\end{align}
\begin{figure}
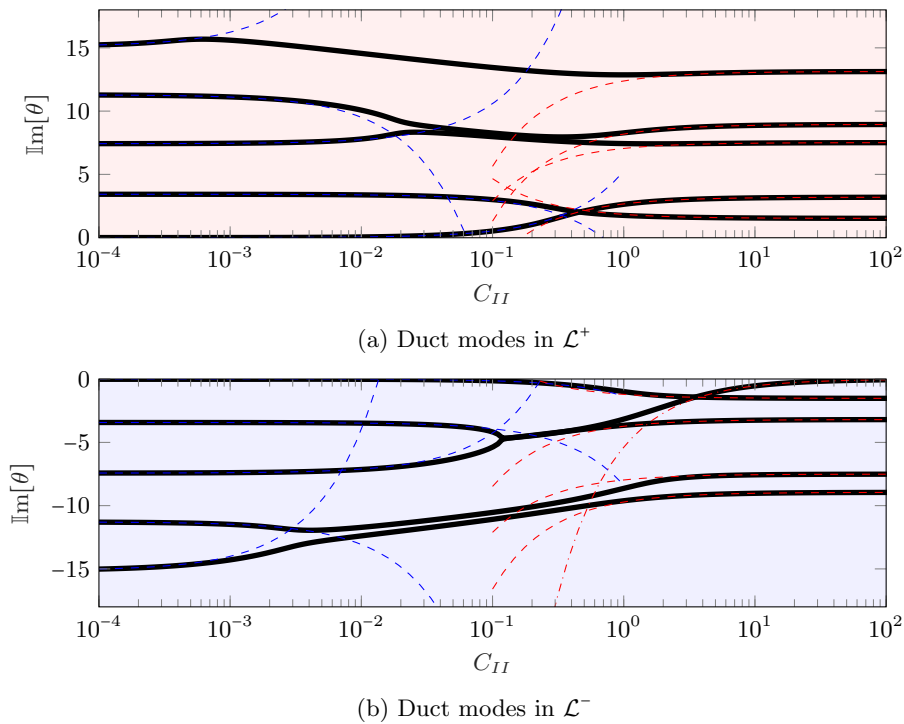

\centering
\begin{subfigure}{\linewidth}
\centering
\setlength{\fheight}{3cm}
\setlength{\fwidth}{.8\linewidth}
\input{images/asympU.tex}
\caption{Duct modes in $\UHP$}
\end{subfigure}

\begin{subfigure}{\linewidth}
\centering
\setlength{\fheight}{3cm}
\setlength{\fwidth}{.8\linewidth}
\input{images/asympL.tex}
\caption{Duct modes in $\LHP$}
\end{subfigure}

\caption{The imaginary parts of the duct modes as a function of (real) porosity coefficient $C_{II}$. The asymptotic approximations for small porosity coefficients \eqref{Eq:smallApprox} as illustrated by the dashed \textcolor{blue}{blue} lines and the asymptotic approximations for large porosity coefficients (\ref{Eq:largeApproxWake}, \ref{Eq:largeApprox}) are denoted by the dashed \textcolor{red}{red} lines. The cascade parameters are defined in case C in table \ref{Tab:resSummary} with $k_x=4$.}
\label{Fig:asympRoots}
\end{figure}
In figure \ref{Fig:asympRoots} we illustrate the imaginary part of the duct modes as a function of porosity coefficient, along with our asymptotic approximations. These approximations are particularly accurate for the duct modes that are initially close to the real axis, deteriorate for modes with large imaginary part. The duct modes undergo rapid changes as $C_{II}$ is increased to 1, which is suggestive that aeroacoustic gains can be made for modest porosity coefficients. This assertion is explored further in the following sections.

\begin{figure}
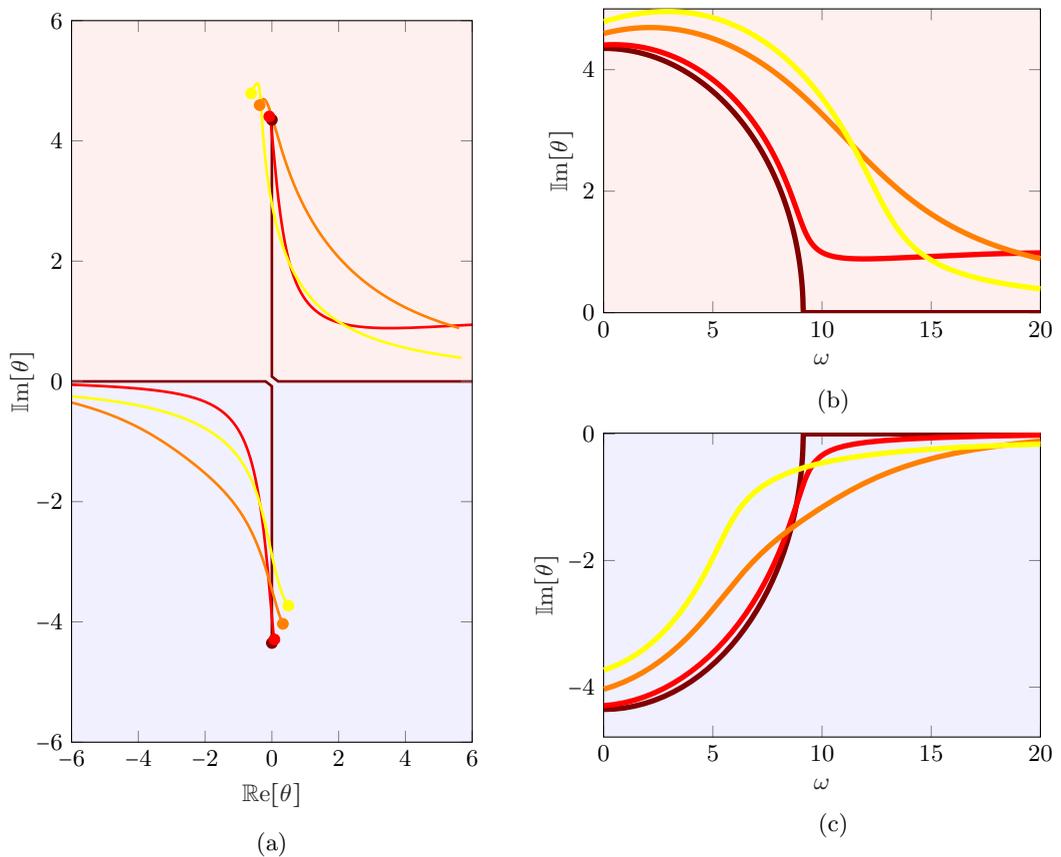

\centering
\begin{minipage}[!b]{.45\linewidth}
\begin{subfigure}[!b]{\linewidth}
\centering
\vspace{.2cm}
\setlength{\fheight}{9.5cm}
\setlength{\fwidth}{1.4\linewidth}
\input{images/mode-traj-freq.tex}
\caption{}
\end{subfigure}
\end{minipage}
\hfill
\begin{minipage}[!b]{.45\linewidth}
\begin{subfigure}[!b]{\linewidth}
\setlength{\fheight}{4cm}
\setlength{\fwidth}{\linewidth}
\input{images/im-part-1.tex}
\caption{}
\end{subfigure}
\begin{subfigure}[!b]{\linewidth}
\setlength{\fheight}{4cm}
\setlength{\fwidth}{\linewidth}
\definecolor{mycolor1}{rgb}{1.00000,1.00000,0.00000}%
\definecolor{mycolor2}{rgb}{0.94118,0.94118,1.00000}%
\begin{tikzpicture}[%
trim axis left, trim axis right
]

\begin{axis}[%
width=0.951\fwidth,
height=\fheight,
at={(0\fwidth,0\fheight)},
scale only axis,
xmin=0,
xmax=20,
xlabel style={font=\color{white!15!black}},
xlabel={$\omega$},
ymax=0.0100000000000002,
ylabel style={font=\color{white!15!black}},
ylabel={$\Im[\theta]$},
axis background/.style={fill=mycolor2},
ylabel shift = -.1cm
]
\addplot [color=black!50!red, line width=2.0pt, forget plot]
  table[row sep=crcr]{%
0.0110045022511258	-4.35079347657107\\
0.481216108054028	-4.34475792496831\\
0.941423211605802	-4.32763925975711\\
1.40163031515758	-4.29929658990709\\
1.86183741870935	-4.25950593883455\\
2.31204002001001	-4.20919151718696\\
2.75223811905953	-4.14870899043828\\
3.18243171585793	-4.07834057015137\\
3.6026208104052	-3.99829656153946\\
4.01280540270135	-3.90871417262409\\
4.41298549274637	-3.80965330827468\\
4.79315657828915	-3.7040286398584\\
5.16332316158079	-3.58944301543454\\
5.51348074037018	-3.46934493338832\\
5.85363381690846	-3.34059628228828\\
6.17377788894447	-3.20725010493654\\
6.47391295647824	-3.07012682724525\\
6.75403901950975	-2.93009055117359\\
7.02416058029015	-2.78235288989388\\
7.27427313656828	-2.63259451124426\\
7.50437668834417	-2.48182156011647\\
7.71447123561781	-2.33118101633275\\
7.91456128064032	-2.1737723318833\\
8.09464232116058	-2.01780947450826\\
8.25471435717859	-1.86497026119942\\
8.40478189094547	-1.70622666229318\\
8.5348404202101	-1.55294257977344\\
8.65489444722361	-1.3941564394415\\
8.75493946973487	-1.24447323217725\\
8.844979989995	-1.09058236329757\\
8.925016008004	-0.931056700511782\\
8.98504302151076	-0.789355862861282\\
9.03506553276638	-0.647055923251333\\
9.07508354177088	-0.504325699937635\\
9.10509704852426	-0.361403138854417\\
9.12510605302651	-0.218857033635839\\
9.14511505752876	-3.15390238370128e-08\\
20	1.0706294517604e-08\\
};
\addplot [color=red, line width=2.0pt, forget plot]
  table[row sep=crcr]{%
0.0110045022511258	-4.29243800807951\\
0.471211605802903	-4.27515808956714\\
0.931418709354677	-4.24665620371854\\
1.39162581290645	-4.20682643227788\\
1.85183291645823	-4.15547504120008\\
2.30203551775888	-4.09381390721792\\
2.74223361680841	-4.02223935664328\\
3.1724272136068	-3.94107701195863\\
3.59261630815408	-3.8505846472236\\
4.00280090045023	-3.75095267576427\\
4.40298099049525	-3.64230219448215\\
4.79315657828915	-3.52468032280661\\
5.16332316158079	-3.40154561083456\\
5.52348524262131	-3.26996327259017\\
5.86363831915958	-3.13396012780222\\
6.19378689344672	-2.98986100347103\\
6.50392646323161	-2.84234880202566\\
6.80406153076538	-2.68693865212197\\
7.0841875937969	-2.52909808844548\\
7.34430465232616	-2.36984659710571\\
7.5944172086043	-2.2034264868639\\
7.82452076038019	-2.03699622344271\\
8.04461980990495	-1.86397374062874\\
8.25471435717859	-1.68436699668216\\
8.4548044022011	-1.4987422593884\\
8.66489894947474	-1.28834674373841\\
9.19513756878439	-0.746138131420743\\
9.3151915957979	-0.647541389470049\\
9.42524112056028	-0.571252851238629\\
9.54529514757379	-0.502546449767109\\
9.66534917458729	-0.44672387816863\\
9.80541220610305	-0.394447043199964\\
9.96548424212106	-0.347390462192728\\
10.1555697848924	-0.304035001858157\\
10.3856733366683	-0.263997094213067\\
10.6657993996998	-0.227382168333872\\
11.0159569784892	-0.193534932389763\\
11.4661595797899	-0.161971506310966\\
12.0564252126063	-0.13260787186535\\
12.8467808904452	-0.105337603874695\\
13.9272671335668	-0.0800658968590824\\
15.4479514757379	-0.0565443772009928\\
17.6189284642321	-0.0350066819083708\\
20	-0.0198342057640843\\
};
\addplot [color=orange, line width=2.0pt, forget plot]
  table[row sep=crcr]{%
0.0110045022511258	-4.03196620682935\\
0.471211605802903	-3.98082542553655\\
0.921414207103552	-3.91961195802129\\
1.36161230615308	-3.84859145432187\\
1.79180590295148	-3.76800886699908\\
2.22199949974987	-3.6758218260395\\
2.64218859429715	-3.57400085991011\\
3.0523731865933	-3.46282567146645\\
3.45255327663832	-3.34268318637155\\
3.85273336668334	-3.21071116697881\\
4.26291795897949	-3.06322598205291\\
4.69311155577789	-2.8961430306128\\
5.19333666833417	-2.68917272738067\\
6.40388144072036	-2.18239990806713\\
6.80406153076538	-2.03098315964008\\
7.19423711855928	-1.89539842899514\\
7.60442171085543	-1.76508152356793\\
8.05462431215608	-1.63429955204536\\
8.56485392696348	-1.49814665757585\\
9.16512406203102	-1.3499288804431\\
9.8654392196098	-1.18886015418658\\
10.595767883942	-1.03236946746059\\
11.2760740370185	-0.897888036883014\\
11.9063576788394	-0.784569620138992\\
12.5166323161581	-0.686148373942757\\
13.1369114557279	-0.597581378289735\\
13.787204102051	-0.516388376994556\\
14.4775147573787	-0.441836353582932\\
15.2278524262131	-0.372463507580516\\
16.0482216108054	-0.308244570594837\\
16.9486268134067	-0.249277482141824\\
17.9490770385193	-0.195295058502868\\
19.0595767883942	-0.146899144356514\\
20	-0.113671575140021\\
};
\addplot [color=mycolor1, line width=2.0pt, forget plot]
  table[row sep=crcr]{%
0.0110045022511258	-3.73047688286521\\
0.411184592296149	-3.68066337970667\\
0.801360180090047	-3.62057995399878\\
1.18153126563282	-3.55056055025821\\
1.55169784892446	-3.47084574895684\\
1.91185992996498	-3.38159715079115\\
2.25201300650325	-3.28590154546384\\
2.58216158079039	-3.18153659595597\\
2.90230565282641	-3.06857314744148\\
3.2124452226113	-2.94707671953125\\
3.51258029014507	-2.81713792252308\\
3.80271085542772	-2.67891880588071\\
4.08283691845923	-2.5327273132595\\
4.35295847923962	-2.37914244151269\\
4.62308004002001	-2.21284493873525\\
4.91321060530265	-2.02090269471357\\
5.33339969984992	-1.72716214302764\\
5.6735527763882	-1.4946766114313\\
5.89365182591296	-1.35809073856785\\
6.09374187093547	-1.2475407359378\\
6.28382741370685	-1.15517737843067\\
6.48391745872937	-1.07037979767226\\
6.70401650825413	-0.989873757253459\\
6.94412456228114	-0.914649958164141\\
7.2042416208104	-0.845028263673239\\
7.50437668834417	-0.776709738248936\\
7.84452976488244	-0.711295405242797\\
8.23470535267634	-0.648246585079146\\
8.67490345172586	-0.588766753238545\\
9.18513306653326	-0.53147086392503\\
9.77539869934967	-0.476871510926923\\
10.4657093546773	-0.424815342830765\\
11.2660695347674	-0.376157633978831\\
12.2064927463732	-0.330632885958565\\
13.3269969984992	-0.288136810490009\\
14.6676003001501	-0.249067773840085\\
16.2883296648324	-0.213603255395867\\
18.2592166083042	-0.182198677395476\\
20	-0.161997248054501\\
};
\end{axis}
\end{tikzpicture}%
\caption{}
\end{subfigure}
\end{minipage}

\caption{The trajectories of a pair of duct modes as a function of frequency for a range of porosity constants. The duct modes for $\omega=0$ are denoted by $\CIRCLE$. The real and complex trajectories are plotted in figure (a). The imaginary part of the modes in the upper and lower half planes are plotted in figures b and c respectively. The colour represent different porosity coefficients: \textcolor{brown}{$C_{II}=0$} (i.e. rigid), \textcolor{red}{$C_{II}=0.01$}, \textcolor{orange}{$C_{II}=0.1$}, and \textcolor{yellow!10!black}{$C_{II}=0.1$}. The upper half plane $\UHP$ is shaded in red and the lower half plane $\LHP$ is shaded in blue. The relevant parameters correspond to case D in table \ref{Tab:resSummary} with $k_x=2$.}
\label{Fig:freq}
\end{figure}

Further insight may be gained by examining the duct modes as a function of frequency for fixed porosity values. Figure \eqref{Fig:freq} illustrates the duct modes for four porosity coefficients at a range of frequencies. For each porosity value, the imaginary part of the mode decreases as the frequency is increased. However, in contrast for the rigid case, the imaginary part never vanishes for non-zero porosity values. In some cases (in $\UHP$), the imaginary part of the duct modes undergoes a slight increase before decreasing towards the real axis. The role of porosity is particularly important for small to moderate frequencies. As the frequency is increased, the difference between the modes reduces.

We now consider the effect of porosity on sound generation and sound transmission.
\subsection{Sound Generation}
Sound generation is caused when a pressure-free gust (i.e. $ k_x = \omega$) interacts with the blade row, resulting in the production of pressure waves. In order to enable comparison against prior works, we consider cases analysed by \cite{Glegg1999} and \cite{Posson2010} as defined in table \ref{Tab:resSummary}. %
\subsubsection{Unsteady Lift}
During the solution to the Wiener--Hopf problem associated with the scattering by a blade row with complex boundaries, we observed that the major difference with the rigid case is that the duct modes are modified. Consequently, we expect complex boundary conditions to have a significant effect of the unsteady loading of the blades. In this section we test that hypothesis.

The analytic expression for $D$ \eqref{c5:Eq:Dsol} enables the swift calculation of the unsteady loading on the blades. The unsteady loading is defined as the integral of the unsteady pressure over the blade surface:
\begin{align}
C_p&=\frac{1}{2\pi w_0 } \int_{0}^{2} \Delta_0 [p](x) \d  x. \label{c5:Eq:lift}
\end{align}
Integration by parts and application of the boundary conditions \eqref{c5:Eq:upBC2} and \eqref{c5:Eq:npjWake2} yields the identity
\begin{align*}
D(-\omega \delta M^2) & =  \frac{1}{2 \pi }\int_{-\infty}^{\infty} \Delta_0 [\phi](x) \e^{-\i \omega \delta M^2 x} \d x= - \frac{1}{\i \omega 2 \pi }\int_{0}^{2} \frac{\partial}{\partial x} \left( \Delta_0 [\phi](x) \e^{-\i \omega \delta x} \right) \e^{\i \omega x} \d x.
\end{align*}
Consequently, the normalised unsteady lift \eqref{c5:Eq:lift} may be written as
\begin{align*}
C_p & = \frac{- \i \omega }{\namp} D(-\omega \delta M^2).
\end{align*}
\setlength{\fheight}{8cm}
\setlength{\fwidth}{.8\linewidth}
\begin{figure}
\centering
\input{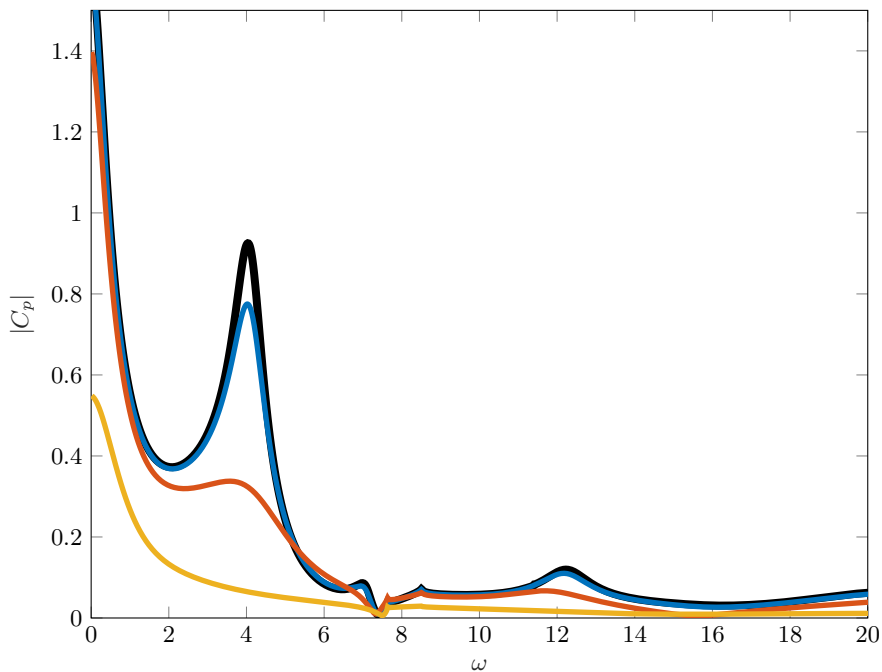}
\hfil
\caption[Unsteady lift for a range of frequencies and porosities.]{Unsteady lift for a range of frequencies and porosities. The aerodynamic and aeroacoustic parameters are defined in table \ref{Tab:resSummary} and correspond to those in figure 3 of \cite{Glegg1999}. The colours correspond to the porosity parameters \textcolor{black}{\bm{$C_{II}=0$}} (i.e. rigid), \textcolor{matlab1}{\bm{$C_{II}=.01$}}, \textcolor{matlab2}{\bm{$C_{II}=.1$}} and \textcolor{matlab3}{\bm{$C_{II}=1$}}.}
\label{Fig:ulGlegg}
\end{figure}
{\flushleft%
The modified boundary conditions have a strong effect on the unsteady loading, as illustrated in figure \ref{Fig:ulGlegg}. The unsteady loading for a rigid cascade is compared against the loading for a range of porosity parameters, which correspond to the $C_{II}$ values. The results indicate that the effect of the modified boundary conditions is to shift the locations of the duct modes, as indicated by the shifts in the local maximum around $\omega \approx 12$, which has previously been identified with the cut-on frequency of the duct mode \citep{Glegg1999}. As $C_{II}$ increases, the pressure jump across the blade must decrease in accordance with Darcy's law \eqref{c5:Eq:darcyFlow}, thus ensuring that the seepage velocity through the blade is proportional to the pressure jump across the blade. This is observed in figure \ref{Fig:ulGlegg}, where the effects of increasing the porosity result in an almost uniform reduction in the unsteady lift.}

The impact of porosity on unsteady lift appears most pronounced when the frequency is low. In particular, the range $0<\omega<7.62$ shows the largest reduction in lift. After the acoustic mode is cut-on at $\omega \approx 7.62$, the reductions in lift are generally smaller, but this is possibly because the lift on the rigid cascade is also reduced.

Another reason for the reduction in unsteady lift is that porosity reduces the build-up of pressure in the inter-blade duct by allowing the flow field to dissipate energy through the blades. Accordingly, modes are less likely to become ``trapped'' in the inter-blade region, since fluid is now permitted to travel across the blades. Consequently, the unsteady lift is greatly reduced since the cascade produces a lower pressure jump across each blade.

\subsubsection{Sound Power Output}
Analytical expressions for the sound power output are available by a similar method to \cite{Glegg1999}. The modal upstream or downstream sound power output for the $m$-th mode is given by
\begin{align*}
W_{\pm}(m) & = \frac{\omega \pi^2}{\Delta } \Re \left[ \frac{\left| \zeta_m^\pm D(\aModePM) \right|^2}{\sqrt{\omega^2 w^2 - f_m^2}}\right].
\end{align*}
As noted multiple times in this paper, modifications to the surface boundary conditions do not affect the acoustic modes $\aModePM$. Consequently, the cut-on frequencies of these upstream and downstream modes are unaffected by porosity, for example. This is observed in figure \ref{Fig:dspoa}, where the downstream sound power output for the first, second and third modes are illustrated for a large frequency range. Clearly the modes are cut-on at the same frequency, but the magnitude of the sound power output is strongly affected by porosity. For small porosities ($C_{II}=0.01,\, 0.1$) there is little impact on the sound power output of the first mode until the channel modes become cut-on at $\omega \approx 12$. Following this cut-on frequency, we observe a large decrease in the sound power for all modes. Similarly to the unsteady lift, this reduction in sound power output can be attributed to the reduction in the pressure jump across the blade caused by porosity, thereby reducing the scattered sound in the upstream and downstream regions.

We define the sound power level as 
\begin{align*}
L_{W_\pm}(m) & = 10 \log_{10} \left(\frac{W_\pm(m)}{W_{r\pm}(m)}\right) \textnormal{ dB},
\end{align*}
where we take the reference sound power $W_{r\pm}$ to be the sound power for a rigid blade row. The reduction in sound power level is illustrated in figure \ref{Fig:dspl} for a rage of porosity parameters. We observe that even a modest porosity of $C_{II}=0.1$ is capable of large sound power reductions of 5 dB for the first mode and 20 dB for the second mode. 
\setlength{\fheight}{6cm}
\setlength{\fwidth}{.8\linewidth}
\begin{figure}
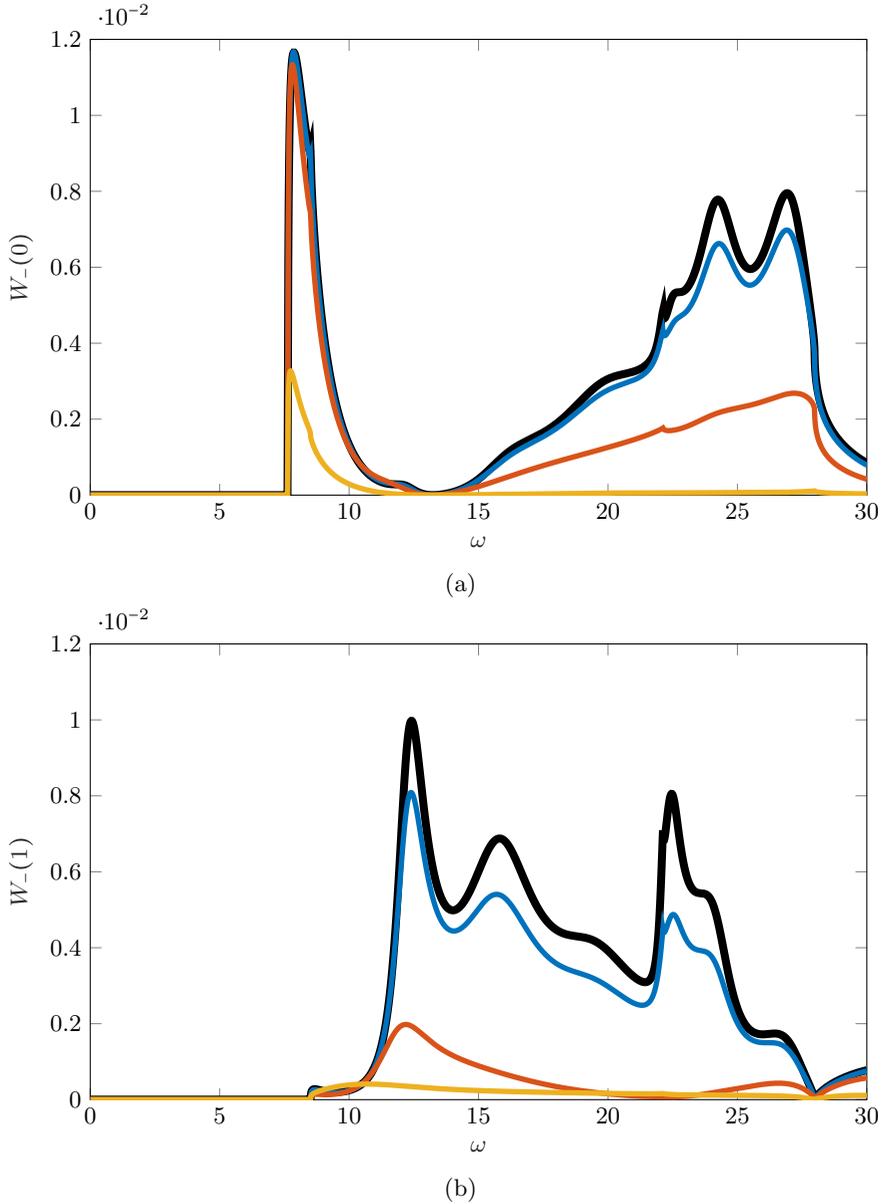

\begin{subfigure}[!h]{\linewidth}
\centering
\setlength{\fwidth}{.8\linewidth}

\input{images/spo1-glegg}
\caption{}
\label{Fig:dspoa}
\end{subfigure}
\begin{subfigure}[!h]{\linewidth}
\setlength{\fwidth}{.8\linewidth}
\centering

\input{images/spo2-glegg}
\caption{}
\label{Fig:dspob}
\end{subfigure}
\caption[Modal downstream sound power output for a cascade of porous blades at a range of frequencies.]{Modal downstream sound power output for a cascade of porous blades at a range of frequencies for (a) the first mode ($m=0$) and (b) the second mode ($m=1$). The aerodynamic and aeroacoustic parameters are defined in table \ref{Tab:resSummary} and correspond to those in figure 9 of \cite{Glegg1999}. In the colours correspond to the porosity parameters \textcolor{black}{\bm{$C_{II}=0$}} (i.e. rigid), \textcolor{matlab1}{\bm{$C_{II}=0.01$}}, \textcolor{matlab2}{\bm{$C_{II}=0.1$}} and \textcolor{matlab3}{\bm{$C_{II}=1$}.}}
\label{Fig:dspo}
\end{figure}

\setlength{\fheight}{6cm}
\setlength{\fwidth}{.8\linewidth}
\begin{figure}
\begin{subfigure}[!h]{\linewidth}
\centering
\setlength{\fwidth}{.8\linewidth}

\input{images/spl1-glegg}
\caption{}
\label{Fig:spl1}
\end{subfigure}
\begin{subfigure}[!h]{\linewidth}
\setlength{\fwidth}{.8\linewidth}
\centering

\input{images/spl2-glegg}
\caption{}
\label{Fig:spl2}
\end{subfigure}
\caption[Modal downstream sound power level for a cascade of porous blades at a range of frequencies.]{Modal downstream sound power level for a cascade of porous blades at a range of frequencies for (a) the first mode ($m=0$) and (b) the second mode ($m=1$). The aerodynamic and aeroacoustic parameters are defined in table \ref{Tab:resSummary} and correspond to those in figure 9 of \cite{Glegg1999}. In the colours correspond to the porosity parameters \textcolor{matlab1}{\bm{$C_{II}=0.01$}}, \textcolor{matlab2}{\bm{$C_{II}=0.1$}} and \textcolor{matlab3}{\bm{$C_{II}=1$}}.}
\label{Fig:dspl}
\end{figure}

\subsection{Sound Transmission}
We now use our model to investigate sound transmission through a cascade of flat plates with complex boundaries. Sound transmission occurs when a sound wave interacts with a cascade of blades, and is reflected and transmitted through the blade row. In order to explore the effects of modified boundary conditions on the transmission and reflection, we follow \cite{Bouley2017} and write the incident acoustic field \eqref{c5:Eq:incFieldNon} in the form
\begin{align*}
\phi_i & = \e^{\i \zeta_0^+ y} \e^{-\i \aModePv{0} x}, \qquad \qquad x \leq y d/h,
\end{align*}
where the mode  is assumed cut-on, so that $\aModePv{0}$ is real. \iffalse Comparison with the incident field \eqref{c5:Eq:incFieldNon} implies that the frequency satisfies
\begin{align*}
\delta k_x &= -\aModeMv{m}, \\
\omega k_y &= - \zModeMv{m}.
\end{align*}
We set the frequency in the $y$ direction and solve for $k_x$. The inter-blade phase angle must satisfy
\begin{align*}
\miBP & = -\zModeMv{m} \vPS  \pm  \hPS \sqrt{\omega^2 w^2 - \frac{(\zModeMv{m}\vPS + 2 m \pi)^2}{\vPS^2}}
\end{align*}
Then, $k_x$ is obtained by use of the inter-blade phase angle
\begin{align*}
k_x & = \frac{\miBP - \omega k_y \vPS}{\delta \hPS}\\
&= \frac{\miBP + \zModeMv{m} \vPS}{\delta \hPS}\\
&=  \pm  \frac{1}{\delta} \sqrt{\omega^2 w^2 - \frac{(\zModeMv{m}\vPS + 2 m \pi)^2}{\vPS^2}}.
\end{align*}
We take the positive root.
\fi
Comparison with the solution in section \ref{c5:Sec:FourInversion} shows that the reflected acoustic field takes the form
\begin{align*}
\phi_r & = \sum_{m = -\infty}^{\infty} R_m \e^{\i \zeta_m^+ y} \e^{-\i \aModeP x}, \qquad \qquad x \leq y d/h,
\end{align*}
and the transmitted acoustic potential takes the form
\begin{align*}
\phi_t & = \sum_{m = -\infty}^{\infty} T_m \e^{-\i \zeta_m^- y} \e^{-\i \aModeM (x-2)}, \qquad \qquad x \leq 2 + y d/h,
\end{align*}
where $\boldsymbol{R}$ and $\boldsymbol{T}$ are the vectors of the reflection and transmission coefficients respectively. Note that the mode corresponding to the jump in acoustic potential across the wake ($\delta \omega$) is not included in this analysis since it does not contribute to the pressure field.

By comparison with \eqref{c5:Eq:solU} and \eqref{c5:Eq:solD} we obtain the expressions for the transmission and reflection coefficients as 
\begin{align*}
R_m & =  + \frac{\pi \zeta_m^+ D^{(1,3)}(\aModeP)}{\sqrt{ \omega^2 w^2 - f_m^2}}, \qquad T_m =  -\frac{\pi \zeta_m^{-} D^{(2,4)}(\aModeM) } {\sqrt{\omega^2 w^2 - f_m^2}}\e^{-2 \i \aModeM}.
\end{align*}

We now perform a parametric study on the effects of porosity on the transmission and reflection coefficients. In figure \ref{Fig:tr1}, we plot the total pressure field and the associated amplitude of the transmission and reflection coefficients. The blades are rigid ($C_{II}=0$) in figures \ref{Fig:0pres} and \ref{Fig:0tr}. In this case, two modes are cut-on, as indicated in figure \ref{Fig:0tr}. When the cascade is rigid, the blade row has a strong effect on the reflected and transmitted fields, which can be observed by the significant distortion of the acoustic field in figure \ref{Fig:0pres}. In figure \ref{Fig:.5pres}, the porosity is increased to the relatively modest value of $C_{II}=0.5$. We now observe that the amplitude of the reflected wave substantially decreases, so that the upstream field seems almost unperturbed by the cascade. In contrast, the downstream field is still distorted by the cascade, and the presence of the cut-on modes can still be observed. We note that the two modes are still cut-on, and the porosity does not affect their cut-on frequency. This is consistent with the argument made in the Wiener--Hopf analysis, namely that modifying the boundary conditions does not affect the modal structure of the far-field solution. Finally, when the porosity is increased to $C_{II}=3$, the acoustic field is essentially unaffected by the blade row, as seen in figure \ref{Fig:1pres}. This is because the blades are sufficiently porous that the fluid can pass through the blade unhindered. Accordingly, the only significant transmission and reflection coefficient remaining is that of the incident perturbation, as shown in figure \ref{Fig:tr1}.

\newcounter{row}
\makeatletter
\@addtoreset{subfigure}{row}
\makeatother
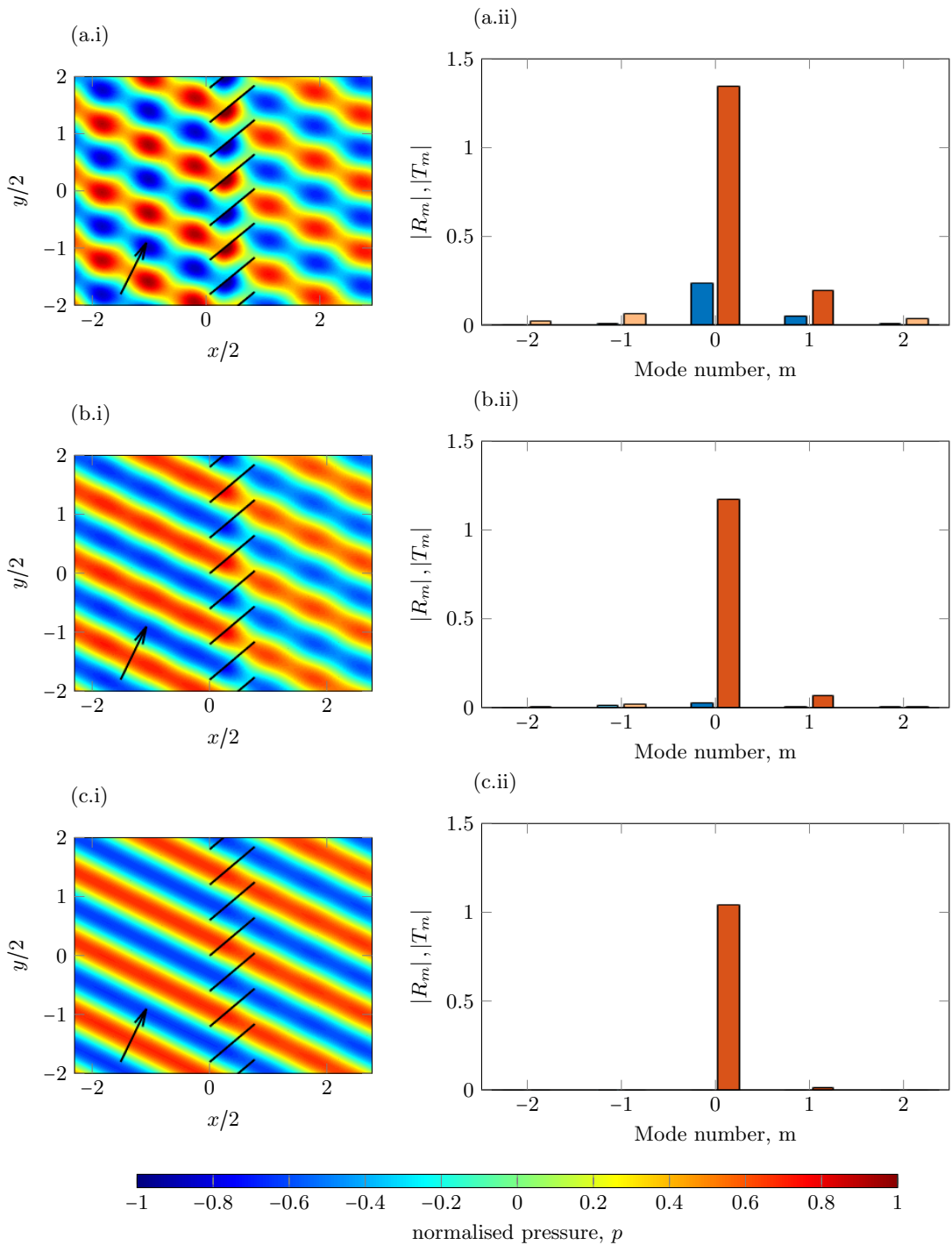
\begin{figure}
\captionsetup[subfigure]{position=top, singlelinecheck=off,justification=raggedright}
  \renewcommand{\thesubfigure}{\alph{row}.\roman{subfigure}}%
  \centering 
  \setcounter{row}{1}%
\begin{subfigure}[!h]{.35\linewidth}
\centering
\setlength{\fwidth}{.95\linewidth}
\caption{} \label{Fig:0pres}
 \begin{tikzpicture}[%
trim axis left, trim axis right
]
 \begin{axis}[
  set layers,
ylabel near ticks,
xlabel near ticks,
 width=\fwidth,
 height=0.76491\fwidth,
 at={(0,0)},
scale only axis,
enlargelimits=false,
axis on top,
xlabel={$x/2$},
ylabel={$y/2$}]
\addplot graphics[xmin=-2.3116,xmax=2.9178,ymin=-2,ymax=2] {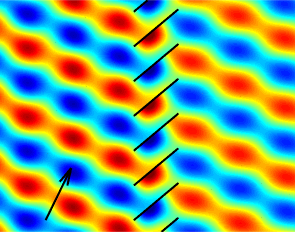};

    \end{axis}
    \end{tikzpicture}
\end{subfigure}
\hfill
\begin{subfigure}[!h]{.55\linewidth}
\setlength{\fheight}{4cm}
\setlength{\fwidth}{.95\linewidth}
\centering
\caption{} \label{Fig:0tr}
 \begin{tikzpicture}[%
trim axis left, trim axis right
]
 \begin{axis}[
  set layers,
ylabel near ticks,
xlabel near ticks,
 width=\fwidth,
 height=\fheight,
 at={(0,0)},
scale only axis,
enlargelimits=false,
axis on top,
xlabel={Mode number, m },
ylabel={$\left| R_m \right|, \left| T_m \right|$}]
\addplot graphics[xmin=-2.4857,xmax=2.4857,ymin=0,ymax=1.5] {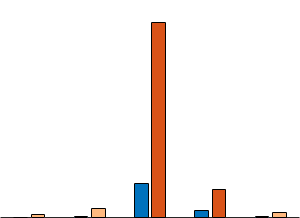};

    \end{axis}
    \end{tikzpicture}
\end{subfigure}

\stepcounter{row}
\begin{subfigure}[!h]{.35\linewidth}
\centering
\caption{}\label{Fig:.5pres}
\setlength{\fwidth}{.95\linewidth}
 \begin{tikzpicture}[%
trim axis left, trim axis right
]
 \begin{axis}[
  set layers,
ylabel near ticks,
xlabel near ticks,
 width=\fwidth,
 height=0.78871\fwidth,
 at={(0,0)},
scale only axis,
enlargelimits=false,
axis on top,
xlabel={$x/2$},
ylabel={$y/2$}]
\addplot graphics[xmin=-2.3005,xmax=2.7711,ymin=-2,ymax=2] {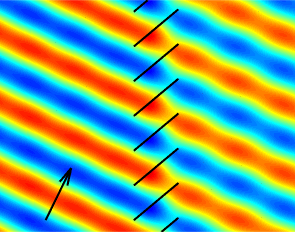};

    \end{axis}
    \end{tikzpicture}
\end{subfigure}
\hfill
\begin{subfigure}[!h]{.55\linewidth}
\setlength{\fwidth}{.95\linewidth}
\centering
\caption{} \label{Fig:.5tr}
\setlength{\fheight}{4cm}
 \begin{tikzpicture}[%
trim axis left, trim axis right
]
 \begin{axis}[
  set layers,
ylabel near ticks,
xlabel near ticks,
 width=\fwidth,
 height=\fheight,
 at={(0,0)},
scale only axis,
enlargelimits=false,
axis on top,
xlabel={Mode number, m },
ylabel={$\left| R_m \right|, \left| T_m \right|$}]
\addplot graphics[xmin=-2.4857,xmax=2.4857,ymin=0,ymax=1.5] {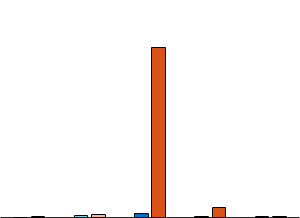};

    \end{axis}
    \end{tikzpicture}
\end{subfigure}

\stepcounter{row}
\begin{subfigure}[!h]{.35\linewidth}
\centering
\caption{} \label{Fig:1pres}
\setlength{\fwidth}{.95\linewidth}
 \begin{tikzpicture}[%
trim axis left, trim axis right
]
 \begin{axis}[
  set layers,
ylabel near ticks,
xlabel near ticks,
 width=\fwidth,
 height=0.78871\fwidth,
 at={(0,0)},
scale only axis,
enlargelimits=false,
axis on top,
xlabel={$x/2$},
ylabel={$y/2$}]
\addplot graphics[xmin=-2.3005,xmax=2.7711,ymin=-2,ymax=2] {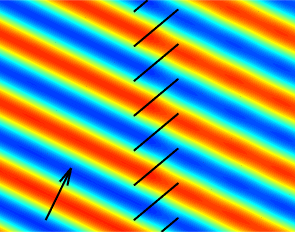};

    \end{axis}
    \end{tikzpicture}
\end{subfigure}
\hfill
\begin{subfigure}[!h]{.55\linewidth}
\setlength{\fwidth}{.95\linewidth}
\centering
\caption{} \label{Fig:1tr}
\setlength{\fheight}{4cm}
 \begin{tikzpicture}[%
trim axis left, trim axis right
]
 \begin{axis}[
  set layers,
ylabel near ticks,
xlabel near ticks,
 width=\fwidth,
 height=\fheight,
 at={(0,0)},
scale only axis,
enlargelimits=false,
axis on top,
xlabel={Mode number, m },
ylabel={$\left| R_m \right|, \left| T_m \right|$}]
\addplot graphics[xmin=-2.4857,xmax=2.4857,ymin=0,ymax=1.5] {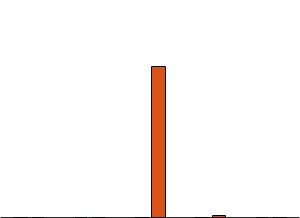};

    \end{axis}
    \end{tikzpicture}
\end{subfigure}
\centering
\begin{tikzpicture}
\begin{axis}[
    hide axis,
    scale only axis,
    height=0pt,
    width=0pt,
    colormap/jet,
    colorbar horizontal,
    point meta min=-1,
    point meta max=1,
    colorbar style={
        width=.85\textwidth,
        height=.2cm,
        xlabel = {normalised pressure, $p$},
        x label style={anchor=north},
        xtick={-1,-.8,...,1},
        xticklabel style={
	        	align=center,
				/pgf/number format/.cd,
	   	        fixed,
}
    }]
    \addplot [draw=none] coordinates {(0,0)};
\end{axis}
\end{tikzpicture}
\renewcommand\subfigure{\thefigure}
\caption[The total pressure field and reflection and transmission coefficients for a cascade of porous aerofoils.]{Left: the total (incident and scattered), normalised pressure field. Right: the amplitudes of the normalised reflection (\textcolor{matlab1}{\textbf{blue}}) and transmission (\textcolor{matlab2}{\textbf{orange}}) coefficients. The $C_{II}$ values correspond to $C_{II}=0$ (a), $C_{II}=0.5$ (b) and $C_{II}=1$ (c). The darker bars indicate modes that are cut-on whereas the lighter bars are cut-off. The parameters in this case are defined in case E of table \ref{Tab:resSummary}. The arrow indicates the direction of the incident wave.}
\label{Fig:tr1}
\end{figure}

In general, blade porosity has a strong effect on the amplitude of the reflection and transmission coefficients. In figure \ref{Fig:TRbehavioura}, the amplitudes of the transmission and reflection coefficients are plotted for a range of porosities. We observe a rapid decay in these amplitudes in the region $10^{-2}<C_{II}<1$ . When plotted on a log-scale in figure  \ref{Fig:TRbehaviourb}, the behaviour resembles an exponential dependence. This is a promising result for sound reduction technologies, as it suggests that the sound transmission for blade rows can be drastically altered with small amounts of porosity.

\begin{figure}
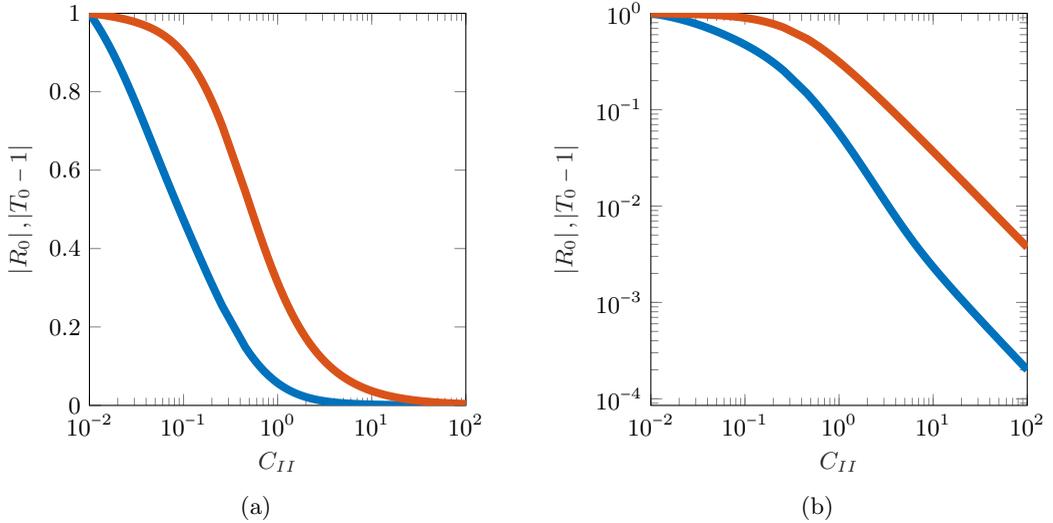

\centering
\begin{subfigure}[!h]{.45\linewidth}
\centering

\setlength{\fwidth}{.85\linewidth}
\setlength{\fheight}{\fwidth}
\centering
\input{images/trc-log.tex}
\caption{}
\label{Fig:TRbehavioura}
\end{subfigure}
\hfill
\begin{subfigure}[!h]{.45\linewidth}
\centering

\setlength{\fwidth}{.85\linewidth}
\setlength{\fheight}{\fwidth}
\input{images/trc.tex}
\caption{}
\label{Fig:TRbehaviourb}
\end{subfigure}
\caption[Reflection and transmission coefficients as a function of porosity coefficient, $C_{II}$.]{The normalised first reflection (\textcolor{matlab1}{\textbf{blue}}) and transmission (\textcolor{matlab2}{\textbf{orange}}) coefficients as a function of porosity coefficient, $C_{II}$, on (a) a linear scale, and (b) a log scale. The coefficients are normalised by their values for a rigid cascade. The cascade parameters are equivalent to those in figure \ref{Fig:tr1}.}
\label{Fig:TRbehaviour}
\end{figure}

\section{Conclusions}
\label{c5:Sec:Conclusion}
We have derived an analytic solution for the scattering of an unsteady perturbation incident on a rectilinear cascade of flat blades with complex boundaries. The analytic nature of the solution means that it is extremely rapid to compute, and offers physical insight into the role played by different boundary conditions. In contrast with previous studies that focussed on the effects of rigid blades \citep{Glegg1999,Posson2010}, the formulation of the present research allows a range of boundary conditions to be studied with minimal effort, such as porosity, compliance and flow impedance. In terms of the spectral plane, the effect of modifying these boundary conditions is to change the locations of the zeros of the Wiener--Hopf kernel, whilst the poles are unchanged. Accordingly, the modal structure of the far-field scattered pressure is invariant under modifications to the flat blades' boundary conditions. Conversely, the modal structure of the near-field region undergoes large deformations since the zeros of the kernel correspond to the duct modes of the inter-blade region. This has a strong effect on the surface pressure fluctuations and unsteady loading, which has implications for the far-field sound. 

We have particularly investigated the role of blade porosity in the form of a Darcy-type condition. The results show that substantial reductions in both the unsteady lift and sound power output are possible for even modest values of porosity. At low frequencies, we observe a significant change in the unsteady loading and a moderate effect on the sound power output. Conversely, at high frequencies we observe a significant effect on the sound power output and a small effect on the unsteady loading. Furthermore, the amplitudes of the reflection and transmission coefficients rapidly decrease as the blade porosity is increased. We attribute these considerable reductions to several physical mechanisms associated with porosity. Increasing the blade porosity promotes a compliant relationship between adjacent channels, which prevents the build-up of trapped modes. Moreover, the flow seepage afforded by porous channels permits the dissipation of energy through the blades, thus reducing unsteady pressure fluctuations on the blade. This reduction in pressure fluctuations corresponds to reductions in the scattered pressure and far-field sound. Finally, this study shows that modified boundary conditions have a large impact on gust-cascade interaction noise and offers a useful design tool that can model aeroacoustic and aeroelastic effects.

\begin{appendices} 
\renewcommand{\thesection}{\Alph{section}}
\renewcommand\thefigure{\thesection.\arabic{figure}}    
\renewcommand\theequation{\thesection.\arabic{equation}}
\section{Wiener--Hopf Solution} \label{c5:Sec:WHsol}

We now solve the integral equation \eqref{c5:Eq:intEq} subject to the boundary conditions of no discontinuities upstream \eqref{c5:Eq:upBC}, a modified no-flux condition \eqref{c5:Eq:generalBC} and no pressure jump across the wake \eqref{c5:Eq:npjWake}. In a similar way to \cite{Glegg1999}, we split this problem into four coupled problems that are amenable to the Wiener--Hopf method.\iffalse This analysis is different to that of chapter \ref{Chap:cascAcou1}: firstly, the modified boundary conditions must be accounted for and we show how to do this in a concise manner. Secondly, chapter \ref{Chap:cascAcou1} was restricted to the case of a pressure-free gust ($k_x = \omega$), whereas the analysis in the present chapter is for a more general sound wave. Consequently, we present the full solution to the Wiener--Hopf solution in this chapter, although there are similarities with that presented in chapter \ref{Chap:cascAcou1}.\fi

We write
\begin{align}
\Delta_0 \left[\ssig\right](x) &= \Delta_0 \left[\ssig^{(1)}\right](x)+\Delta_0 \left[ \ssig^{(2)} \right](x)+\Delta_0 \left[\ssig^{(3)} \right](x)+\Delta_0 \left[\ssig^{(4)} \right](x), \label{c5:Eq:ssigSplits}
\end{align}
and its Fourier transform
\begin{align}
\dsig(\gamma)&=\dsig^{(1)}(\gamma)+\dsig^{(2)}(\gamma)+\dsig^{(3)}(\gamma)+\dsig^{(4)}(\gamma), \label{c5:Eq:dsigSplits}
\end{align}
where each $\Delta_0 \left[\ssig^{(n)} \right]$ and $\dsig^{(n)}$ satisfy a semi-infinite integral equation of the form
\begin{align}
f^{(n)}(x)&= -4 \pi\int_{-\infty}^{\infty} D^{(n)}(\gamma)  j(\gamma) \e^{-\i \gamma x} \, \d \gamma, \label{c5:Eq:sigGEfn} %
\end{align}
for $n=1,2,3,4$. The corresponding boundary conditions are
{
	\setlength{\jot}{20pt}
	\begin{align}
	\addtocounter{equation}{1}
	\fsig^{(1)}(x)= &\mu_0 \Delta_0 \left[\phi^{(1)} \right](x) +\mu_1 \Delta_0 \left[\phi_x^{(1)}\right](x) + \mu_2 \Delta_0\left[\phi_{x,x}^{(1)}\right](x) \notag \\
	-&2 \namp \exp \left[\i  \delta k_x x\right], & x>0, \tag{\theequation.a} \label{c5:Eq:bc1a}\\
	\addtocounter{equation}{1}
	\fsig^{(2)}(x)= &\mu_0 \Delta_0 \left[\phi^{(2)} \right](x) +\mu_1 \Delta_0 \left[\phi_x^{(2)}\right](x) + \mu_2 \Delta_0\left[\phi_{x,x}^{(2)}\right](x), & x<2, \tag{\theequation.a} \label{c5:Eq:bc2a}\\
	  \addtocounter{equation}{1}
	\fsig^{(3)}(x)= &\mu_0 \Delta_0 \left[\phi^{(3)} \right](x) +\mu_1 \Delta_0 \left[\phi_x^{(3)}\right](x) + \mu_2 \Delta_0\left[\phi_{x,x}^{(3)}\right](x), & x>0, \label{c5:Eq:bc3a} \tag{\theequation.a}\\
	 \addtocounter{equation}{1}
	\fsig^{(4)}(x)=& \mu_0 \Delta_0 \left[\phi^{(4)} \right](x) +\mu_1 \Delta_0 \left[\phi_x^{(4)}\right](x) + \mu_2 \Delta_0\left[\phi_{x,x}^{(4)}\right](x), & x<2, \tag{\theequation.a} \label{c5:Eq:bc4a}
	\end{align}
}
and
{
	\setlength{\jot}{20pt}
	\begin{align}
			\addtocounter{equation}{-3}
	\Delta_0 \left[\ssig^{(1)} \right](x) =&0, &x<0, \tag{\theequation.b} \label{c5:Eq:bc1b} \\
	\addtocounter{equation}{1}
	 \Delta_0 \left[ \ssig^{(1)} \right] (x)+\Delta_0 \left[ \ssig^{(2)}\right](x)=&2 \pi \i \Psig^{(2)} \e^{\i \delta \omega  x}, & x> 2, \label{c5:Eq:bc2b} \tag{\theequation.b}\\
	  \addtocounter{equation}{1}
	 \Delta_0 \left[\ssig^{(2)} \right](x)+\Delta_0\left[ \ssig^{(3)} \right](x) + \Delta_0 \left[ \ssig^{(4)}\right](x)=&0, & x<0, \label{c5:Eq:bc3b} \tag{\theequation.b}\\ 
	 \addtocounter{equation}{1}
	 \Delta_0 \left[\ssig^{(3)}\right](x)+\Delta_0 \left[\ssig^{(4)}\right](x) =& 2 \pi \i \Psig^{(4)} \e^{\i \delta \omega x}, & x> 2, \label{c5:Eq:bc4b} \tag{\theequation.b}
	\end{align}
}
{
\flushleft
where $\Psig^{(2)}$ and $\Psig^{(4)}$ are two constants of integration that will be specified to enforce the Kutta condition.
Summing the four above conditions results in the original boundary conditions and, consequently, we may apply the Wiener--Hopf method to each semi-infinite integral equation and sum the resulting contributions to obtain a solution to the original equations.
}
\subsubsection{Solution to First Wiener--Hopf Equation -- $\dsig^{(1)}$} \label{c5:Sec:WHsig1}
In this section, we solve the integral equation \eqref{c5:Eq:sigGEfn} for $n=1$
\begin{align}
\addtocounter{equation}{1}
\fsig^{(1)}(x)&= -4 \pi \int_{-\infty}^{\infty} \dsig^{(1)}(\gamma)  j(\gamma)\e^{-\i \gamma x} \, \d \gamma, \label{c5:Eq:sigIE1} \tag{\theequation.a}
\end{align}
subject to \eqref{c5:Eq:bc1a} and \eqref{c5:Eq:bc1b}.
\iffalse \eqref{c5:Eq:sigBCf1}:
\begin{align*}
\fsig^{(1)}(x)&=  - A k_y \exp \left[\i k_x x\right] + B \left(\i \omega \delta \phi^{(1)}(x,0) - \phi_x^{(1)}(x,0)\right) , \;\; x>0; \\
 \Delta \ssig^{(1)}&=0, \;\; x<0. %
\end{align*}
\fi
Taking a Fourier transform of \eqref{c5:Eq:sigIE1} in $x$ gives
\begin{align}
\FsigM^{(1)}(\gamma)+\FsigP^{(1)}(\gamma)&= -4 \pi \dsigp^{(1)}(\gamma) j(\gamma), \label{c5:Eq:sigWH1}
\end{align}
where
\begin{align}
\addtocounter{equation}{1}
\FsigM^{(1)}(\gamma)&=\frac{1}{2 \pi}\int^0_{-\infty} \fsig^{(1)}(x)\e^{\i \gamma x} \d x, &\FsigP^{(1)}(\gamma)&= \frac{1}{2 \pi} \int_0^{\infty} \fsig^{(1)}(x)\e^{\i \gamma x} \d x, \label{c5:Eq:F1sigDef} \tag{\theequation.a}
\end{align}
\begin{align}
\dsigp^{(1)}(\gamma)=\int_0^{\infty} \Delta_0 \left[ \ssig^{(1)}\right] (x) \e^{\i \gamma x} \d x. \label{c5:Eq:D1sigDef} \tag{\theequation.b}
\end{align}
We may employ \eqref{c5:Eq:bc1a} to obtain 
\begin{align*}
\FsigP^{(1)}(\gamma)&= - \frac{\namp}{\pi \i ( \gamma +\delta k_x )} + \left( \mu_0 - \i \mu_1 \gamma - \mu_2 \gamma^2 \right) D^{(1)}(\gamma).
\end{align*}
Consequently, the Wiener--Hopf equation \eqref{c5:Eq:sigWH1} may be expressed as
\begin{align}
\FsigM^{(1)}(\gamma) - \frac{\namp}{\pi \i ( \gamma +\delta  k_x)}  &= -4 \pi \dsigp^{(1)}(\gamma) K(\gamma), \label{c5:Eq:sigWH1a}
\end{align}
where
\begin{align}
K(\gamma) & = j(\gamma)  + \frac{1}{4 \pi} \left( \mu_0 - \i \mu_1 \gamma - \mu_2 \gamma^2 \right) . \label{c5:Eq:kDef}
\end{align}
The multiplicative splitting $K = K_+ K_-$ is performed in appendix \ref{Ap:splitting}. This splitting enables us to write 
\begin{align}
\frac{\FsigM^{(1)}(\gamma)}{K_-(\gamma)} - \frac{\namp}{ \pi \i ( \gamma +\delta  k_x)K_-(\gamma)}  &= -4 \pi \dsigp^{(1)}(\gamma) K_+(\gamma). \label{c5:Eq:sigWH1c}
\end{align}
We now additively factorise the left hand side of \eqref{c5:Eq:sigWH1c}. We apply pole removal \citep{Noble1988} to obtain the additive splitting
\begin{align*}
 \frac{1}{( \gamma + \delta k_x)K_-(\gamma)} = \underbrace{ \frac{1}{( \gamma + \delta  k_x)K_-(-\delta k_x)}}_+ +\underbrace{\frac{1}{( \gamma + \delta  k_x)K_-(\gamma)} - \frac{1}{( \gamma + \delta  k_x)K_-(-\delta k_x)}  }_{-},
\end{align*}
where the underbrace $\pm$ denotes that the function is analytic in $\ULHP$ respectively. Therefore, \eqref{c5:Eq:sigWH1a} becomes
\begin{align}
\frac{\FsigM^{(1)}(\gamma)}{K_-(\gamma)} -\frac{\namp}{ \pi \i ( \gamma + \delta k_x)K_-(\gamma)} + \frac{\namp}{\pi \i ( \gamma + \delta k_x)K_-(-\delta k_x)} \notag   \\
= -4 \pi \dsigp^{(1)}(\gamma) K_+(\gamma) + \frac{\namp}{\pi \i ( \gamma + \delta k_x)K_-(-\delta k_x)}. \label{c5:Eq:sigWH1d}
\end{align}
We may now apply the standard Wiener--Hopf argument: since the left and right hand sides of \eqref{c5:Eq:sigWH1d} are analytic in $\LUHP$ respectively, and they agree on a strip, each side defines the analytic continuation of the other. Therefore, equation \eqref{c5:Eq:sigWH1d} defines an entire function, $E_1(\gamma)$. As $|\gamma| \rightarrow \infty$ in the $\LHP$, the left hand side of \eqref{c5:Eq:sigWH1d} decays. Similarly, as $|\gamma| \rightarrow \infty$ in $\UHP$, the right hand side of \eqref{c5:Eq:sigWH1d} vanishes. Therefore, $E_1(\gamma)$ is bounded in the entire plane and must be constant according to Liouville's theorem. Moreover, since $E_1(\gamma)$ decays, this constant must be zero. Finally, we rearrange the right hand side of \eqref{c5:Eq:sigWH1d} to obtain the solution to the first Wiener--Hopf problem as
\begin{align}
D^{(1)}(\gamma) = \frac{\namp}{(2 \pi)^2 \i ( \gamma + \delta k_x)K_-(-\delta k_x) K_+(\gamma)}. \label{c5:Eq:dSig1}
\end{align}
\subsubsection{Solution to Second Wiener--Hopf Equation -- $\dsig^{(2)}$} \label{c5:Sec:WHsig2}

In this section we solve the integral equation \eqref{c5:Eq:sigGEfn} for $n=2$,
\begin{align}
\addtocounter{equation}{1}
\fsig^{(2)}(x)&= -4 \pi \int_{-\infty}^{\infty} \dsig^{(2)}(\gamma)  j(\gamma)\e^{-\i \gamma x} \d \gamma , \label{c5:Eq:sigIE2} \tag{\theequation.a}
\end{align}
subject to \eqref{c5:Eq:bc2a} and \eqref{c5:Eq:bc2b}
\iffalse
\begin{align}
\fsig^{(2)}(x)&= 0, \;\; x<1; & \Delta \ssig^{(1)}(x)+\Delta \ssig^{(2)}(x)&=2 \pi \i \Psig^{(2)} \e^{-\i \gModeO x}, \;\; x> 1, \label{c5:Eq:sigBC2} \tag{\theequation.b}
\end{align}
\fi
Taking a Fourier transform of \eqref{c5:Eq:sigIE2} and applying \eqref{c5:Eq:bc2a} yields
\begin{align}
\left( \mu_0 - \i \mu_1 \gamma - \mu_2 \gamma^2 \right)  D^{(1)}(\gamma) + \FsigP^{(2)}(\gamma) = -4 \pi\left( \dsigm^{(2)}(\gamma) + \dsigp^{(2)}(\gamma) \right) j(\gamma), \label{c5:Eq:sigWH2b}
\end{align}
where 
{\setlength{\jot}{10pt}
	\begin{align}
	\addtocounter{equation}{1}
	\FsigP^{(2)}(\gamma)=\frac{1}{2 \pi} \int_{2}^{\infty} \fsig^{(2)} ( x )\e^{\i \gamma x} \d x & = \frac{ \e^{2\i \gamma }}{2 \pi} \int_0^{\infty} \fsig^{(2)} (x + 2)\e^{\i x \gamma} \d x=\e^{2\i\gamma} \FsigP^{*(2)} (\gamma), \tag{\theequation.a} \label{c5:Eq:sigFPstar} \\ 
	\dsigp^{(2)}(\gamma)=\frac{1}{2 \pi} \int_{2}^{\infty} \Delta_0 \left[\ssig^{(2)}\right]( x )\e^{\i \gamma x} \d x&= \frac{\e^{2\i \gamma}}{2 \pi} \int_0^{\infty} \Delta_0 \left[\ssig^{(2)} \right]( x+ 2)\e^{\i x \gamma} \d x = \e^{2\i\gamma } \dsigp^{*(2)}(\gamma), \label{c5:Eq:sigDPstar} \tag{\theequation.b}\\
	\dsigm^{(2)}(\gamma)=\frac{1}{2 \pi} \int_{-\infty}^{2} \Delta_0 \left[\ssig^{(2)} \right](x )\e^{\i \gamma x} \d x &= \frac{ \e^{2\i \gamma }}{2 \pi} \int_{-\infty}^{0} \Delta_0 \left[\ssig^{(2)}\right] (x + 2)\e^{\i x \gamma} \d x=\e^{2\i\gamma } \dsigm^{*(2)}(\gamma). \label{c5:Eq:sigDMstar} \tag{\theequation.c}
	\end{align}}
Factoring out the $\e^{2 \i \gamma}$ dependence and employing \eqref{c5:Eq:kDef} transforms the Wiener--Hopf equation \eqref{c5:Eq:sigWH2a} to
\begin{align}
\FsigP^{*(2)}(\gamma)=-4 \pi\left( \dsigm^{*(2)}(\gamma) + \dsigp^{*(2)}(\gamma) \right) K(\gamma), \label{c5:Eq:sigWH2a}
\end{align}
and we may use the multiplicative splitting of $K$ to write
\begin{align}
\frac{\FsigP^{*(2)}(\gamma)}{K_+(\gamma)} =-4 \pi\left( \dsigm^{*(2)}(\gamma) + \dsigp^{*(2)}(\gamma) \right) K_-(\gamma). \label{c5:Eq:sigWH2d}
\end{align}
We may use the downstream boundary condition for this problem \eqref{c5:Eq:bc2b} to write 
\begin{align}
\dsigp^{*(2)}(\gamma)&= - \frac{ \Psig^{*(2)} }{\gamma + \delta \omega} -\frac{1}{2 \pi} \int_0^\infty \Delta_0 \left[ \ssig^{(1)} \right](x+ 2) \e^{\i \gamma x} \d x .\label{c5:Eq:sigBC2D}
\end{align}
where $ \Psig^{*(2)} = \Psig^{(2)} \e^{2 \i \delta k_x}$. To calculate the remaining integral we use the inversion formula for the Fourier transform:
\begin{align*}
\Delta_0 \left[\ssig^{(1)}\right](x)=\int_{-\infty - \i \tau_1}^{\infty-\i \tau_1} \dsig^{(1)}(\gamma) \e^{-\i \gamma x} \d \gamma .
\end{align*}
By substituting this representation into our desired integral, we obtain
\begin{align*}
\frac{1}{2 \pi} \int_0^{\infty} \Delta_0\left[ \ssig^{(1)} \right] (x +2) \e^{\i \gamma x} \d x&=\frac{1}{2 \pi} \int_0^{\infty} \int_{-\infty - \i \tau_1}^{\infty-\i \tau_1} \dsig^{(1)}(\gamma_1) \e^{-\i \gamma_1 (x + 2)} \d \gamma_1 \e^{\i \gamma x} \d x.
\end{align*}
Rearranging the order of integration and computing the resulting $x$-integral results in
\begin{align}
\frac{1}{2 \pi} \int_0^{\infty} \Delta_0 \left[ \ssig^{(1)} \right](x +2) \e^{\i \gamma x} \d x&=\frac{1}{2 \pi \i}  \int_{-\infty - \i \tau_1}^{\infty-\i \tau_1} \frac{\dsig^{(1)}(\gamma_1) \e^{-2 \i \gamma_1 }}{\gamma_1-\gamma}  \d \gamma_1 . \label{c5:Eq:D2int}
\end{align}
Since $\hPS<2$, we may close the remaining integral in $\LHP$. There are no branches in the integrand and the integral consists of the residues of simple poles at $\gamma = - \delta k_x$, $\dModeM$. Consequently, inserting \eqref{c5:Eq:dSig1} into the above integral yields
{
	\setlength{\jot}{10pt}
	\begin{align*}
	\frac{1}{2 \pi \i} & \int_{-\infty - \i \tau_1}^{\infty-\i \tau_1} \frac{1}{( \gamma + \delta k_x)K_-(-\delta k_x) K_+(\gamma)}\cdot\frac{\e^{-2\i \gamma_1}}{\gamma_1-\gamma}  \d \gamma_1 \\ 
	=& \frac{\e^{2\i \delta k_x}}{(\gamma+ \delta k_x)K(-\delta k_x)} +\sum_{n=0}^{\infty} \frac{\e^{-2\i \dModeM }}{(\dModeM + \delta k_x) K_-(- \delta k_x) K_+^{\prime}(\dModeM)(\gamma-\dModeM)},
	\end{align*}
}
where the derivatives of $K_\pm$  evaluated at the duct modes $\dModeMP$ are given by
\begin{align*}
K_{\pm}^{\prime}(\dModeMP) = \frac{K^\prime(\dModeMP)}{K_{\mp}(\dModeMP)} = \frac{-1}{4 \pi K_{\mp}(\dModeMP)} \Bigg( & \frac{(\dModeMP/\zeta(\dModeMP))\sin(s \zeta(\dModeMP)) + s \dModeMP \cos(s \zeta(\dModeMP))}{\cos(s \zeta(\dModeMP)) - \cos(d \dModeMP + \sigma^\prime)} \\
& \frac{ \dModeMP s\sin(s \zeta(\dModeMP)^2 + d \zeta(\dModeMP) \sin(d \dModeMP + \sigma^\prime)\sin(s \zeta(\dModeMP)) }{(\cos(s \zeta(\dModeMP)) - \cos(d \dModeMP + \sigma^\prime))^2}  \\
&  \qquad \qquad  \i \mu_1 + 2 \mu_2 \dModeMP \Bigg).
\end{align*}
Therefore, substitution of \eqref{c5:Eq:D2int} into \eqref{c5:Eq:sigBC2D} yields 
\begin{align}
\dsigp^{*(2)}(\gamma) &= \frac{\namp \e^{2\i \delta k_x}}{(2 \pi)^2 \i(\gamma + \delta k_x)K(- \delta k_x)} -\sum_{n=0}^{\infty} \frac{\Asig \e^{-2\i \dModeM }}{\i (\dModeM + \delta \omega)(\gamma-\dModeM)}, \notag\\
& -  \frac{\Psig^{*(2)}}{\gamma + \delta \omega}, \label{c5:Eq:sigBC3D}
\end{align}
where
\begin{align*}
\Asig &= \frac{- \namp ( \dModeM + \delta \omega )}{(2 \pi)^2 (\dModeM + \delta k_x) K_-(-\delta k_x) K_+^{\prime}(\dModeM) } .
\end{align*}
We use the notation
\begin{align*}
\tilde{K}_-(\gamma,\eta^-) \coloneqq \frac{K_-(\gamma)}{\gamma - \eta^-},
\end{align*}
so that substitution of \eqref{c5:Eq:sigBC3D} into the Wiener--Hopf equation \eqref{c5:Eq:sigWH2b} yields
\begin{align}
\frac{\FsigP^{*(2)}(\gamma)}{4 \pi K_+(\gamma)}= &-K_-(\gamma) \dsigm^{*(2)}(\gamma) + \Psig^{*(2)} \tilde{K}_-(\gamma, -\delta \omega) \notag\\
& +  \frac{\namp\e^{2\i \delta k_x}}{(2 \pi)^2 \i} \cdot \frac{\tilde{K}_-(\gamma,-\delta k_x)}{K(- \delta k_x)} -\sum_{n=0}^{\infty} \frac{\Asig \e^{-2\i \dModeM }}{\i (\dModeM + \delta \omega)} \tilde{K}_-(\gamma,\dModeM). \label{c5:Eq:wh2}
\end{align}
We note the additive splitting
\begin{align}
\tilde{K}_-(\gamma,\eta^-) = \left[\tilde{K}_-(\gamma,\eta)\right]_+ + \left[\tilde{K}_-(\gamma,\eta)\right]_-, \label{c5:Eq:kTilde}
\end{align}
where
\begin{align*}
 \left[\tilde{K}_-(\gamma,\eta)\right]_+ =  \frac{K_-(\eta^-)}{\gamma - \eta^-}, \qquad  \left[\tilde{K}_-(\gamma,\eta)\right]_- =  \frac{K_-(\gamma)-K_-(\eta^-)}{\gamma - \eta^-}.
\end{align*}
Substituting these splittings into \eqref{c5:Eq:wh2} yields
\begin{align}
\frac{\FsigP^{*(2)}(\gamma)}{2 \pi K_+(\gamma)} - \Psig^{*(2)} \left[\tilde{K}_-(\gamma, -\delta \omega) \right]_+ - \frac{\namp\e^{2\i \delta k_x}}{(2 \pi)^2 \i} \cdot \frac{\left[\tilde{K}_-(\gamma,-\delta k_x)\right]_+}{K(- \delta k_x)} \notag \\
 + \sum_{n=0}^{\infty} \frac{\Asig \e^{-2\i \dModeM }}{\i (\dModeM + \delta \omega)} \left[\tilde{K}_-(\gamma,\dModeM)\right]_+ =-K_-(\gamma) \dsigm^{*(2)}(\gamma) + \Psig^{*(2)} \left[ \tilde{K}_-(\gamma, -\delta \omega) \right]_- \notag\\
 +  \frac{\namp\e^{2\i \delta k_x}}{(2 \pi)^2 \i} \cdot \frac{\left[\tilde{K}_-(\gamma,-\delta k_x)\right]_-}{K(- \delta k_x)} -\sum_{n=0}^{\infty} \frac{\Asig \e^{-2\i \dModeM }}{\i (\dModeM + \delta \omega)} \left[\tilde{K}_-(\gamma,\dModeM)\right]_-, \label{c5:Eq:sigWH2c}
\end{align}
\iffalse
\begin{align}
&\frac{\FsigP^{*(2)}(\gamma)}{2 \pi K_+(\gamma)} -\frac{\Psig^{*(2)} K_-(\kappa) }{\gamma- \gModeO}  + \frac{\namp K_-(-k_x)\e^{\i k_x}}{(2 \pi)^2 \i(\gamma+k_x)K(-k_x)}+\sum_{n=0}^{\infty} \frac{ A_n K_-(\dModeM) \e^{-\i \dModeM }}{(\dModeM -\kappa)(\gamma-\dModeM)} \notag \\
& = -K_-(\gamma) \dsigm^{*(2)}(\gamma) + \frac{\Psig^{*(2)}(K_-(\gamma) - K_-(\kappa)) }{\gamma- \gModeO}  - \frac{\namp (K_-(\gamma)- K_-(-k_x))\e^{\i k_x}}{(2 \pi)^2 \i(\gamma+k_x)K(-k_x)} \notag\\
& -\sum_{n=0}^{\infty} \frac{(K_-(\gamma) - K_-(\dModeM))A_n \e^{-\i \dModeM }}{(\dModeM -\kappa)(\gamma-\dModeM)}
\end{align}
\fi
In a similar way to section \ref{c5:Sec:WHsig1}, we now apply the typical Wiener--Hopf argument. We enforce the unsteady Kutta condition \citep{Ayton2016} which restricts the pressure at the trailing edge to be finite. Consequently, the left hand side of \eqref{c5:Eq:sigWH2c} decays as $|\gamma| \rightarrow \infty$ in $\UHP$ and the left hand side of \eqref{c5:Eq:sigWH2c} tends to an unknown constant as $|\gamma| \rightarrow \infty$ in $\LHP$. Applying analytic continuation and Liouville's theorem determines that this constant must be zero. Accordingly, the coefficient of $\gamma^{-1}$ on the left hand side of \eqref{c5:Eq:sigWH2c} must vanish so that
\begin{align*}
\Psig^{*(2)}= \frac{\namp \e^{2\i \delta k_x}}{(2 \pi)^2 \i K_+(-\delta k_x)} \cdot \frac{1}{K_-(-\delta \omega)} - \sum_{n=0}^{\infty} \frac{ \Asig \e^{-2\i \dModeM }}{\i(\dModeM +\delta \omega)} \cdot \frac{K_-(\dModeM)}{K_-(- \delta \omega)}.  \label{c5:Eq:sigP2}
\end{align*}
\iffalse
By rearranging \eqref{c5:Eq:sigWH2c}, we may write

\begin{align}
\dsig^{*(2)}(\gamma) = \sum_{l=-\infty}^{\infty} &\frac{\Ssig \e^{ -\i \gModel 1}}{\i (\gModeO-\gModel) (\gamma - \gModel)}\left(1 - \frac{K_-(\gModel)}{K_-(\gamma)}\right) \notag \\
+ \sum_{n=0}^{\infty} &\frac{\Asig \e^{- \i \dModeM 1}}{\i(\dModeM- \gModeO)(\gamma - \dModeM)}\left(1-\frac{K_-(\dModeM)}{K_-(\gamma)} \right) +\frac{\Psig^{*(2)}}{\gamma -  \gModeO}\left(1-\frac{K_-(\gModeO)}{K_-(\gamma)} \right)
\end{align}
\fi
So that, after substituting in the downstream representation \eqref{c5:Eq:sigBC3D} and the expression for the pressure constant \eqref{c5:Eq:sigP2}, the right hand side of \eqref{c5:Eq:sigWH2c} yields
\begin{align}
\dsig^{(2)}(\gamma)&= \frac{\namp (\delta \omega - \delta k_x) \e^{2 \i (\gamma + \delta k_x)}}{(2 \pi)^2 \i(\gamma+ \delta k_x)(\gamma + \delta \omega ) K_+(- \delta k_x)} \cdot \frac{1}{K_-(\gamma)} \notag \\
&- \sum_{n=0}^{\infty} \frac{\Asig \e^{2 \i (\gamma - \dModeM) }}{\i(\gamma + \delta \omega )(\gamma - \dModeM)}\cdot\frac{K_-(\dModeM) }{K_-(\gamma)}.
\end{align}
It should be noted that the poles of $\dsig^{(2)}$ in $\UHP$ are only at the zeros of $K_-$.
\iffalse[DETAILS OF 34 SOLUTION HERE]

\subsubsection{Solution to Third and Fourth Wiener--Hopf Equations -- $\dsig^{(3)},\, \dsig^{(4)}$} 

The solution to the couple 3rd and 4th Wiener--Hopf equations is identical to that in previous sections, except

\begin{align*}
D_{1,\Sigma-\textrm{res}}^{(2)}(\theta_k^+)&=\frac{-\e^{\i \theta_k^+1}}{\i(\theta_k^++k \delta)K_-^\prime(\theta_k^+)} \left\{\sum_{l=-\infty}^{\infty}\frac{ \mathcal{S}_l \e^{-\i \gModel 1}}{\theta_k^+ - \gModel} \cdot K_-(\gModel) + \sum_{n=0}^{\infty} \frac{\mathcal{A}_{1,\Sigma,n}  \e^{- \i \dModeM1}}{\theta_k^+ - \dModeM}\cdot K_-(\dModeM) \right\}
\end{align*}
and the solutions are
\begin{align*}
D_{1,\Sigma}^{(3)}(\gamma)&=- \sum_{n=0}^{\infty} \frac{\mathcal{B}_{1,\Sigma,n}}{\gamma-\dModeP} \cdot \frac{K_+(\dModeP)}{K_+(\gamma)}\\
D_{1,\Sigma}^{(4)}(\gamma)&=-\sum_{n=0}^{\infty} \frac{\mathcal{C}_{1,\Sigma,n} \e^{\i1 (\gamma - \dModeM)}}{\i(\gamma-+k \delta)(\gamma - \dModeM)}\cdot\frac{K_-(\dModeM)}{K_-(\gamma)}
\end{align*}

\subsubsection{Solution of 3rd and 4th Wiener--Hopf Equations}
\fi
\subsubsection{Solution to Third and Fourth Wiener--Hopf Equations -- $\dsig^{(3)},\, \dsig^{(4)}$} \label{c5:Sec:WHsig34}
Since the integral equations for $\Delta \ssig^{(3)}$ and $\Delta \ssig^{(4)}$ are coupled, we must solve for them simultaneously. Taking a Fourier Transform of \eqref{c5:Eq:sigGEfn} and applying the boundary conditions for \eqref{c5:Eq:bc3a} and \eqref{c5:Eq:bc4a} for $n=3,4$  gives
\begin{align}
\FsigM^{(3)}(\gamma) &=-4 \pi \left[\dsigm^{(3)}(\gamma) + \dsigp^{(3)}(\gamma) \right] K(\gamma), \label{c5:Eq:sigWH3} \\
\FsigP^{*(4)}(\gamma) &= -4 \pi \left[\dsigm^{*(4)}(\gamma) + \dsigp^{*(4)}(\gamma) \right] K(\gamma),\label{c5:Eq:sigWH4}
\end{align}
\iffalse 
\begin{align*}
\FsigM^{(3)}(\gamma) = \frac{1}{2 \pi} & \int_{-\infty}^{0} \fsig^{(3)}(x) \e^{\i \gamma x} \d x, \\
\dsigm^{(3)}(\gamma) = \frac{1}{2 \pi} \int_{-\infty}^{0} \Delta \ssig^{(3)}(x) \e^{\i \gamma x} \d x,\quad &\quad \dsigp^{(3)}(\gamma) = \frac{1}{2 \pi} \int_{0}^{\infty} \Delta \ssig^{(3)}(x) \e^{\i \gamma x} \d x,
\end{align*}

and

where

\begin{align*}
\FsigP^{*(4)}(\gamma) = \frac{1}{2 \pi} & \int_{0}^{\infty} \fsig^{(4)}(x+1) \e^{\i \gamma x} \d x, \\
\dsigm^{*(4)}(\gamma) = \frac{1}{2 \pi} \int_{-\infty}^{0} \Delta \ssig^{(4)}(x + 1) \e^{\i \gamma x} \d x,\quad &\quad \dsigp^{*(4)}(\gamma) = \frac{1}{2 \pi} \int_{0}^{\infty} \Delta \ssig^{(4)}(x + 1) \e^{\i \gamma x} \d x.
\end{align*}
\fi
where $\FsigM^{(3)}$, $\dsigm^{(3)}$ and $\dsigp^{(3)}$ are defined in an analogous way to \eqref{c5:Eq:F1sigDef} and \eqref{c5:Eq:D1sigDef}, and $\FsigP^{*(4)}$, $\dsigp^{*(4)}$ and $\dsigm^{*(4)}$ are defined in an analogous way to \eqref{c5:Eq:sigFPstar}, \eqref{c5:Eq:sigDPstar} and \eqref{c5:Eq:sigDMstar}. Using a similar approach to \mySec \ref{c5:Sec:WHsig2}, the upstream boundary condition \eqref{c5:Eq:bc3b} may be expressed as
\begin{align*}
\dsigm^{(3)}(\gamma)=\frac{1}{2 \pi \i} \int_{-\infty+\i \tau_0}^{\infty+\i \tau_0} \frac{\dsig^{(2)}(\gamma_1) + \dsig^{(4)}(\gamma_1)}{\gamma_1-\gamma} \d \gamma_1.
\end{align*}
Consequently, we may express $\dsigm^{(3)}$ in terms of its poles $\dModeP$ as
\begin{align}
\dsigm^{(3)}(\gamma)=-\sum_{n=0}^{\infty} \frac{\Bsig}{\gamma -\dModeP}, \label{c5:Eq:sigBres}
\end{align}
where $\Bsig$ are the residues of $\dsigm^{(2)}(\gamma_1) + \dsigm^{(4)}(\gamma_1)$ at $\gamma = \dModeP$. The residues of $\dsigm^{(4)}$ are currently unknown, but the residues of $\dsigm^{(2)}$ are given by
\begin{align*}
\dsigres^{(2)}= %
 \frac{-\e^{2 \i \dModePr }}{K_-^\prime(\dModePr)} &\left\{ \frac{\namp (\delta \omega - \delta k_x) \e^{2 \i \delta k_x}}{(2 \pi)^2 \i(\dModePr+ \delta k_x)(\dModePr+ \delta \omega ) K_+(- \delta k_x)} \right.\\
&\left. + \sum_{n=0}^{\infty} \frac{\Asig \e^{\i (\dModePr - \dModeM) }K_-(\dModeM)}{\i(\dModePr + \delta \omega)(\dModePr - \dModeM)} \right\}.
\end{align*}
We may now substitute \eqref{c5:Eq:sigBres} into \eqref{c5:Eq:sigWH3} to obtain the Wiener--Hopf equation
\begin{align*}
\frac{\FsigM^{(3)}(\gamma)}{4 \pi K_-(\gamma)}+ \sum_{n=0}^{\infty} \frac{\Bsig}{\gamma-\dModeP}K_+(\dModeP) &=\dsigp^{(3)}(\gamma) K_+(\gamma)-\sum_{n=0}^{\infty} \frac{\Bsig}{\gamma-\dModeP} \left(K_+(\gamma) - K_+(\dModeP) \right).
\end{align*}
The edge conditions are identical to those applied in  \mySec \ref{c5:Sec:WHsig1}, and we employ the typical Wiener--Hopf argument to obtain
\begin{align*}
\dsigp^{(3)}(\gamma)&= \sum_{n=0}^{\infty} \frac{\Bsig}{\gamma - \dModeP} \bigg\{1- \frac{K_+(\dModeP)}{K_+(\gamma)} \bigg\} .
\end{align*}
Combining this solution with \eqref{c5:Eq:sigBres} yields
\begin{align}
\dsig^{(3)}(\gamma)&=- \sum_{n=0}^{\infty} \frac{\Bsig}{\gamma - \dModeP} \cdot \frac{K_+(\dModeP)}{K_+(\gamma)}. \label{c5:Eq:sigD3def}
\end{align}
We proceed to the solution for $\dsig^{(4)}$. In a similar way to \mySec \ref{c5:Sec:WHsig2}, we may invert the Fourier transform for the downstream boundary condition \eqref{c5:Eq:bc4b} to write
\begin{align*}
\dsigp^{*(4)}(\gamma) %
&=-\frac{\Psig^{*(4)}}{\i(\gamma + \delta \omega)}-\frac{1}{2 \pi\i} \int_{-\infty-\i \tau_1}^{\infty+\i \tau_1} \sum_{n=0}^{\infty} \frac{\Bsig}{(\gamma_1-\gamma)(\gamma_1 - \dModeP)} \bigg\{1- \frac{K_+(\dModeP)}{K_+(\gamma_1)} \bigg\}\e^{-2\i\gamma_1} \d \gamma_1. 
\end{align*}
This integral can be closed in $\LHP$ to obtain
\begin{align}
\dsigp^{*(4)}(\gamma)&=-\frac{\Psig^{*(4)}}{\gamma +\delta \omega}-\sum_{n=0}^{\infty} \frac{\Csig \e^{-2 \i \dModeM}}{\i(\dModeM + \delta \omega )(\gamma - \dModeM)}, \label{c5:Eq:sigBC4}
\end{align}
where
\begin{align*}
\Csig &= \sum_{k=0}^{\infty} \frac{\i(\dModeMv{n} + \delta \omega )}{(\dModePv{k} - \dModeMv{n})} \cdot \frac{K_+(\dModePv{k})}{K_+^{\prime}(\dModeMv{n})} \cdot \BsigR.
\end{align*}
After truncation, we may write this system of equations in matrix form
\begin{align}
\mathbf{\mathcal{C}}=\mathbf{L}\mathbf{\mathcal{B}}, \label{c5:Eq:sigMat1}
\end{align}
where
\begin{align*}
\left\{ \mathbf{L} \right\}_{n,m}=\frac{\i(\dModeMv{n} + \delta \omega )}{(\dModePv{m} -\dModeMv{n})} \cdot \frac{K_+(\dModePv{m})}{K_+^{\prime}(\dModeMv{n})} .
\end{align*}
\iffalse DETAILS OF FACTORISATIONS
Therefore,

\begin{align}
D_{1+}^{*(4)}(\gamma)K_-(\gamma)&=-\frac{KK_-(\gamma)}{\i(\gamma+k\delta)}-K_-(\gamma)\sum_{n=0}^{\infty} \frac{\mathcal{C}_n\e^{-\i \theta_n^-x_b}}{\i(\theta_n^-+k \delta)(\gamma - \theta_n^-)}\\
&=\underbrace{-\frac{K(K_-(\gamma)-K_-(-k\delta))}{\i(\gamma+k\delta)}- \sum_{n=0}^{\infty} \frac{(K_-(\gamma)-K_-(\theta_n^-))\mathcal{C}_n \e^{-\i \theta_n^-x_b}}{\i(\theta_n^-+k \delta)(\gamma - \theta_n^-)}}_{-}+\underbrace{-\frac{K K_-(-k\delta)}{\i(\gamma+k\delta)})-\sum_{n=0}^{\infty} \frac{K_-(\theta_n^-)\mathcal{C}_n \e^{-\ix_b \theta_n^-}}{\i(\theta_n^-+k \delta)(\gamma - \theta_n^-)}}_+.\\
\end{align}

\fi
By applying the notation introduced in \eqref{c5:Eq:kTilde}, we may express the Wiener--Hopf equation \eqref{c5:Eq:sigWH4} in the form
\begin{align}
\frac{\FsigP^{*(4)} (\gamma)}{4\pi K_+(\gamma)} +\Psig^{*(4)} \left[\tilde{K}_-(\gamma,-\delta \omega )\right]_+ + \sum_{n=0}^{\infty} \frac{\left[\tilde{K}_-(\gamma,\dModeM)\right]_+}{\i(\dModeM + \delta \omega)}\Csig \e^{-2\i \dModeM} \notag \\ 
= \dsigm^{*(4)}(\gamma)K_-(\gamma) -\Psig^{*(4)}  \left[\tilde{K}_-(\gamma,- \delta \omega)\right]_-  -\sum_{n=0}^{\infty} \frac{\left[\tilde{K}_-(\gamma,\dModeM)\right]_-}{\i(\dModeM + \delta \omega)}\Csig\e^{-2\i \dModeM}. \label{c5:Eq:WHsig4}
\end{align}
Employing the unsteady Kutta condition in \eqref{c5:Eq:WHsig4} yields
\begin{align*}
\Psig^{*(4)} =-\sum_{n=0}^{\infty} \frac{\Csig\e^{-2 \i \dModeM}}{\i(\dModeM + \delta \omega)} \cdot \frac{K_-(\dModeM)}{K_-(- \delta\omega)} .
\end{align*}
Finally, applying the downstream boundary condition \eqref{c5:Eq:sigBC4} and rearranging \eqref{c5:Eq:WHsig4} yields
\iffalse
\begin{align*}
\dsigm^{*(4)}(\gamma)&= \frac{K}{\i(\gamma - \gModeO)}\left(1-\frac{K_-(\gModeO)}{K_-(\gamma)}\right)+\sum_{n=0}^{\infty} \frac{\Csig \e^{-\i1 \dModeM}}{\i(\dModeM - \gModeO)(\gamma - \dModeM)}\left(1-\frac{K_-(\dModeM)}{K_-(\gamma)}\right)
\end{align*}

\begin{align*}
D_{1}^{*(4)}(\gamma)&= -\frac{K}{\i(\gamma+k\delta)}\cdot\frac{K_-(-k\delta)}{K_-(\gamma)}-\sum_{n=0}^{\infty} \frac{\mathcal{C}_n \e^{-\ix_b \theta_n^-}}{\i(\theta_n^-+k \delta)(\gamma - \theta_n^-)}\cdot\frac{K_-(\theta_n^-)}{K_-(\gamma)}\\
&=-\sum_{n=0}^{\infty} \frac{\mathcal{C}_n \e^{-\ix_b \theta_n^-}}{\i(\theta_n^-+k \delta)(\gamma - \theta_n^-)}\cdot\frac{K_-(\theta_n^-)}{K_-(\gamma)}\\
\end{align*}

\fi
\begin{align}
\dsig^{(4)}(\gamma)&=-\sum_{n=0}^{\infty} \frac{\Csig \e^{2 \i (\gamma - \dModeM)}}{\i(\dModeM + \delta \omega)(\gamma - \dModeM)}\cdot\frac{K_-(\dModeM)}{K_-(\gamma)}. \label{c5:Eq:sigD4def}
\end{align}
We are now able to calculate the residues of $\dsig^{(3)}$ as
\begin{align*}
\Bsig&=\dsigres-\sum_{m=0}^{\infty} \frac{\Csig \e^{2 \i(\dModePv{n}- \dModeMv{m})}}{\i (\dModePv{n} + \delta \omega )(\dModePv{n} - \dModeMv{m})}\cdot\frac{K_-(\dModeMv{m})}{K^\prime_-(\dModePv{n})},
\end{align*}
or, in matrix form,
\begin{align}
\mathbf{\mathcal{B}}=\mathbf{\dsigres}+\mathbf{F}\mathbf{\mathcal{C}}, \label{c5:Eq:sigMat2}
\end{align}
where
\begin{align*}
\{\mathbf{F}\}_{n,m}&= -\frac{\e^{2\i(\dModePv{n}- \dModeMv{m})}}{\i(\dModePv{n} + \delta\omega)(\dModePv{n} - \dModeMv{m})}\cdot \frac{K_-(\dModeMv{m})}{K^\prime_-(\dModePv{n})}.
\end{align*}
The matrix equations \eqref{c5:Eq:sigMat1} and \eqref{c5:Eq:sigMat2} may be combined and solved to give expressions the final expressions for $\Bsig$ and $\Csig$.
\section{Details of Fourier Inversion} \label{Ap:FourInv}
\setcounter{figure}{0}
The acoustic field is given by
\begin{align*}
\phi(x,y)&=\frac{1}{2}\int_{-\infty}^{\infty} D(\gamma) A(\gamma, x,y) \d \gamma ,
\end{align*}
where $A = A_u + A_d$ and
\begin{align*}
A_{u}(\gamma,x,y)&=  \e^{-\i \gamma x} \cdot \frac{\e^{\i d \gamma+ \i \sigma^\prime} \cos\left(\zeta y\right)}{\cos\left(d \gamma+\sigma^\prime\right)- \cos\left(s \zeta\right)}, \\ %
A_{d}(\gamma,x,y)&= \e^{-\i \gamma x} \cdot \frac{-\cos\left(\zeta(y- s)\right)}{\cos\left(d \gamma+\sigma^\prime\right)- \cos\left(s \zeta\right)}.%
\end{align*}
We calculate the above integral by splitting the physical plane into five separate regions, as illustrated in figure \ref{c5:Fig:FourInvRegions}.
Both $A_{a}$ and $A_{b}$ have poles at the acoustic modes  $\lambda_m^\pm$ where the residues are given by
\begin{align*}
A_{u}^r(\lambda_m^\pm,x,y)&= \mp \frac{\zeta_m^\pm \e^{\i d \lambda_m^\pm+ \i \sigma^\prime} \cos\left(\zeta_m^\pm y\right)}{\Delta \sin\left( s \zeta_m^\pm\right) \sqrt{k^2 w^2-f_m^2} }  \e^{-\i \lambda_m^\pm x},\\
A_{d}^r(\lambda_m^\pm,x,y)&= \pm \frac{\zeta_m^\pm\cos\left(\zeta_m^\pm(y- s)\right)}{\Delta   \sin\left( s \zeta_m^\pm\right)\sqrt{k^2 w^2-f_m^2}}  \e^{-\i \lambda_m^\pm x}.
\end{align*}
In order to proceed, we split the acoustic potential into four components
\begin{align*}
\phi(x,y)&=\underbrace{\frac{1}{2}\int_{-\infty}^{\infty} \left(D^{(1,3)} (\gamma) \right) A_u(\gamma) \e^{-\i \gamma \left( x-y d/s\right)} \d \gamma}_{\phi_{u,u}} \\
&+\underbrace{\frac{1}{2}\int_{-\infty}^{\infty} \left(D^{(1,3)} (\gamma) \right) A_d(\gamma) \e^{-\i \gamma \left( x-y d/s\right)} \d \gamma}_{\phi_{u,d}} \\
&+\underbrace{\frac{1}{2}\int_{-\infty}^{\infty} \left(D^{(2,4)}(\gamma) \right) A_u(\gamma) \e^{-\i \gamma \left( \phi-\psi d/h\right)} \d \gamma }_{\phi_{d,u}} \\
&+\underbrace{\frac{1}{2}\int_{-\infty}^{\infty} \left(D^{(2,4)}(\gamma) \right) A_d(\gamma) \e^{-\i \gamma \left( \phi-\psi d/h\right)} \d \gamma }_{\phi_{d,d}},
\end{align*}
where $D^{(i,j)} = D^{(i)} + D^{(j)}$.
\iffalse
We note that away from the poles of $D$ we have the limiting behaviours 
\begin{align*}
|D^{(1,3)}(\gamma)| \rightarrow 0, \qquad
|D^{(2,4)}(\gamma)|\e^{-2\i \gamma} \rightarrow 0,
\end{align*}
as  $\left| \gamma \right| \rightarrow \infty$. Therefore, \peterAlert{???}
\begin{align*}
\left|\left(D^{(1)} (\gamma)+D^{(3)} (\gamma)\right) A(\gamma) \e^{-\i \gamma(\phi-\psi d/h)}\right| \rightarrow 0 \qquad \textrm{in LHP}\\
\left|\left(D^{(2)} (\gamma)+D^{(4)} (\gamma)\right)A(\gamma) \e^{-\i \gamma \phi_b}\e^{\i \gamma(\phi_b-(\phi-\psi d/h))} \right|=\left|\left(D^{(2)} (\gamma)+D^{(4)} (\gamma)\right)A(\gamma) \e^{-\i \gamma(\phi-\psi d/h)} \right|\rightarrow 0 \qquad \textrm{in UHP}
\end{align*}
and therefore the integrals for $h_0^l(\phi,\psi)$ and $h_0^u(\phi,\psi)$ can be closed in the lower and upper half-planes respectively. 
 In the subsequent integrations, we will use the facts
\begin{align*}
A_1(\theta_n^{\pm}) &= \e^{-\i \theta_n^\pm \phi}\cos\left(n \pi  \frac{\psi}{h} \right) \frac{\e^{\i d \theta_n^\pm+ \i \sigma}- (-1)^n}{\cos \left( d \theta_n^\pm + \sigma \right)- (-1)^n}\\
A_2(\theta_n^{\pm}) &= \i\e^{-\i \theta_n^\pm \phi}\sin\left(n \pi \frac{ \psi}{h} \right) \frac{ \e^{\i d \theta_n^\pm+ \i \sigma}- (-1)^n}{\cos \left( d \theta_n^\pm + \sigma \right)- (-1)^n}\\
\end{align*}
\fi
\begin{figure}
\centering
\begin{tikzpicture}[scale = .8]
\def\indentationRadius{.2}

\def\diamondBlue#1#2#3{ %
	\begin{scope}[shift={#1}]
	    \node [scale=#2, draw, rotate=45, shape=rectangle,  anchor=center, fill = prsared] at (0,0) {};
     \node[below = .1cm] at (0,0) {#3};	    
	\end{scope}
}

\def\diamondRed#1#2#3{ %
	\begin{scope}[shift={#1}]
	    \node [scale=#2, draw, rotate=45, shape=rectangle,  anchor=center, fill = myblue] at (0,0) {};
     \node[below = .1cm] at (0,0) {#3};	    
	\end{scope}
}
\def\rTriangle#1#2#3{ %
	\begin{scope}[shift={#1}]
	    \node [scale=#2, draw, rotate=45, triangle,  anchor=center, fill = myblue] at (0,0) {};
     \node[below = .1cm] at (0,0) {#3};	    
	\end{scope}
}

\draw[Latex-Latex] (0,-5.5) -- (0,5.5);  %
\draw[Latex-Latex] (-5.5,0) -- (5.5,0);   

\coordinate (upole4) at (-2,0); \draw[fill=mypurple] (upole4) circle (\indentationRadius/2) node[below = .1cm] {$\lambda_0^-$};
\draw[fill=mypurple] (-3,-2) circle (\indentationRadius/2) node[left = .1cm] {$\lambda_{-1}^-$};
\draw[fill=mypurple] (-1,-2) circle (\indentationRadius/2) node[left = .1cm] {$\lambda_{+1}^-$};
\draw[fill=mypurple] (-0,-4) circle (\indentationRadius/2) node[left = .1cm] {$\lambda_{+2}^-$};

\coordinate (upole1) at (-4,0); \diamondBlue{(upole1)}{4*\indentationRadius}{$- \delta k_x $};
\coordinate (upole2) at (-1,0); \diamondRed{(upole2)}{4*\indentationRadius}{$- \delta \omega $};

\coordinate (dpole2) at (1,0); \draw[fill=myorange] (dpole2) circle (\indentationRadius/2) node[above = .1cm] {$\theta_0^+$};
\draw[fill=myorange] (0,1.5) circle (\indentationRadius/2) node[right = .1cm] {$\theta_{1}^+$};
\draw[fill=myorange] (.3,3) circle (\indentationRadius/2) node[right = .1cm] {$\theta_{2}^+$};
\draw[fill=myorange] (-.2,4.5) circle (\indentationRadius/2) node[right = .1cm] {$\theta_{3}^+$};

\coordinate (upoleT) at (-.25,0); \draw[fill=mygreen] (upoleT) circle (\indentationRadius/2); \node[above = .1cm] at (upoleT) {$\theta_0^-$};
\draw[fill=mygreen] (0,-1.5) circle (\indentationRadius/2) node[right = .1cm] {$\theta_{1}^-$};
\draw[fill=mygreen] (-0.2,-3) circle (\indentationRadius/2) node[right = .1cm] {$\theta_{2}^-$};
\draw[fill=mygreen] (0.3,-4.5) circle (\indentationRadius/2) node[right = .1cm] {$\theta_{3}^-$};

\coordinate (dpole1) at (.5,0); \draw[fill=myyellow] (dpole1) circle (\indentationRadius/2) node[above = .1cm] {$\lambda_0^+$};
\draw[fill=myyellow] (1.5,2) circle (\indentationRadius/2) node[left = .1cm] {$\lambda_{+1}^+$};
\draw[fill=myyellow] (-.5,2) circle (\indentationRadius/2) node[left = .1cm] {$\lambda_{-1}^+$};
\draw[fill=myyellow] (2.5,4) circle (\indentationRadius/2) node[left = .1cm] {$\lambda_{+2}^+$};
\draw[fill=myyellow] (-1.5,4) circle (\indentationRadius/2) node[left = .1cm] {$\lambda_{-2}^+$};

\draw[thick, dashed, myred, decoration =
{markings,
 mark=at position 0.2 with {\arrow{Latex}}, 
 mark=at position 0.4 with {\arrow{Latex}},
  mark=at position 0.6 with {\arrow{Latex}},
   mark=at position 0.8 with {\arrow{Latex}}},
      postaction = {decorate}
] (5,0) arc (0:180:5);
\draw[thick, dashed, myred, decoration =
{markings,
 mark=at position 0.2 with {\arrow{Latex}}, 
 mark=at position 0.4 with {\arrow{Latex}},
  mark=at position 0.6 with {\arrow{Latex}},
   mark=at position 0.8 with {\arrow{Latex}}},
      postaction = {decorate}
] (5,0) arc (0:-180:5);
\draw[thick,myred,xshift=2pt,
decoration={ markings,
      mark=at position 0.04 with {\arrow{Latex}}, 
      mark=at position 0.2 with {\arrow{Latex}},
      mark=at position 0.7 with {\arrow{Latex}}, 
      mark=at position 0.98 with {\arrow{Latex}}}, 
      postaction={decorate}]
 (-5,0) -- ($(upole1)-(\indentationRadius,0)$) arc (180:0:\indentationRadius) --
		   ($(upole4)-(\indentationRadius,0)$) arc (180:0:\indentationRadius) --
	   	   ($(upole2)-(\indentationRadius,0)$) arc (180:0:\indentationRadius) --
		   ($(upoleT)-(\indentationRadius,0)$) arc (180:0:\indentationRadius) --
		   ($(dpole1)-(\indentationRadius,0)$) arc (-180:0:\indentationRadius) --
   		   ($(dpole2)-(\indentationRadius,0)$) arc (-180:0:\indentationRadius) --
		   (5,0);
\end{tikzpicture}
\caption{Illustration of the locations of poles in the complex \mbox{$\gamma$-plane}, and the relevant contours of integration.} \label{Fig:PoleLoc}
\end{figure}
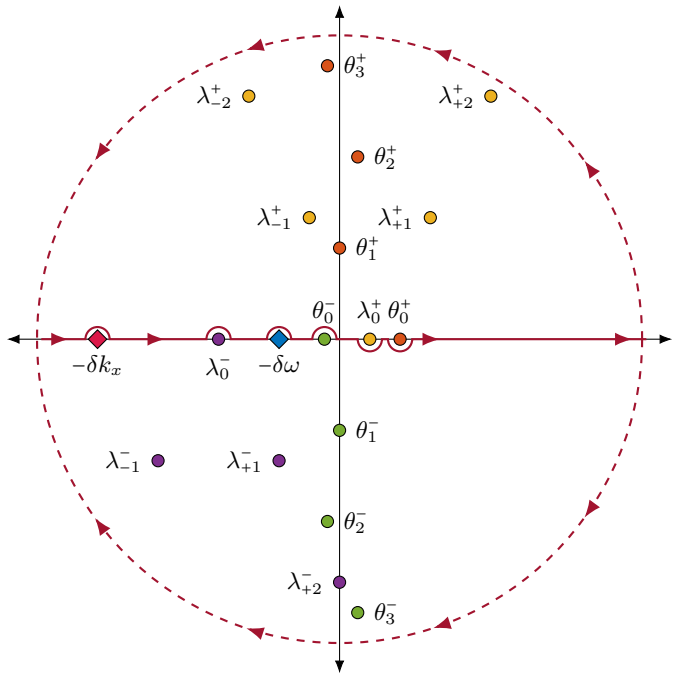

We first calculate 
\begin{align}
\phi_{u,u}(x,y)&=
\frac{1}{2}\int_{-\infty}^{\infty} D_0^{(1,3)}(\gamma)  A_u(\gamma,x,y) \d \gamma \notag\\
&=2 \pi \int_{-\infty}^{\infty} \left\{ D_0^{(1,3)}(\gamma) J_+(\gamma) \right\} \cdot J_-(\gamma)\cdot \frac{\e^{\i (d-\phi) \gamma+ \i \sigma} \cos\left(\zeta \psi\right)}{\zeta \sin \left(h \zeta \right)}\d \gamma. \label{c5:Eq:firstInt}
\end{align}
Since $J_+(\gamma)$ has algebraic growth, inspection of \eqref{c5:Eq:Dsol} shows that $\left| D_0^{(1,3)}(\gamma) J_+(\gamma)\right| \rightarrow 0$ as $\left| \gamma \right| \rightarrow \infty$ except at the poles at $\theta_n^+$ and $-\delta k_x$. Consequently, we close the above integral in $\UHP$ if $x<d$ and in $\LHP$ if $x>d$. Accordingly, we obtain for $x<d$
\begin{align*}
\phi_{u,u}(x,y)&= - \pi \i \sum_{n=0}^{\infty}  \Bsig A_{u}(\dModeP,x,y) + \pi \i \sum_{m=-\infty}^{\infty} D^{(1,3)}(\lambda_m^+) A_{u}^r(\aModeP,x,y) ,
\end{align*}
and for $x>d$,
\begin{align*}
\phi_{u,u}(x,y)&= -\pi \sum_{n=0}^\infty \frac{\Asig+\Csig }{\theta_n^-+\delta \omega}A_u(\theta_n^-,x,y) -\pi \frac{A_u(-\delta k_x,x,y)}{K(-\delta k_x)} \cdot \frac{w_0}{(2 \pi)^2}.
\end{align*}
We proceed by considering
\begin{align*}
\phi_{u,d}(x,y)&=\frac{1}{2}\int_{-\infty}^{\infty} D^{(1,3)}(\gamma) A_d(\gamma, x,y) \d \gamma \\
&=2 \pi \int_{-\infty}^{\infty} \left\{ D_0^{(1,3)}(\gamma) J_+(\gamma) \right\} \cdot J_-(\gamma)\cdot \frac{-\e^{-\i \gamma x}\cos\left(\zeta (y -s)\right)}{\zeta \sin \left(s \zeta \right)}\d \gamma .
\end{align*}
Using a similar argument to the analysis for \eqref{c5:Eq:firstInt}, we close the above integral in  $\LHP$ if $x>0$ in $\UHP$ if $x<0$. Consequently, for $x>0$ we obtain
\begin{align*}
\phi_{u,d}(x,y)&=-\pi \sum_{n=0}^\infty \frac{\Asig+\Csig}{\dModeM + \delta \omega}A_d(\dModeM,x,y) -\pi \frac{A_d(-\delta k_x,x,y)}{K(-\delta k_x)} \cdot \frac{w_0}{(2 \pi)^2},
\end{align*}
and for $x<0$ we obtain
\begin{align*}
\phi_{u,d}(x,y)&= - \pi \i \sum_{n=0}^{\infty}  \Bsig A_{d}(\dModeP,x,y) + \pi \i \sum_{m=-\infty}^{\infty} D^{(1,3)}(\lambda_m^+) A_{d}^r(\aModeP,x,y) ,
\end{align*}
We now consider
\begin{align*}
\phi_{d,u}(x,y)&=\frac{1}{2}\int_{-\infty}^{\infty} D^{(2,4)}_0(\gamma) A_u(\gamma, x,y)\d \gamma \\
&= 2 \pi \int_{-\infty}^{\infty} \left\{D^{(2,4)}(\gamma)\e^{-2 \i \gamma} J_-(\gamma) \right\} J_+(\gamma) \frac{\e^{\i (d + 2-x) \gamma+ \i \sigma^\prime} \cos\left(\zeta y\right)}{\zeta \sin(s \zeta)}\d \gamma .
\end{align*}
Following a similar argument to \eqref{c5:Eq:firstInt}, $\left|D^{(2,4)}_0(\gamma)\e^{-2 \i \gamma} J_-(\gamma) \right| \rightarrow 0$ as $\left| \gamma \right| \rightarrow \infty$ except at the poles at $-\delta \omega$, $-\delta k_x$ and $\dModePM$. Consequently, we close the integral in $\UHP$ if $x<d + 2$ and in $\LHP$ if $x>d+2$. For $x<d+2$ we obtain
\begin{align*}
\phi_{d,u}(x,y) &=\pi \i \sum_{n=0}^\infty B_n A_u(\theta_n^+),
\end{align*}
and for $x>d+2$ we obtain
\begin{align*}
\phi_{d,u}(x,y)&= \pi \sum_{n=0}^{\infty} \frac{\Asig+\Csig}{\dModeM+\delta \omega} A_u(\dModeM,x,y) + \pi \i P A_u(-\delta \omega,x,y) \\
&+\pi \frac{A_u(-\delta k_x,x,y)}{K(-\delta k_x)} \cdot \frac{w_0}{(2 \pi)^2} - \pi \i \sum_{m=-\infty}^{\infty} D^{(2,4)}(\aModeM) A_{u}^r(\aModeM,x,y).
\end{align*}
The final integral is
\begin{align*}
\phi_{d,d}(x,y)&=\frac{1}{2}\int_{-\infty}^{\infty} D^{(2,4)}(\gamma)  \e^{-\i \gamma \phi} A_d(\gamma,x,y) \d \gamma 
\\&= 2 \pi \int_{-\infty}^{\infty} \left\{D^{(2,4)}(\gamma)\e^{-2 \i \gamma} J_-(\gamma) \right\} J_+(\gamma) \frac{ -\e^{\i (2-x) \gamma} \cos\left(\zeta (y - s) \right)}{\zeta \sin(s \zeta)}\d \gamma .
\end{align*}
Using similar arguments to the previous integrals, we close the above integral in $\UHP$  if $x<2$ and in $\LHP$ if $x>2$. Consequently, for $x<2$ we obtain
\begin{align*}
\phi_{d,d}(\phi,\psi) &=\pi \i \sum_{n=0}^\infty B_n A_d(\theta_n^+) ,
\end{align*}
and for $x>2$ we obtain
\begin{align*}
\phi_{d,d}(x,y)&= \pi \sum_{n=0}^{\infty} \frac{\Asig+\Csig}{\dModeM+\delta \omega} A_d(\dModeM,x,y) + \pi \i P A_d(-\delta \omega,x,y) \\
&+\pi \frac{A_d(-\delta k_x,x,y)}{K(-\delta k_x)} \cdot \frac{w_0}{(2 \pi)^2} - \pi \i \sum_{m=-\infty}^{\infty} D^{(2,4)}(\aModeM) A_{d}^r(\aModeM,x,y).
\end{align*}
Summing the contributions from each integral yields the full Fourier inversion in section \ref{c5:Sec:FourInversion}
\section{Factorisation of Kernel Function} \label{Ap:splitting}

The kernel function is defined as \eqref{c5:Eq:kDef}
\begin{align}
K(\gamma) &= \frac{\zeta \sin(s \zeta)}{4 \pi\left(\cos(s \zeta) - \cos( d \gamma + \sigma^\prime) \right)} +\frac{1}{4 \pi} \left( \mu_0 - \i \mu_1 \gamma - \mu_2 \gamma^2 \right). \label{Eq:apKdef} %
\end{align}
We seek a multiplicative factorisation of this function into parts that are analytic in $\ULHP$ respectively. Consequently, we restrict our attention to the upper half plane and a corresponding factorisation for the lower half plane can be constructed by an appropriate symmetry argument.

\subsection{Factorisation of Poles of $K$}
Previous work \citep{Peake1992,Glegg1999} has factorised the poles of $K$ into the form:
\begin{align*}
\cos(s \zeta) - \cos( d \gamma + \sigma^\prime) &= E(\gamma) \prod_{m = -\infty}^\infty \left( 1 - \frac{\gamma}{\lambda_m^+} \right) \prod_{m = -\infty}^\infty \left( 1 - \frac{\gamma}{\lambda_m^-} \right),
\end{align*}
where $E(\gamma)$ is an entire function that contains no zeros, and
\begin{align*}
\lambda_m^\pm &= - f_m \sin(\chi) \pm \cos(\chi) \zeta(f_m), & f_m = \frac{\sigma^\prime - 2 \pi m}{\Delta}.
\end{align*}
\subsubsection{Asymptotic Behaviour of Poles}
We note that the asymptotic behaviour of these poles is
\begin{align}
\lambda_m^+ &\sim \asymPorousAcouP{0}{R} m + \asymPorousAcouP{2}{R} + o (1), \label{PEq:asymAcouExpI} \\ 
\lambda_{-m}^+ &\sim \asymPorousAcouP{0}{L} m + \asymPorousAcouP{2}{L} + o (1), \label{PEq:asymAcouExpII}
\end{align}
as $n \rightarrow \infty$ where
\begin{align*}
\asymPorousAcouP{0}{R} &=  2 \pi \frac{d + \i s}{\Delta^2},&
\asymPorousAcouP{2}{R} &= -\frac{d + \i s }{\Delta^2} \sigma^\prime ,\\
\asymPorousAcouP{0}{L} &= 2 \pi \frac{-d + \i s}{\Delta^2},&
\asymPorousAcouP{2}{L} &= \frac{-d + \i s }{\Delta^2} \sigma^\prime.
\end{align*}
The subscripts $R$ and $L$ indicate that the pole is in the right or left hand side of $\UHP$ respectively.
\subsection{Factorisation of Zeros of $K$} \label{c5:Sec:zerosFac}
We now outline the procedure for factorising the zeros of $K$. In contrast to previous analyses for rigid plates \citep{Peake1992,Glegg1999}, no analytic factorisation is available. A numerical root finding algorithm is sufficient to find the locations of these roots, but we also require some knowledge about their asymptotic behaviour. The reason for this is that during the Wiener-Hopf method we must know the asymptotic behaviour of the factorised kernel function. This asymptotic behaviour is inextricably linked to the asymptotic behaviour of the kernel's zeros and poles.

We focus on the zeros located in the first quadrant of the complex $\gamma-$plane. These roots are labelled as $\porousRootsI{n}$. The asymptotic behaviour of the roots in the other quadrants can be determined by a similar procedure.

Recall the definition of the branch cut of $\zeta$:
\begin{align*}
\zeta = \sqrt{k^2 w^2 - \gamma^2} = \e^{i \psi_1/2} \e^{i \psi_2/2} \left| k^2 w^2 - \gamma^2 \right|^{1/2},
\end{align*}
where 
\begin{align*}
\psi_1 = \arg (k w - \gamma), \qquad \psi_2 = \arg (k w + \gamma),
\end{align*}
and 
\begin{align*}
\pi /2 < \psi_1 < 5\pi/2, \qquad -\pi /2 < \psi_2 < 3\pi/2.
\end{align*}
Since $\porousRootsI{n}$ are in the first quadrant, we have
\begin{align*}
\zeta(\porousRootsI{n}) \sim  \i \porousRootsI{n}, \qquad \textnormal{as } n \rightarrow \infty.
\end{align*}
This leads us to determine the following asymptotic behaviours:
\begin{align*}
\sin(s \zeta(\porousRootsI{n})) \sim -\frac{1}{2 \i} \e^{ s\porousRootsI{n}},
\end{align*}
\begin{align*}
\cos(s \zeta(\porousRootsI{n})) \sim \frac{1}{2} \e^{ s\porousRootsI{n}}, \qquad \cos(d \porousRootsI{n} + \sigma^\prime) \sim \frac{1}{2} \e^{-\i (d \porousRootsI{n} + \sigma^\prime)}.
\end{align*}
We now substitute these representations into \eqref{Eq:apKdef} to obtain asymptotic expansions for the roots $\porousIRootsI{n}$. Each case must be considered separately, although the asymptotic behaviours are similar.\\

\subsubsection*{Case I Boundary Condition}
For the no-mean-flow boundary condition, the asymptotic behaviour of the roots obeys
\iffalse
\begin{align*}
\zeta(\porousRootsI{n}) \sin(s \zeta(\porousRootsI{n})) + \mu_0 (\cos(s \zeta(\porousRootsI{n})) - \cos( d \porousRootsI{n} + \sigma^\prime))\\
\sim  -\frac{\porousRootsI{n}}{2} \e^{ s\porousRootsI{n}} + \mu_0 \left( \frac{1}{2} \e^{ s\porousRootsI{n}} - \frac{1}{2} e^{-\i( d \porousRootsI{n}+ \sigma^\prime)}\right)
\end{align*}
\fi
\begin{align}
\porousRootsI{n} \sim \mu_0 \left( 1 - \exp \left[-\i((d-\i s)\porousRootsI{n} + \sigma^\prime) \right] \right). \label{PEq:rootAsymp}
\end{align}
We seek an asymptotic expansion of the first quadrant roots of the classical form
\begin{align}
\porousRootsI{n} & \sim \sum_{m = 0}^\infty \asymPorousRootsCoefsI{k} \asymPorousRootsFunI{k}{n},  \label{PEq:asymExpans}
\end{align}
where $\asymPorousRootsFunI{k+1}{n} = o\left( \asymPorousRootsFunI{k}{n} \right)$ as $n \rightarrow \infty$.
In \eqref{PEq:rootAsymp} we require the linear and exponential terms to match. However, since $\left| \porousRootsI{n}\right| \rightarrow \infty$ as $n \rightarrow \infty$, the exponential term will grow at a faster rate than the linear term. Consequelty, the real part of the argument of the exponential must asymptotically small compared to the imaginary part. We therefore expand the roots into real and imaginary parts as $\porousRootsI{n} = \porousRRootsI{n} + \i \porousIRootsI{n}$ and write
\begin{align}
d \porousIRootsI{n}  - s \porousRRootsI{n} = G(n), \label{PEq:asym1}
\end{align}
where $G(n) = o(\asymPorousRootsFunI{0}{n})$. Rearranging yields
\begin{align*}
\porousRootsI{n}&=\porousRootsI{n} + \i \porousIRootsI{n}= \left(1+ \i\frac{s}{d} \right) \porousRRootsI{n}  -\i \frac{G(n)}{d}.
\end{align*}
Since the arguments of the left and right hand sides of equation \eqref{PEq:rootAsymp} must match, we obtain an expression for the imaginary part of the argument of the exponential:
\begin{align}
\sigma^\prime + d \porousRRootsI{n} + s \porousIRootsI{n} &= \pi - \arctan\left(\frac{s}{d}\right) - 2 n \pi + o\left(1\right) .  \label{PEq:asym2}
\end{align}
Applying the asymptotic expansion \eqref{PEq:asymExpans} and taking the leading order terms of \eqref{PEq:asym1} and \eqref{PEq:asym2} yields
\begin{align}
\begin{split}
\asymPorousRootsFunI{0}{n} &= n ,\\
\asymPorousRootsCoefsI{0} &= \frac{2 \pi (d + \i s)}{\Delta^2} .
\end{split}
\label{PEq:coefAsymp0}
\end{align}
We may now substitute the expansion for $\porousRootsI{n}$ so far into \eqref{PEq:rootAsymp} to obtain
\begin{align*}
\frac{2 \pi (d + \i s)}{\Delta^2} n + o(n) \sim {\mu_0} \left( 1 + \exp \left[ G(n) + \i  \arctan\left(\frac{s}{d}\right) \right] \right).
\end{align*}
We now match leading order terms to obtain
\begin{align*}
\frac{2 \pi }{\Delta} n = \mu_0 \exp \left[ G(n) \right] ,
\end{align*}
so that
\begin{align*}
G(n) &= \log \left(n\right) + \log \left(\frac{2 \pi }{ \mu_0 \Delta} \right).
\end{align*}
Similarly, taking the leading order terms in \eqref{PEq:asym1} and \eqref{PEq:asym2} yields
\begin{align}
\begin{split}
\asymPorousRootsFunI{1}{n} &= \log(n) , \\
\asymPorousRootsCoefsI{1} &= \i \frac{d+ \i s}{\Delta^2} .
\end{split}
\label{PEq:coefAsymp1}
\end{align}
and
\begin{align}
\begin{split}
\asymPorousRootsFunI{2}{n} &= 1 , \\
\asymPorousRootsCoefsI{2} &= \frac{d+ \i s}{\Delta^2} \cdot \left(\pi - \sigma^\prime + \i \log \left(\frac{2 \pi }{\mu_0 \Delta}\right) - \arctan\left(\frac{s}{d}\right)\right) .
\end{split}
\label{PEq:coefAsymp2}
\end{align}
Substitution of (\ref{PEq:coefAsymp0}, \ref{PEq:coefAsymp1}, \ref{PEq:coefAsymp2}) into the asymptotic expansion \eqref{PEq:asymExpans} yields 
\begin{align}
\porousRootsI{n} & \sim  \frac{d+ \i s}{\Delta^2} \left( 2 \pi n + \i \log(n)  + \pi - \sigma^\prime + \i \log \left(\frac{2 \pi }{\mu_0 \sqrt{\Delta^2}}\right) - \arctan\left(\frac{s}{d}\right) + o(1)  \right). \label{PEq:asymExpansFinI}
\end{align}
Similar analysis yields the asymptotic behaviour for the roots in the second quadrant as
\begin{align}
\porousRootsII{n} & \sim  \frac{-d+ \i s}{\Delta^2} \left( 2 \pi n -\i \log(n)  - \pi + \sigma^\prime - \i \log \left(\frac{2 \pi }{\mu_0 \Delta}\right) - \arctan\left(\frac{s}{d}\right) + o(1)  \right). \label{PEq:asymExpansFinII}\\
\end{align}

 \subsubsection*{Case II Boundary Condition}
For the Darcy-type boundary condition, the asymptotic behaviour of the roots obeys the equation
\iffalse
\begin{align*}
\zeta(\porousRootsI{n}) \sin(s \zeta(\porousRootsI{n})) + \mu (\cos(s \zeta(\porousRootsI{n})) - \cos( d \porousRootsI{n} + \sigma^\prime))\\
\sim  -\frac{\porousRootsI{n}}{2} \e^{ s\porousRootsI{n}} + \mu \left( \frac{1}{2} \e^{ s\porousRootsI{n}} - \frac{1}{2} e^{-\i( d \porousRootsI{n}+ \sigma^\prime)}\right)
\end{align*}
\fi
\begin{align}
1 \sim-\i {\mu_1} \left(1 - \exp \left[-\i((d-\i s)\porousRootsI{n} + \sigma^\prime) \right] \right). \label{PEq:rootAsympDarcy}
\end{align}
We assume an asymptotic expansion of the roots of the form \eqref{PEq:asymExpans}.
Similar reasoning to the previous sections yields that the leading order terms are also given by
\begin{align}
\asymPorousRootsCoefsI{0} = \frac{2 \pi (d + \i s)}{\Delta^2}, \qquad\qquad \asymPorousRootsCoefsII{0} = \frac{2 \pi (-d + \i s)}{\Delta^2} .
\label{PEq:coefAsymp0Darcy}
\end{align}
We may now solve \eqref{PEq:rootAsympDarcy} directly to find the coefficients of the next two orders of the asymptotic expansion as
\begin{align*}
\asymPorousRootsCoefsI{1} &= 0,& \asymPorousRootsCoefsII{1} &= 0,\\
\asymPorousRootsCoefsI{2} &=  \frac{d + \i s}{\Delta^2} \left( \i \log\left(1 +\frac{1 }{\i\mu_0}\right) - \sigma^\prime \right), & \asymPorousRootsCoefsII{2} &= \frac{-d + \i s}{\Delta^2} \left( -\i \log\left(1 -\frac{1 }{\i\mu_0}\right) + \sigma^\prime \right).\\
\end{align*}
\iffalse
In figure \ref{Fig:darcyRoots} we plot the error of our asymptotic expansion compared to the actual roots found with a numerical algorithm.
\begin{figure}[!h]
\centering
\includegraphics[width = .5\linewidth]{example-image}
\end{figure}
\fi
\subsubsection*{Case III boundary condition}
For the case III boundary condition, the asymptotic behaviour of the roots obeys
\iffalse
\begin{align*}
\zeta(\porousRootsI{n}) \sin(s \zeta(\porousRootsI{n})) + \mu (\cos(s \zeta(\porousRootsI{n})) - \cos( d \porousRootsI{n} + \sigma^\prime))\\
\sim  -\frac{\porousRootsI{n}}{2} \e^{ s\porousRootsI{n}} + \mu \left( \frac{1}{2} \e^{ s\porousRootsI{n}} - \frac{1}{2} e^{-\i( d \porousRootsI{n}+ \sigma^\prime)}\right)
\end{align*}
\fi
\begin{align}
1 \sim - {\mu_2} {\porousRootsI{n}}  \left( 1 - \exp \left[-\i((d-\i s)\porousRootsI{n} + \sigma^\prime) \right] \right). \label{PEq:rootAsympImpedance3}
\end{align}
Similar analysis to the previous sections possesses an identical asymptotic expansion (up to the terms considered) and we have, at leading order,
\begin{align}
\asymPorousRootsCoefsI{0} = \frac{2 \pi (d + \i s)}{\Delta^2}, \qquad\qquad \asymPorousRootsCoefsII{0} = \frac{2 \pi (-d + \i s)}{\Delta^2} .
\label{PEq:coefAsymp0Darcy2}
\end{align}
Substitution of \eqref{PEq:coefAsymp0Darcy2} into \eqref{PEq:rootAsympImpedance3}
yields the coefficients of the next two orders of the asymptotic expansion as
\begin{align*}
\asymPorousRootsCoefsI{1} &= 0, & \asymPorousRootsCoefsII{1} &= 0,\\
\asymPorousRootsCoefsI{2} &=  \frac{\sigma^\prime (-d - \i s)}{\Delta^2},  & \asymPorousRootsCoefsII{2} &=  \frac{\sigma^\prime(-d + \i s)}{\Delta^2}  .
\end{align*}

\subsection{Full Factorisation of Kernel, $K$}
We propose a multiplicative splitting of $K$ of the form:
\begin{align*}
K(\gamma) &= K_+(\gamma) K_-(\gamma),
\end{align*}
where
\begin{align}
\addtocounter{equation}{1}
K_-(\gamma) & = \e^{E(\gamma)} \frac{\prod_{m=1}^{\infty} \left(1 - \gamma / \porousRootsI{m} \right)\left(1 - \gamma / \porousRootsII{m}\right)}{\prod_{m=-\infty}^{\infty} \left(1 - \gamma / \lambda_{m}^+ \right)}, \label{Eq:Kfac} \tag{\theequation .a}\\
K_+(\gamma) & = \e^{-E(\gamma)} \frac{\prod_{m=1}^{\infty} \left(1 - \gamma / \porousRootsIM{m} \right)\left(1 - \gamma / \porousRootsIIM{m}\right)}{\prod_{m=-\infty}^{\infty} \left(1 - \gamma / \lambda_{m}^- \right)}K(0). \tag{\theequation .b}
\end{align}
The entire function $E$ is included to ensure that $K_\pm$ has algebraic growth in $\ULHP$ respectively. Previous works for rigid blades \citep{Peake1992,Glegg1999}  have derived the form of the entire function $E$ and shown that it is a polynomial. In the remainder of this section, we will show that $E$ is in fact a constant for the present problem.

\subsubsection*{Asymptotic Behaviour of Proposed Factorisation}
We focus on $K_-$ since the asymptotic behaviour of $K_+$ follows in an analogous manner. We first establish the existence and growth of some relevant products.
\begin{prop}
	The infinite product
	\begin{align}
	P_1 = \prod_{m=1}^\infty \left(\frac{\asymPorousRootsCoefsI{0} m + \asymPorousRootsCoefsI{2}}{\porousRootsI{m}}\right) \cdot
	\left(\frac{\asymPorousRootsCoefsII{0} m + \asymPorousRootsCoefsII{2}}{\porousRootsII{m}}\right) 
	\label{PEq:product}
	\end{align}
	exists.
	\label{Pprop}
\end{prop}
\begin{proof}
	We use the asymptotic expansions \eqref{PEq:asymExpansFinI} and \eqref{PEq:asymExpansFinII} to obtain
	\begin{align*}
	\frac{\asymPorousRootsCoefsI{0} m + \asymPorousRootsCoefsI{1}}{\porousRootsI{m}} &\sim 1 - \frac{ \asymPorousRootsCoefsI{1} \log(m) }{\asymPorousRootsCoefsI{0} m} + o(m^{-1}), \\
	\frac{\asymPorousRootsCoefsII{0} m + \asymPorousRootsCoefsII{1}} {\porousRootsII{m}}&\sim 1 - \frac{ \asymPorousRootsCoefsII{1} \log(m) }{\asymPorousRootsCoefsII{0} m} + o(m^{-1}),
	\end{align*}
	respectively. Substitution into the product \eqref{PEq:product} yields
	\begin{align}
P_1 = 	\prod_{m=1}^\infty \left(1 - A \frac{\log (m)}{m} + o(m^{-1}) \right), 
	\label{PEq:product2}
	\end{align}
	where
	\begin{align*}
	A &=  \frac{ \asymPorousRootsCoefsI{1}  }{\asymPorousRootsCoefsI{0}} + \frac{ \asymPorousRootsCoefsII{1} }{\asymPorousRootsCoefsII{0}},  %
	\end{align*}
	and application of \eqref{PEq:coefAsymp0} and \eqref{PEq:coefAsymp1} shows that $A = 0$. Consequently, the products \eqref{PEq:product2} and therefore \eqref{PEq:product} exist via the comparison method.
\end{proof}
\begin{prop}
	The infinite product
	\begin{align}
	P_{2} (\gamma) &= \prod_{m=1}^{\infty} \left(\frac{\porousRootsI{m} - \gamma }{\asymPorousRootsCoefsI{0} m + \asymPorousRootsCoefsI{2} - \gamma } \right) \cdot \left(\frac{\porousRootsII{m} - \gamma }{\asymPorousRootsCoefsII{0} m + \asymPorousRootsCoefsII{2} - \gamma } \right) .
	\label{PEq:product1a}
	\end{align}
	has at worst algebraic growth as $\gamma \rightarrow \infty$ in $\LHP$.
	\label{Pprop2}
\end{prop}
\begin{proof}
Consider the log-derivative of $P_2(\gamma)$:
\begin{align}
L(\gamma)&= \frac{\d \log(P_2(\gamma))}{\d \gamma}\notag\\
&= \sum_{m=1}^\infty \left(\frac{\porousRootsI{m} - (\asymPorousRootsCoefsI{0}m + \asymPorousRootsCoefsI{2})}{( \asymPorousRootsCoefsI{0} m +\asymPorousRootsCoefsI{2}- \gamma )(\porousRootsI{m}- \gamma )} + \frac{\porousRootsII{m} - (\asymPorousRootsCoefsII{0}m + \asymPorousRootsCoefsII{2})}{( \asymPorousRootsCoefsII{0} m +\asymPorousRootsCoefsII{2}- \gamma )(\porousRootsII{m}- \gamma )}  \right) \notag \\
&= \sum_{m=1}^\infty f_m(\gamma). \label{Eq:L2def}
\end{align}
Recall \eqref{PEq:asymExpansFinI} and $\Im[\gamma]<0$ and $\Im[\asymPorousRootsCoefsI{0},\, \asymPorousRootsCoefsII{0}]>0$. Therefore, $\exists M_1 > 0 $ s.t. $\forall m_1 > M_1$,
\begin{align*}
\Im[\porousRootsI{m_1} - \gamma ], \Im[\asymPorousRootsCoefsI{0} m_1 + \asymPorousRootsCoefsI{2} - \gamma ] &> m_1\Im [ \asymPorousRootsCoefsI{0}], \\
\Im[\porousRootsII{m_1} - \gamma ], \Im[\asymPorousRootsCoefsII{0} m_1 + \asymPorousRootsCoefsII{2} - \gamma ] &> m_1\Im [ \asymPorousRootsCoefsII{0}].
\end{align*}
Furthermore, \eqref{PEq:asymExpansFinI} shows that $\exists M_2>0$ s.t. $\forall m_2>M_2$,
 \begin{align*}
 \left|\porousRootsI{m_2} - (\asymPorousRootsCoefsI{0}m_2 + \asymPorousRootsCoefsI{2})\right| < C_R \log(m_2),\\
  \left|\porousRootsII{m_2} - (\asymPorousRootsCoefsII{0}m_2 + \asymPorousRootsCoefsII{2})\right| < C_L \log(m_2),
 \end{align*}
 for some $C_L$ and $C_R$.
 Therefore, we have a uniform bound on $f_m (\gamma)$ for \mbox{$m>M = \max(M_1,M_2)$}:
 \begin{align*}
\left| f_m (\gamma) \right| &< C \frac{\log(m)}{m^2}, %
 \end{align*}
 for some $C$. Consequently, the series $L(\gamma)$ converges uniformly in $\LHP$ and we may write
\begin{align*}
\lim_{\substack{\gamma \rightarrow \infty \\\gamma \in \LHP}} L(\gamma) = \lim_{\substack{\gamma \rightarrow \infty \\\gamma \in \LHP}} \left(\sum_{m=1}^\infty f_m(\gamma) \right) 
= \sum_{m=1}^\infty\lim_{\substack{\gamma \rightarrow \infty \\\gamma \in \LHP}} \left( f_m(\gamma) \right) 
=0.
\end{align*}
By reverting to the initial definition of $L$ \eqref{Eq:L2def}, we obtain the limit
\begin{align*}
\lim_{\substack{\gamma \rightarrow \infty \\\gamma \in \LHP}} \frac{P_2^\prime(\gamma)}{P_2(\gamma)} = 0.
\end{align*}
Accordingly, we may write $P_2(\gamma) \sim A \gamma^{\alpha}$ for some constants $A$ and $\alpha$ so that $P_2$ has algebraic growth in $\LHP$.
\end{proof}
We now calculate the asymptotic behaviour of $K_-$ by comparison with the product
\begin{align}
P(\gamma) &= \prod_{m=1}^{\infty} \dfrac{\left(1 - \gamma / \porousRootsI{m} \right)\e^{\gamma/\porousRootsI{m}} \left(1 - \gamma / \porousRootsII{m}\right) \e^{\gamma / \porousRootsII{m}}}{\left(1 - \gamma / (\asymPorousRootsCoefsI{0} m + \asymPorousRootsCoefsI{2}) \right) \e^{\gamma/\asymPorousRootsCoefsI{0}{m}} \left(1 - \gamma / (\asymPorousRootsCoefsII{0} m + \asymPorousRootsCoefsII{2})\right) \e^{\gamma/\asymPorousRootsCoefsII{0}{m}}}.
\label{PEq:pFunc}
\end{align}
This function may be written as
\begin{align*}
P(\gamma) =& \exp \left(\gamma \sum_{m=1}^{\infty} \frac{1}{\porousRootsI{m}} - \frac{1}{ \asymPorousRootsCoefsI{0} m} - \frac{1}{\porousRootsII{m}} - \frac{1}{ \asymPorousRootsCoefsII{0} m} \right) P_1 P_2(\gamma).
\end{align*}
The product $P_1$ is known to exist by proposition \ref{Pprop}. Applying proposition \ref{Pprop2} yields
\begin{align*}
P(\gamma) \sim& B \gamma^{\alpha} \exp \left(\gamma \sum_{m=1}^{\infty} \frac{1}{\porousRootsI{m}} - \frac{1}{ \asymPorousRootsCoefsI{0} m} + \frac{1}{\porousRootsII{m}} - \frac{1}{ \asymPorousRootsCoefsII{0} m} \right), 
\end{align*}
for some constant B. By applying the asymptotic behaviour of the Gamma function \cite[B7 \& B8]{Peake1992}, we derive the relation
\begin{align*}
\prod_{n=1}^{\infty} \left(1 - \frac{\gamma}{a m + b}\right) \exp \left[ \frac{\gamma}{a m}\right] & \sim C \exp \left[\frac{\gamma}{a} (\mathcal{E} - 1 - \log(-a)) + \left(\frac{\gamma}{a} - \frac{b}{a} - \frac{1}{2}\right) \log(\gamma)\right],
\end{align*}
where $\mathcal{E}$ is the Euler–Mascheroni constant and 
\begin{align*}
C & = \frac{-b}{\sqrt{2 \pi}} \Gamma \left(\dfrac{b}{a} \right) \left(-a\right)^{\frac{b}{a} - \frac{1}{2}}.
\end{align*}
This representation may be substituted into the denominator of \eqref{PEq:pFunc}. Rearranging yields the asymptotic behaviour
\begin{align*}
\prod_{m=1}^{\infty} & {\left(1 - \gamma / \porousRootsI{m} \right)\e^{\gamma/\porousRootsI{m}} \left(1 - \gamma / \porousRootsII{m}\right) \e^{\gamma / \porousRootsII{m}}} \sim B_2 \gamma^\alpha  \exp \left(\gamma \sum_{m=1}^{\infty} \frac{1}{\porousRootsI{m}} + \frac{1}{\porousRootsII{m}} \right) \\
&\times \exp \left[\frac{\gamma}{\asymPorousRootsCoefsI{0}\asymPorousRootsCoefsII{0}} \left((\mathcal{E} - 1)\left(\asymPorousRootsCoefsI{0}+\asymPorousRootsCoefsII{0}\right) - \asymPorousRootsCoefsII{0} \log(\asymPorousRootsCoefsI{0}) + \asymPorousRootsCoefsI{0}\log(\asymPorousRootsCoefsII{0}) \right) + \right. \\
& \left. \left(\gamma \left( \frac{1}{\asymPorousRootsCoefsI{0}} + \frac{1}{\asymPorousRootsCoefsII{0}} \right)   - \left( \frac{\asymPorousRootsCoefsI{2}}{\asymPorousRootsCoefsI{0}} + \frac{\asymPorousRootsCoefsII{2}}{\asymPorousRootsCoefsII{0}} \right) -1\right) \log(\gamma)\right],
\end{align*}
for some constant $B_2$. We also note that
\begin{align*}
\prod_{m=-\infty}^{\infty} &\left(1 - \gamma / \lambda_{m}^+ \right)\e^{\gamma/\lambda_m^+} \sim B_3 \gamma   \exp \left(\gamma \sum_{m=1}^{\infty} \frac{1}{\asymPorousAcouP{0}{m}} + \frac{1}{\asymPorousAcouP{0}{-m}} \right) \\
&\times \exp \left[\frac{\gamma}{\asymPorousAcouP{0}{m}\asymPorousAcouP{0}{-m}} \left((\mathcal{E} - 1)\left(\asymPorousAcouP{0}{m}+\asymPorousAcouP{0}{-m}\right) - \asymPorousAcouP{0}{-m} \log(\asymPorousAcouP{0}{m}) + \asymPorousAcouP{0}{m}\log(\asymPorousAcouP{0}{-m}) \right) \right. \\
& +\left. \left(\gamma \left( \frac{1}{\asymPorousAcouP{0}{m}} + \frac{1}{\asymPorousAcouP{0}{-m}} \right)   - \left( \frac{\asymPorousAcouP{2}{m}}{\asymPorousAcouP{0}{m}} + \frac{\asymPorousAcouP{2}{-m}}{\asymPorousAcouP{0}{-m}} \right) -1\right) \log(\gamma)\right],
\end{align*}
for some constant $B_3$. Noting that $\asymPorousAcouP{0}{m} = \asymPorousRootsCoefsI{m}$ and $\asymPorousAcouP{0}{-m} = \asymPorousRootsCoefsII{m}$, we obtain
\begin{align*}
\frac{\prod_{m=1}^{\infty} \left(1 - \gamma / \porousRootsI{m} \right)\left(1 - \gamma / \porousRootsII{m}\right)}{\prod_{m=-\infty}^{\infty} \left(1 - \gamma / \lambda_{m}^+ \right)}\sim
&B_4 \gamma^{\alpha-1} \exp \left[ - \left( \frac{\asymPorousRootsCoefsI{2}}{\asymPorousRootsCoefsI{0}} + \frac{\asymPorousRootsCoefsII{2}}{\asymPorousRootsCoefsII{0}} \right) \log(\gamma)\right],
\end{align*}
for some constant $B_4$. Consequently, the entire function $E$ in \eqref{Eq:Kfac} is a constant.

\end{appendices}

\bibliographystyle{class/jfm}
\bibliography{utility/library}

\end{document}